\documentclass[a4paper,11pt]{article}
\usepackage{amssymb}
\usepackage[latin1]{inputenc}
\usepackage[british]{babel}
\usepackage{natbib}
\usepackage{graphicx}
\usepackage[pdftex]{lscape}
\usepackage{verbatim}
\usepackage{amsmath}
\usepackage{mathtools}
\usepackage{arydshln}
\usepackage{amsthm}
\usepackage{amssymb}
\usepackage{amsfonts}
\usepackage{setspace}
\usepackage{latexsym}
\usepackage{a4}
\usepackage{bbm}
\usepackage{mathrsfs}
\usepackage{fixmath}
\usepackage{calc}
\usepackage{accents}
\usepackage{tikz}
\usepackage{subfigure}
\usepackage{float}
\usepackage{bm}
\usepackage{dsfont}
\usepackage{circledsteps}

\usepgflibrary{shapes}

\date{}

\DeclareMathAlphabet\mathbfcal{OMS}{cmsy}{b}{n}
\newcommand{\ud}{\mathrm{d}}
\newcommand{\rk}{\mathrm{rank}}
\newcommand{\ve}{\mathrm{vec}\,}
\newcommand{\Dnt}{\tilde{D}_n}
\newcommand{\nt}{\tilde{n}}

\newcommand{\Cholpi}{\Sigma_{p,tr}^{-1}}
\newcommand{\Cholpit}{\Sigma_{p,tr}^{-1\prime}}

\newcommand{\Atf}{A_{10}^\prime,A_{1+}^\prime}
\newcommand{\Ats}{A_{s0}^\prime,A_{s+}^\prime}
\newcommand{\Ap}{A_{p0}}
\newcommand{\Apt}{A_{p0}^\prime}
\newcommand{\Api}{A_{p0}^{-1}}
\newcommand{\Apit}{A_{p0}^{-1\prime}}
\newcommand{\Apd}{A_{p+}}
\newcommand{\Apdt}{A_{p+}^{\prime}}
\newcommand{\Qp}{Q_p}
\newcommand{\Qpt}{Q_p^{\prime}}

\newcommand{\e}{\varepsilon}

\newcommand{\Rj}{\textbf{R}_{j}}
\newcommand{\Lj}{\textbf{L}_{j}}
\newcommand{\Sj}{\textbf{S}_{j}}
\newcommand{\Atot}{\left(\mathcal{A}_0,\mathcal{A}_+\right)}
\newcommand{\Btot}{\left(\mathcal{B},\mathbold{\Sigma}\right)}
\newcommand{\Gp}{G(A_{p0}^\prime,A_{p+}^\prime)}

\newcommand{\Vj}{\textbf{V}_j}
\newcommand{\Vtj}{\textbf{V}_j(\theta)}
\newcommand{\tVtj}{\bm{\widetilde V}_j(\theta)}
\newcommand{\tVt}{\tilde{\textbf{V}}(\theta)}
\newcommand{\tVtz}{\accentset{\:\,\sim}{\textbf{V}}(\theta)}
\newcommand{\Ttt}{\accentset{\approx}{T}}
\newcommand{\Vtt}{\accentset{\:\approx}{\textbf{V}}(\theta)}
\newcommand{\Vttz}{\accentset{\:\,\approx}{\textbf{V}}(\theta_0)}
\newcommand{\IS}{IS_\eta(\phi)}
\newcommand{\ISp}{IS_{\eta,p}(\phi)}
\newcommand{\Ar}{\ensuremath\mathcal{A}(\phi)}
\newcommand{\Ons}{\ensuremath\mathcal{O}\left(ns\right)}
\newcommand{\Qr}{\ensuremath\mathcal{Q}(\phi)}
\newcommand{\Qrp}{\ensuremath\mathcal{Q}_p(\phi)}

\newcommand{\Qt}{\tilde{\ensuremath\mathcal{Q}}}

\newcommand{\AllReg}{\{1,\ldots,s\}}
\ifdefined\O
    \renewcommand{\O}{\ensuremath\mathcal{O}\left(n\right)}
\else
    \newcommand{\O}{\ensuremath\mathcal{O}\left(n\right)}
\fi
\ifdefined\Re
    \renewcommand{\Re}{\ensuremath\mathbb{R}}
\else
    \newcommand{\Re}{\ensuremath\mathbb{R}}
\fi
\ifdefined\F
    \renewcommand{\F}{\textbf{F}}
\else
    \newcommand{\F}{\textbf{F}}
\fi
\ifdefined\S
    \renewcommand{\S}{\textbf{S}}
\else
    \newcommand{\S}{\textbf{S}}
\fi

\def\thesection{\Roman{section}}
\makeatletter

\makeatletter
\newcommand*{\encircled}[1]{\relax\ifmmode\mathpalette\@encircled@math{#1}\else\@encircled{#1}\fi}
\newcommand*{\@encircled@math}[2]{\@encircled{$\m@th#1#2$}}
\newcommand*{\@encircled}[1]{%
  \tikz[baseline,anchor=base]{\node[draw,circle,outer sep=0pt,inner sep=.2ex] {#1};}}
\makeatother

\newtheorem{theo}{Theorem}%[section]
\newtheorem{corol}{Corollary}%[section]
\newtheorem{lemma}{Lemma}%[section]
\newtheorem{algo}{Algorithm}%[section]
\theoremstyle{definition}
\newtheorem{defin}{Definition}
\newtheorem{condition}{Condition}

\newtheorem{remark}{Remark}

\begin{document}

\title{SVARs with breaks: Identification and inference{\Large \thanks{We would like to thank Majid Al Sadoon, Giuseppe Cavaliere, Marco Del Negro, Thorsten Drautzburg, Luca Fanelli, Elisabetta Mirto, Luca Neri, Michele Piffer, Matthew Read, Anna Simoni, Alessio Volpicella, for beneficial discussions and comments. We have also benefited from presenting this work at the 7th RCEA time series workshop (June 2021), 3rd IWEEE (January 2022), Durham University (February 2022), 10th Italian Conference in Econometrics and Empirical Economics (ICEEE, May 2023), 9th Annual Conference of the International Association for Applied Econometrics (IAAE, June 2023). Emanuele Bacchiocchi gratefully acknowledges financial support from Italian Ministry of University and Research (PRIN 2022, Grant 20229PFAX5) and the University of Bologna (RFO grants).}}}
\author{%
\begin{tabular}{ccc}
Emanuele Bacchiocchi\thanks{University of Bologna, Department of Economics. Email: e.bacchiocchi@unibo.it} &  & 
Toru Kitagawa\thanks{Brown University, Department of Economics. Email: toru\_kitagawa@brown.edu} \\ 
University of Bologna &  & Brown University \\ 
\end{tabular}%
}
\vspace{4cm}

\date{May 2024}
\maketitle

\begin{abstract}
In this paper we propose a class of structural vector autoregressions (SVARs) characterized by structural breaks (SVAR-WB). Together with standard restrictions on the parameters and on functions of them, we also consider constraints across the different regimes. Such constraints can be either (a) in the form of stability restrictions, indicating that not all the parameters or impulse responses are subject to structural changes, or (b) in terms of inequalities regarding particular characteristics of the SVAR-WB across the regimes. We show that all these kinds of restrictions provide benefits in terms of identification. We derive conditions for point and set identification of the structural parameters of the SVAR-WB, mixing equality, sign, rank and stability restrictions, as well as constraints on forecast error variances (FEVs). As point identification, when achieved, holds locally but not globally, there will be a set of isolated structural parameters that are observationally equivalent in the parametric space. In this respect, both common frequentist and Bayesian approaches produce unreliable inference as the former focuses on just one of these observationally equivalent points, while for the latter on a non-vanishing sensitivity to the prior. To overcome these issues, we propose alternative approaches for estimation and inference that account for all admissible observationally equivalent structural parameters. Moreover, we develop a pure Bayesian and a robust Bayesian approach for doing inference in set-identified SVAR-WBs. Both the theory of identification and inference are illustrated through a set of examples and an empirical application on the transmission of US monetary policy over the great inflation and great moderation regimes.
\end{abstract}

\vspace{3cm}

\begin{flushleft}
\end{flushleft}

\newpage

%%%%%%%%%%%%%%%%%%%%%%%%%%%  INTRODUCTION %%%%%%%%%%%%%%%%%%%%%%%%%%%%%%%%%%%%%%%%%%%

\section{Introduction}
\label{sec:intro}

In a recent contribution, \cite{MM14ECTA} indicate two major concerns in the macroeconomic literature: the limited variation in the data leading to weak identification from one side, and the parameter instability characterizing many empirical relations among macroeconomic and financial variables from the other. Widening the horizon to the relationships between economic theory and empirical modeling, one can think of a third further concern regarding the lack of information coming from the economic theory to uniquely identify the parameters of the empirical models. This paper deals with these three concerns by considering a widespread family of dynamic models, structural vector autoregressions (SVARs), allowing for possible changes in the parameters and where such instability can be exploited constructively to solve or reduce the identification issue. 

There are strong empirical evidences of instabilities both in the data and in the relations among macroeconomic variables. Perhaps the most debated one is the sharp decline in the volatility of many macroeconomic variables documented, among others, by \cite{MPQ00AER} and \cite{StockWatson02NBER} and labeled ``Great Moderation''. However, while some authors, such as Stock and Watson (2002, 2003)\nocite{StockWatson02NBER}\nocite{StockWatson03}, \cite{JP08AER} and \cite{SimsZha06AER}, attribute this changes to a decrease in the volatility of the shocks (``good luck''), \cite{BG06RESTATS} and \cite{InoueRossi11RESTATS} provide empirical evidence of a difference in the conduct of the monetary policy during the Great Moderation and of changes in the private sector's parameters. A more recent literature, instead, investigates the disconnect between inflation and the business cycle observed since mid 1990s that could be interpreted as a sort of flattening of the Phillips curve; see, among others, \cite{DLPT20} and \cite{BFV23}.

Apart few exceptions, such instabilities have been detected through SVAR models, that have been playing a key role in the empirical macroeconomic literature for the policy analysis evaluation and for understanding stylized facts. However, since their introduction in econometrics (\citeauthor{Sims80Ecta}, \citeyear{Sims80Ecta}), many researchers faced the identification issue and proposed different techniques, generally based on restricting the parametric space, to solve it. A more recent stream of research, instead, points to a different approach and, instead of imposing restrictions on the parameters, searches for further information observable in the data in order to enrich the information set. \cite{Rigobon03REStat} opens the road to this new literature and proposes to use the heteroskedasticity present in the data and show that the identification conditions are much less stringent than in the traditional simultaneous equation models; \cite{LanneLutkepohl08JMCB} extend this idea to SVAR models. The limit of this literature is that the different heteroskedasticity regimes observed in the data are due, by construction, to the changes in the volatility of the shocks. This is in line with some empirical literature mentioned above, but completely at odds with other findings in which the different paths in the data are to be ascribed to policy regime shifts or structural breaks in behavioral relationships among economic agents. In this direction, \cite{SimsZha06AER}, in a multivariate regime-switching framework, allow both structural parameters and variances of the shocks to change over time. Moreover, in some of their specifications, they consider subsets of parameters to change only, leaving the others to remain constant over time. \cite{BacchiocchiFanelli15} exploit this idea of mixing constant and time-varying parameters to alleviate the identification issue, as advocated by \cite{MM14ECTA}, although in a different setup.

Our paper presents new results in terms of identification, estimation and inference in structural vector autoregression characterized by structural breaks (SVAR-WB henceforth, as in the terminology of \citeauthor{BCF17}, \citeyear{BCF17}). Importantly, we assume the break dates to be exogenously determined. This choice is coherent with many empirical papers in the literature, where the break dates correspond to policy shifts or the occurrence of financial crises: see, among many others, \cite{Rigobon03REStat}, \cite{LanneLutkepohl08JMCB}, \cite{BG06RESTATS}, \cite{ABCF17}. While this choice has no impact on the identification issue, estimation and inference are substantially different when the break dates have to be determined endogenously within the model.

Concerning the identification issue, the first contribution of the present paper refers to the kind of restrictions the researcher can impose to identify the model. Other than presenting restrictions on the parameters of the SVAR-WB, we also allow to constraints on functions of the parameters, like impulse responses. Among this set of constraints, we include our stability restrictions, i.e. when one or a set of parameters or functions of them do not change across the regimes. In this respect, when restrictions combine equality and stability constraints, we present a set of criteria to check the identification of the SVAR-WB. Specifically, we present a necessary and sufficient rank condition that, when met, guarantees the model to be identified at a specific parameter point. All our conditions, however, refer to local identification, which represents the best possibility one can get for SVAR-WBs to benefit of the multiple regimes to have gains in terms of identification (see \citeauthor{BK19}, \citeyear{BK19}).

The second main result is about the estimation of the parameters in case of local identification. In such particular case, in fact, the Gaussian likelihood function is characterized by observationally equivalent multiple peaks. The common frequentist approach simply selects one of them, and behaves as it were the only one, may be inducing to misleading results, as shown in \cite{BK19} for locally-identified SVARs. Our proposal, instead, is to calculate all the locally-identified points in the parametric space and account for all of them in the inferential analysis. Three algorithms are developed to estimate all these parameters in the form of the orthogonal matrices rotating the reduced-form parameters and satisfying the equality and stability restrictions.

The third contribution, strictly related to the previous ones, concerns the possibility to impose sign and rank restrictions, as in the spirit of \cite{Uhlig2005} and \cite{AHD21}, respectively. The researcher, thus, is allowed to combine equality and inequality constraints, where, among the formers, we include the stability restrictions, and among the latters, we also propose inequality restrictions across the regimes, involving parameters, impulse responses and FEVs. In this respect, we will consider SVAR-WBs that are only set, and not point, identified. Moreover, sign restrictions can reduce the number of admissible parameters in locally-identified SVAR-WB, as suggested by \cite{BK19}.

The fourth main contribution of the paper, thus, is to use such exhaustive set of admissible observationally equivalent orthogonal matrices in a Bayesian and in a frequentist-valid inferential framework. Moreover, as the researcher, in general, can not form priors over such observationally equivalent parameters, we extend the multi-prior robust Bayesian approach, originally proposed by \cite{GK18}, to the case of SVAR-WBs. 
%As in \cite{BK19} for locally-identified SVARs, the so called identified set, collecting all the observationally equivalent parameters, will consist of a finite number of points, rather than a dense subset of the parametric space as in \cite{GK18}.

\subsection{Related literature}
\label{sec:lit}

This paper is directly inspired by the ``identification through heteroskedasticity'' approach originally proposed by \cite{SF01JoE}, and subsequently extended to simultaneous equation models by \cite{Rigobon03REStat} and to SVARs by \cite{LanneLutkepohl08JMCB}.\footnote{See also \cite{KleinVella10} and \cite{Lewbel12} who exploit the heteroskedasticity present in the data to have gains on the identification of simultaneous equation models, but in a micro-econometric oriented framework.} In all these approaches, however, the heteroskedasticity present in the data is ascribed to different variances of the structural shocks. Two different papers by \cite{LS22} and \cite{BL23} propose testing for time-varying impulse responses in heteroskedastic SVARs by exploiting external information. The specification we adopt in the present paper, instead, is more similar to the proposal by \cite{BacchiocchiFanelli15}, who focus on SVAR models with possible breaks on the contemporaneous relations among the shocks and provide necessary and sufficient conditions for local identification when some of the parameters remain stable in the different regimes. Our approach, however, is more general than the one presented in \cite{BacchiocchiFanelli15}, as we allow for more general restrictions, involving equality and inequality constraints on the parameters and on functions of them, as in the spirit of \cite{RWZ10RES} (henceforth, RWZ).

As gains in terms of identification can be obtained when some parameters remain stable across the different regimes, our contribution is also related to the literature using parameters stability for identification of structural relationships among the variables as in \cite{MM14ECTA}, where the authors a) use stability restrictions as a source of identification of the stable structural parameters in economic models and b) develop new econometric methods for statistical inference. Related to our paper is also the recent literature where censored variables enter multivariate econometric models: \cite{IW20}, \cite{ILMZ22}, \cite{M21}, \cite{JM21}, \cite{CCMM21}, \cite{AMSV22}. All these models consider the nominal interest rate as the censored variable that could be useful to identify the monetary policy effects on the real economy. Such effects, as in the spirit of our research, could be different whether or not the interest rate has reached the zero lower bound. 

Our approach for studying the identification of the SVAR-WBs, instead, is related both to the classical results on local identification for SVAR models developed by \cite{AGbook}, as well as to the more recent ones on global identification by RWZ. Concerning this latter, although some of our results can be seen as generalization of the RWZ ones, when extending to SVAR-WBs, only local identification can be obtained. 

%Moreover, in developing our identification conditions, we find inspiration from the 
%contributions by \cite{Johansen95} and \cite{Lucchetti06ET}, who study the identification of simultaneous systems of equations and 
%traditional SVARs, respectively, simply focusing on the restrictions imposed and independently of the estimated structural parameters, 
%that are clearly unknown when the analysis of identification is performed.

The estimation of SVAR-WBs \textit{\'a la} \cite{BacchiocchiFanelli15} or heteroskedastic SVARs \textit{\'a la} \cite{LanneLutkepohl08JMCB} has been obtained by generalizing the ML estimator usually used for locally-identified SVARs. However,
specifically for SVARs, \cite{BK19} focus on the possible misleading consequences of multiple peaks present in the likelihood function on estimation and inference of the structural parameters and impulse responses. Estimation and inference concerns on locally-identified models, however, are not confined to the largely diffused frequentist approach. Also in standard Bayesian practice, one faces challenges when the likelihood has multiple modes at least for two reasons. Firstly, the posterior remains sensitive to the choice of prior even asymptotically due to the lack of global identification. Secondly, standard posterior sampling algorithms may fail to explore all the modes, resulting in an inaccurate approximation of the posterior. Our contribution, thus, is related to \cite{BK19} both for the estimation and inference methods. The estimation algorithms presented in the paper, for a given reduced-form parameter point, aim at collecting all the solutions of the identification problem, that will define the identified set to be considered in a multi-prior robust Bayesian approach \textit{\'a la} \cite{GK18}. 

In this respect, our paper is also related to the growing literature on identification through sign restrictions and set-identified SVARs (\citeauthor{Faust98}, \citeyear{Faust98}; \citeauthor{CanovadeNicolo02JME}, \citeyear{CanovadeNicolo02JME}; \citeauthor{Uhlig2005}, \citeyear{Uhlig2005}; \citeauthor{MU09}, \citeyear{MU09}, among others). The identified set for an impulse response in this class of models is a set with a positive measure, while the identified set in our setting can also be a finite number of isolated points, each of which corresponds to one of the solutions of the identification problem (when local identification is met). 

Finally, the empirical application completing the paper relates our research to the wide literature on the ``good luck'' versus ``good policy'' debate about the effectiveness of the US monetary policy in reducing the volatility of many macroeconomic variables since
the beginning of the eighties (see, among many others, \citeauthor{BS09AER}, \citeyear{BS09AER}).

The rest of the paper is organized as follows. In Section \ref{sec:model} we present the SVAR-WB while the identification issue is discussed in Section \ref{sec:identification}. Section \ref{sec:Estim} and Section \ref{sec:Inference} discuss on how to estimate and draw inference on the parameters of the model, respectively. Both local and set identification are considered. In Section \ref{sec:EmpApp} we propose an empirical application based on the transmission of US monetary policy. Section \ref{sec:conclusion} provides some concluding remarks. All the proofs and other theoretical results, as well as a set of examples, are reported in the Appendix.

%%%%%%%%%%%%%%%%%%%%%%%%%%%  ECONOMETRIC MODEL  %%%%%%%%%%%%%%%%%%%%%%%%%%%%%%%%%%%%%%%%%%%

\section{Econometric framework: SVAR with breaks}
\label{sec:model}

Suppose that the sample of interest $1\leq t\leq T$ is characterized by $s$ regimes, or, equivalently, by $s-1$ structural breaks $1<T_{B_1}<\ldots<T_{B_{s-1}}<T$.  In this section we present a structural vector autoregression that can assume different parameters 
over the different regimes. Consider the following model
\begin{equation}
	\label{eq:SVARWB}
	A_{p0}y_{t}=a_p+\sum_{i=1}^{l}A_{pi}y_{t-i}+\e_{t}\hspace{2cm}\text{if }\quad T_{B_{p-1}}\leq t<T_{B_p}
\end{equation}
where $y_{t}$ is a n-dimensional vector of observable variables, $\e_{t}$ is a vector of white noise processes, normally distributed with mean zero and covariance matrix $\Lambda_p=I_n$, the identity matrix, independently of the regime we are in.\footnote{We indicate, by convention, $T_{B_0}=1$ and $T_{B_s}=T$. Moreover, we consider the initial conditions for the first regime, $y_{0},\ldots,y_{1-l}$, as given, while for all other regimes they are fixed as the last $l$ observations of the previous regime, in order to guarantee the contiguity of the regimes on the whole sample. The lag length, for simplicity, is fixed to be equal in all the regimes. However, generalizing to a regime-specific number of lags does not alter the results in any directions.} The $n\times 1$ vectors $a_p$ are the intercepts and the $n\times n$ matrices $A_{pi}$, with $i=0,\ldots,l$ are the structural parameters in the regime $p$, with $p\in\{1,\ldots,s\}$. For the \textit{p}-th regime the structural parameters can be indicated as $U_p=\left(\Ap, \Apd\right)\in \Theta_p \subset \Re^{\left(n+m\right)n}$, with $m=nl+1$, and where the $n\times m$ matrix $\Apd$ is defined as $\Apd\equiv \left(a_{p}, A_{p1},\ldots, A_{pl}\right)$. We indicate the open dense set of all structural parameters by $\mathbb{P}^S\equiv\bigcup\limits_{p=1}^{s}U_p\subset \Re^{sn\left(n+m\right)}$. Hereafter, the model in Eq. (\ref{eq:SVARWB}) is denoted as SVAR-WB.\footnote{This acronym has been already used by \cite{BCF17} for a similar, but less general, SVAR model with structural breaks.}

The reduced form VAR model can be written as  
\begin{equation}
	\label{eq:VARWB}
	y_{t}=b_p+\sum_{i=1}^{l}B_{pi}y_{t-i}+u_{t}\hspace{2cm}\text{if }\quad T_{B_{p-1}}\leq t<T_{B_p}
\end{equation}
where $b_p=\Api a_p$, $B_{pi}=\Api A_{pi}$, with $p\in\{1,\ldots,s\}$. Furthermore, for the \textit{p}-th regime, the vector of error terms is defined as $u_t=\Api\e_t$, with $E\left(u_tu_t^\prime\right)=\Sigma_p=\Api A_{p0}^{-1\prime}$. The VAR model, thus, other than presenting different parameters concerning the conditional mean, also presents different covariance matrices of the error terms $\Sigma_1,\ldots,\Sigma_s$, and thus heteroskedasticity, as in \cite{Rigobon03REStat}, \cite{LanneLutkepohl08JMCB}, and 
\cite{BacchiocchiFanelli15}. For each regime $p$, the reduced form parameters are denoted by $\phi_p=(B_{p+},\Sigma_p)\in \Phi_p \subset \Re^{n+n^2p}\times \Omega$, where $B_{p+}\equiv \left(b_p,B_{p1},\ldots,B_{pl}\right)$ and $\Omega$ is the space of positive-semidefinite matrices. The set of all reduced-form parameters is denoted by $\mathbb{P}^R\equiv\bigcup\limits_{p=1}^{s}\Phi_p\subset \Re^{sn\big( m+(n+1)/2\big)}$. 

In order to avoid problems with non-stationarity, we restrict the domain $\Phi_p$ such that all the roots of the characteristic polynomial lie outside the unit circle, and thus,  for each of the regimes $p=\{1,\ldots,s\}$, there exists a corresponding VMA$(\infty)$ representation.

Under this last condition, for each regime it is possible to find the VMA$(\infty)$ representation of the reduced form model in Eq. (\ref{eq:VARWB}) as
\begin{eqnarray}
\label{eq:VMA}
	y_t & = & c_p+\sum_{i=0}^{\infty}C_{pi}u_{t-i}\nonumber\\
	& = & c_p+\sum_{i=0}^{\infty}C_{pi}\Api\e_{t-i}
\end{eqnarray}
where $C_{pi}$ is the \textit{i}-th coefficient matrix of $\bigg(I_n-\sum_{i=1}^{l}B_{pi}L^i\bigg)^{-1}$, with $p\in\{1,\ldots,s\}$. The impulse response $IR_p^h$ becomes
\begin{equation}
\label{eq:IRh}
	IR_p^h=C_{ph}\Api
\end{equation}
whose \textit{(i,j)}-element represents, for each regime, the response of the \textit{i}-th variable of $y_{t+h}$ to a unit shock on the \textit{j}-th element of $\varepsilon_{t}$. 
%The long-run impulse response matrix is defined as
%\begin{equation}
%\label{eq:IRinf}
%	IR_p^{\infty}=\lim_{h \to \infty} IR_p^h=\bigg(I_n-\sum_{j=1}^{p}B_{pj}\bigg)\Api
%\end{equation}
%and the 
The long-run cumulative impulse response matrix, instead, is defined as
\begin{equation}
	\label{eq:CIRinf}
	CIR_p^{\infty}=\sum_{h=0}^{\infty}IR_p^h=\bigg(\sum_{h=0}^{\infty}C_{ph}\bigg)\Api
\end{equation}
where, as usual, $p\in\{1,\ldots,s\}$.

Suppose, for the moment, to focus on each single regime, where the model behaves as a traditional SVAR. As is well known in the SVAR literature, without any restriction it is impossible to uniquely pin down the structural parameters based on the reduced form of the model. This clearly continues to hold even in the case of different regimes. Furthermore, looking at Eq. (\ref{eq:SVARWB}), if all the parameters are completely unrelated across the $s$ regimes, then there is no advantage in investigating the regimes jointly and the analysis on each SVAR can be equivalently performed separately. 

If, instead, we suppose some parameters to remain stable across the regimes, then \cite{BacchiocchiFanelli15} proved there is some gain in terms of identification in performing the analysis jointly in all the regimes. Were all the regimes independent one to the other, the mapping between the reduced and structural form can be defined on each single regime separately as a function $h_p:\Theta_p \longrightarrow \Phi_p$ where 
\begin{equation}
	\label{eq:Hfunction_p}
	h_p(\Ap,\Apd)=\bigg(\Api\Apd,\Api\Apit\bigg).
\end{equation}

Instead, if we allow for potential relationships among the structural parameters across the regimes, we cannot restrict the analysis regime-by-regime, but we need a more general formulation of such a relationship. Indicating with $\Atot$ the generic point on the joint (all the regimes) parametric space, i.e. $\Atot=(A_{10},A_{1+},\ldots,A_{s0},A_{s+})$, we define a function 
$h:\mathbb{P}^S \longrightarrow \mathbb{P}^R$ of the form
\begin{eqnarray}
	h\Atot & = & \Big[A_{10}^{-1}A_{1+}\:,\:A_{10}^{-1}A_{10}^{-1\prime}\:,\:\cdots\:,\:A_{s0}^{-1}A_{s+}\:,\:A_{s0}^{-1}A_{s0}^{-1\prime}\Big]
		\nonumber\\
	& = & \Big[B_{1+}\:,\:\Sigma_1\:,\:\cdots\:,\:B_{s+}\:,\:\Sigma_{s}\Big]\nonumber\\
	& = & \Btot 	\label{eq:Hfunction}
\end{eqnarray}
mapping the structural-form parameters to those of the reduced form for the whole SVAR-WB, denoted by $\Btot\in \mathbb{P}^R$. This function will be used in the identification issue presented in the following section.

%%%%%%%%%%%%%%%%%%%%%%%%%%%%%%%%%%%%%%%%%%%%%%%%%%%%%%%%%%%%%%%%%%%%%%%%%%%%%%%
%														 SECTION IDENTIFICATION														%
%%%%%%%%%%%%%%%%%%%%%%%%%%%%%%%%%%%%%%%%%%%%%%%%%%%%%%%%%%%%%%%%%%%%%%%%%%%%%%%

\section{Theory of identification}
\label{sec:identification}

This section is dedicated to the identification issue in a context of SVAR-WB, that, as shown in \cite{BacchiocchiFanelli15}, is more general than the separate analysis of each single regime. In Section \ref{sec:ident_def} and \ref{sec:ident_restr}  we introduce the notion of identification adopted in the paper and the kind of restrictions to which the theory applies, respectively. Section \ref{sec:GlvsLoc} examines in depth the two notions of global and local identification in SVAR-WBs. In section \ref{sec:Id}, instead, 
we present and discuss some conditions for checking the identification of SVAR-WBs. A set of illustrative examples are shown in Appendix \ref{app:examples}.

\subsection{Identification}
\label{sec:ident_def}

The notion of identification we adopt in the present paper is the one introduced by \cite{Rothenberg71ECTA}, and recently adapted to SVARs by RWZ. For each regime \textit{p}, we state that two parameter points $\left(\Ap,\Apd\right)$ and 
$\left(\tilde{A}_{p0},\tilde{A}_{p+}\right)$, are observationally equivalent if and only if they imply the same distribution of $y_t$ within the \textit{p}-th regime, or equivalently, given the gaussianity assumption made before, they have the same reduced form representation $\left(B_{p+},\Sigma_p\right)$. It is easy to see that this definition refers to the existence of an orthogonal matrix $Q_p \in \O$, the set of $n\times n$ orthogonal matrices, such that two matrices $\Ap$ and $\tilde{A}_{p0}=\Qpt\Ap$ satisfy $\Sigma_p=\Api\Apit=\tilde{A}_{p0}^{-1}\tilde{A}_{p0}^{-1\prime}$.

%If we consider each regime separately, as in the standard literature of SVAR models, it is common to introduce the following definitions of identification:
%\begin{defin}[Global identification in the \textit{p}-th regime]
%	\label{def:global}
%	For the \textit{p}-th regime, a parameter point $\left(\Ap,\Apd\right)\in U_p$ is globally identified if and only if there is no other parameter 
%	point that is observationally equivalent. 
%\end{defin}

%\begin{defin}[Local identification in the \textit{p}-th regime]
%	\label{def:local}
%	For the \textit{p}-th regime, a parameter point $\left(\Ap,\Apd\right)\in U_p$ is locally identified if and only if there is an open 
%	neighborhood about $\left(\Ap,\Apd\right)$ containing no other observationally equivalent parameter point. 
%\end{defin}

However, as we allow for relationships between the structural parameters of different regimes, and consequently the analysis is performed at SVAR-WB level (and not regime-by-regime), we need to extend the previous definitions to the case of identification of the entire model. This is extremely important because, as we will see later on, the SVAR-WB can be identified even if each single regime alone can not. We provide the following definitions.

\begin{defin}[Global identification in SVAR-WB]
	\label{def:globalWB}
	A parameter point $\Atot\in\mathbb{P}^S$ is globally identified if and only if there is no other parameter point that is 
	observationally equivalent. 
\end{defin}

\begin{defin}[Local identification in SVAR-WB]
	\label{def:localWB}
	A parameter point $\Atot\in\mathbb{P}^S$ is locally identified if and only if there is an open neighborhood about $\Atot$ containing 
	no other observationally equivalent parameter point. 
\end{defin}

%These last two definitions are in line with the intuition in \cite{BacchiocchiFanelli15}, according to which, if some structural parameter does not change across the regimes, there are some gains in terms of identification for all the structural parameters of the SVAR-WB. 

As for the single SVARs, it is not difficult to see that Definitions \ref{def:globalWB} and \ref{def:localWB} refer to the existence of a $Q\in \Ons$ orthogonal matrix such that the two block-diagonal matrices
\begin{equation}
	\label{eq:AvsAtilde}
	\left(\begin{array}{ccc}
	A_{10} & &\\
	& \ddots &\\
	& & A_{s0}
	\end{array}\right)\hspace{1cm}\text{and}\hspace{1cm}
		Q^\prime \left(\begin{array}{ccc}
	A_{10} & &\\
	& \ddots &\\
	& & A_{s0}
	\end{array}\right) 
\end{equation}
give rise to the same VAR-WB covariance matrix
\begin{eqnarray}
	\left(\begin{array}{ccc}
	\Sigma_{1} & &\\
	& \ddots &\\
	& & \Sigma_{s}
	\end{array}\right) & = & \left(\begin{array}{ccc}
	A_{10}^{-1} & &\\
	& \ddots &\\
	& & A_{s0}^{-1}
	\end{array}\right)\left(\begin{array}{ccc}
	A_{10}^{-1\prime} & &\\
	& \ddots &\\
	& & A_{s0}^{-1\prime}
	\end{array}\right)\nonumber\\
	& = & \left(\begin{array}{ccc}
	A_{10}^{-1} & &\\
	& \ddots &\\
	& & A_{s0}^{-1}
	\end{array}\right)QQ^{\prime}\left(\begin{array}{ccc}
	A_{10}^{-1\prime} & &\\
	& \ddots &\\
	& & A_{s0}^{-1\prime}
	\end{array}\right)	\label{eq:Sigma_A0}
\end{eqnarray}
where the block-diagonal\footnote{The matrix $Q$ must be block diagonal to guarantee $\mathbold{\Sigma}$ to be block diagonal. As an example, let 
\begin{equation}
	\nonumber
	\left(\begin{array}{cc}\Sigma_1&\\&\Sigma_2\end{array}\right)=
	\left(\begin{array}{cc}A_{10}^{-1}&0\\0&A_{20}^{-1}\end{array}\right)
	\left(\begin{array}{cc}Q_1&Q_{12}\\Q_{21}&Q_2\end{array}\right)
	\left(\begin{array}{cc}Q_1^\prime&Q_{21}^\prime\\Q_{12}^\prime&Q_2^\prime\end{array}\right)
	\left(\begin{array}{cc}A_{10}^{-1\prime}&0\\0&A_{20}^{-1\prime}\end{array}\right).
\end{equation}
With simple algebra we obtain the two conditions $A_{10}^{-1}Q_{12}Q_{12}^\prime A_{10}^{-1\prime}=0$ and $A_{20}^{-1}Q_{21}Q_{21}^\prime A_{20}^{-1\prime}=0$ to hold. However, given that both $A_{10}$ and $A_{20}$ are invertible, then it must hold $Q_{12}=0$ and $Q_{21}=0$. The generalization is straightforward.} orthogonal matrix $Q$ is defined as
\begin{equation}
	\label{eq_bigP}
	Q=\left(\begin{array}{ccc}
	Q_{1} & & \\
	& \ddots & \\
	&  			 & Q_{s}
	\end{array}\right)\nonumber. 
\end{equation}

This clearly corresponds to say that, local or global identification for the SVAR-WB is related to the existence of an orthogonal matrix $Q_p$, such that for the \textit{p}-th regime the two sets of parameters $(A_{p0},A_{p+})$ and $(Q^\prime A_{p0},Q^\prime A_{p+})$ are observationally equivalent in the parametric space or in the neighborhood of a specific point. 

Combining Eq. (\ref{eq:Hfunction}) with Eq. (\ref{eq:Sigma_A0}) clearly shows as the identification issue is related with the mapping between the structural- and reduced-form parameters, as denoted by the $h$ function introduced before.\footnote{This does correspond to the orthogonal reduced-form parameterization as in the spirit of \cite{ARW18}.} Thus, given a reduced-form parameter point $\Btot\in \mathbb{P}^R$, if the set $\big\{\Atot\in\mathbb{P}^S:\:\Atot=h^{-1}\Btot\big\}$ is a singleton, then we have global (or point) identification, if it is a set of numerable isolated points we have local identification, and finally, if the set is open, we have set identification.\footnote{This does not complete all the possibilities. \cite{BK19} show that, for locally-identified SVARs, given a reduced-form parameter point, all the associated structural-form parameters obtained through the mapping function $h$ could be not real.}

%%%%%%%%%%%%%%%%%%%%%%%%%%%%% Subsection: Restrictions  %%%%%%%%%%%%%%%%%%%%%%%%%%%%%%%%%%%%%%%%%%%

\subsection{Identifying restrictions}
\label{sec:ident_restr}

%It is well known that a standard SVAR model is not identified and the common practice to get identification is to restrict the parametric space. Recent contributions, instead, have proposed alternative approaches that, instead of introducing restrictions on the parameters, search for further information contained in the data. \cite{Rigobon03REStat} and \cite{LanneLutkepohl08JMCB} exploit different heteroskedasticity regimes featuring the data to identify the structural parameters of simultaneous equations models and SVARs without imposing restrictions, imposing however that the structural parameters do not change in the different regimes, that are only generated by the different variability of the structural shocks. \cite{BacchiocchiFanelli15} provide a more general framework that considers changes in the structural parameters too. Our theory of identification relates to this literature and genaralizes the \cite{BacchiocchiFanelli15} contribution in two directions: a) we introduce a wider set of possible restrictions on the structural parameters and b) provide conditions for global rather than local identification.

\cite{BacchiocchiFanelli15} propose a SVAR-WB with linear restrictions on $\Api$, as in line with the large part of contributions in the standard SVAR literature, both triangular (e.g. \citeauthor{CEE96}, \citeyear{CEE96}) and non-triangular (e.g. \citeauthor{SimsZha06AER}, \citeyear{SimsZha06AER}). In this perspective, we extend their contribution in three directions. Firstly, we consider a more general class of restrictions based on linear and non-linear functions of the structural parameters as those on the impulse responses proposed by \cite{BlanchardQuah89AER}, \cite{Gali92} and RWZ. In this section we show how to consider these restrictions in a SVAR-WB framework, including the possibility of constraints across the regimes. Secondly, as largely diffused in modern applied macroeconomics, we allow for possible sign restrictions on the parameters, on impulse responses, and on FEVs. Thirdly, as a consequence of the previous one, we open the door to set identified SVAR-WBs.

\vspace{0.5cm}
\noindent
\textit{Equality restrictions}
\vspace{0.5cm}

As in RWZ we consider restrictions on a $g\times n$ transformation of the structural parameters in each regime, through a function $G(\cdot)$ with domain $U_p\subset \Theta_{p}$. Moreover, given that some structural parameters, or transformations of them through $G(\cdot)$, could remain unchanged across the regimes, we define a very general form of imposing restrictions that generalizes RWZ as follows 
\begin{equation}
	\label{eq:G}
	\underset{f_{j}\times sg}{\Rj}\underset{sg\times n}{\left(\begin{array}{c}G\left(\Atf\right)\\ \vdots\\ G\left(\Ats\right)\end{array}\right)}
	\underset{n\times 1}{e_j}=\underset{f_{j}\times 1}{0},\hspace{0.3cm}\text{for}\hspace{0.3cm}1\leq j\leq n  
\end{equation}
where $\left(\Ap,\Apd\right)$ are the structural parameters of the \textit{p}-th regime, and $e_j$ is the \textit{j}-th column of the $n\times n$ identity matrix $I_n$. For simplifying the notation, we also define $\textbf{G}\Atot\equiv\left(\begin{array}{c}G\left(\Atf\right)\\ \vdots\\ G\left(\Ats\right)\end{array}\right)$. The $f_{j}\times sg$ matrix $\Rj$ is a selection matrix allowing to impose restrictions on the \textit{j}-th structural shock of all the regimes. Through the $\Rj$ matrix we allow to impose equality restrictions on the structural parameters within each regime, or to restrict jointly parameters belonging to different regimes. In particular, we consider separately the columns of the transformed parameters $G(\cdot)$, but allow parameters of the same column, but belonging to different regimes, to be linearly related. The form of imposing restrictions described in Eq. (\ref{eq:G}) clearly generalizes RWZ; in fact, if there is no break, the $\Rj$ matrix simply collects the restrictions of the \textit{j}-th column of $G(A_0^\prime,A_{+}^\prime)$, exactly as in RWZ. In Section \ref{sec:Id} and Appendix \ref{app:examples} we will go back to this way of imposing the restrictions with more details and examples.

The rank of $\Rj$ is $f_{j}$, indicating the number of linearly independent restrictions on the \textit{j}-th shock of all the regimes. The total number of restrictions for the whole SVAR-WB, instead, is $f=\sum_{j=1}^{n}f_{j}$. Moreover, as the ordering of the columns is arbitrary, we use the suggestion by RWZ and assume that
\begin{equation}
	\label{eq:ordering}
	f_{1}\geq f_{2}\geq \cdots \geq f_{n}
\end{equation}
that simplifies the identification analysis discussed in the next sections. 

Similarly to RWZ, the transformation $G(\cdot)$ must satisfy a set of conditions in order to be admissible. In particular, the first one imposes the transformation to be the same in all the regimes, while the second states that left multiplication by an orthogonal matrix in the generic \textit{p}-th regime must commute with the transformation. These two conditions, as well as two more technical ones, necessary for the formal proofs of the identification conditions, are formalized as follows.
\begin{condition}
	\label{cond:Gequal}
	The function $G(\cdot)$ is the same in all the regimes.
\end{condition}
\begin{condition}
	\label{cond:LeftMult}
	For the \textit{p}-th regime, the transformation $G(\cdot)$, with the domain $U_p$, is admissible if and only if for any 
	$\Qp \in \O$ and $\left(\Ap,\Apd\right)\in U_p$, $G\left(\Apt\Qp,\Apdt\Qp\right)= G\left(\Apt,\Apdt\right)\Qp$. 
\end{condition}
\begin{condition}
	\label{cond:Regular}
  The transformation $\textbf{G}(\cdot)$, with the domain $\mathbb{P}^S$, is regular if and only if $\mathbb{P}^S$ is open and 
	$\textbf{G}$ is continuously differentiable with $\textbf{G}^\prime\Atot$ of rank $sgn$ for all $\Atot\in \mathbb{P}^S$.
\end{condition}
\begin{condition}
	\label{cond:StrRegular}
  For the \textit{p}-th regime, the transformation $\textbf{G}(\cdot)$, with the domain $\mathbb{P}^S$, is strongly regular 	if and only if it is regular and $\textbf{G}(\mathbb{P}^S)$ is dense in the set of $sg\times n$ matrices.
\end{condition}

It is well known that the restrictions alone are not able to uniquely identify the parameters of SVAR models, as simple multiplication by any diagonal $n\times n$ matrix $D$ with plus or minus ones along the diagonal has the effect of changing the sign to each 
row of $\left(\Ap,\Apd\right)$ leaving unchanged their relationships with those of the reduced form. The model is thus not identified and a normalization rule is required. 

\begin{defin}[Normalization]
	\label{def:norm} For each regime, a normalization rule can be characterized by a set $N_p\subset \Theta_p$ such that for any parameter point $\left(\Ap,\Apd\right)\in \Theta_p$, there exists a unique $n\times n$ diagonal matrix $D$ with plus or minus 
	ones along the diagonal such that $\left(\Apt D,\Apdt D\right) \in N_p$, $p\in\{1,\ldots,s\}$. Furthermore, as might be expected, we fix all the normalization rules to be identical, i.e. $N_1=\ldots=N_s\equiv N$.
\end{defin}

\vspace{0.5cm}
\noindent
\textit{Restrictions on the parameters as restrictions on the orthogonal matrices $Q_1,\ldots,Q_s$}
\vspace{0.5cm}

\noindent
The way we impose restrictions on the parameters, as described in Eq. (\ref{eq:G}), can be easily reconciled with the strategy of constraining the orthogonal matrices mapping the reduced-form and the structural-form parameters, as in the spirit of \cite{ARW18}, 
\cite{GK18} and \cite{BK19} for traditional SVARs. Focusing on the \textit{p}-th regime and following the intuition in \cite{Uhlig2005}, we define $A_{p0}$ as a rotation of the Cholesky factor of $\Sigma_p$, i.e. $A_{p0}=Q_p^\prime\Cholpi$, which will admit the 
decomposition $\Sigma_p=\Api\Apit$. The function $\Gp$, using the mapping between the reduced-form and the structural-form parameters, can be rewritten as $\Gp=G(\Cholpit Q_p,B_p^\prime Q_p)=G(\Cholpit,B_p^\prime)Q_p$, where the last result follows from Condition \ref{cond:LeftMult}, before. Extending the simple result to all regimes allows to write the restrictions in Eq. (\ref{eq:G}) as 
%\begin{equation}
%	\label{eq:GQ}
%	\Rj\left(\begin{array}{c}G(\Sigma_{1,tr}^{-1\prime}\,,\,B_{1+}^\prime)\:Q_1\\ \vdots\\ 
%	G\left(\Sigma_{s,tr}^{-1\prime}\,,\,B_{s+}^\prime\right) Q_s\end{array}\right)e_j=0 
%	\hspace{0.3cm}\Longleftrightarrow \hspace{0.3cm}
%	\Rj\left(\begin{array}{c}G\left(\Sigma_{1,tr}^{-1\prime}\,,\,B_{1+}^\prime\right)q_{1,j}\\ \vdots\\ 
%	G\left(\Sigma_{s,tr}^{-1\prime}\,,\,B_{s+}^\prime\right)q_{s,j}\end{array}\right)=0,
%	\hspace{0.3cm}\text{for}\hspace{0.3cm}1\leq j\leq n
%\end{equation}
\begin{eqnarray}
	\Rj\left(\begin{array}{c}G(\Sigma_{1,tr}^{-1\prime}\,,\,B_{1+}^\prime)\:Q_1\\ \vdots\\ 
	G\left(\Sigma_{s,tr}^{-1\prime}\,,\,B_{s+}^\prime\right) Q_s\end{array}\right)e_j=0 
	& \hspace{0.3cm}\Longleftrightarrow \hspace{0.3cm} &
	\Rj\left(\begin{array}{c}G\left(\Sigma_{1,tr}^{-1\prime}\,,\,B_{1+}^\prime\right)q_{1,j}\\ \vdots\\ 
	G\left(\Sigma_{s,tr}^{-1\prime}\,,\,B_{s+}^\prime\right)q_{s,j}\end{array}\right)=0,\nonumber\\
	&&\nonumber\\
	& \hspace{0.3cm}\Longleftrightarrow \hspace{0.3cm} &
	\mathbfcal{R}_j(\phi,Q) = 0,
	\hspace{0.3cm}\text{for}\hspace{0.3cm}1\leq j\leq n \label{eq:GQ}
\end{eqnarray}
with $\phi\equiv\Btot\in \mathbb{P}^R$, and where it clearly emerges as the restrictions on the structural parameters become restrictions on the columns of the orthogonal matrices. Moreover, this way of writing the restrictions allows to manage constraints across the regimes that, as said before, is crucial for the aim of the present paper. In practice, other than imposing restrictions within each regime as in RWZ or \cite{GK18}, for instance, we allow the coefficients on the same equation but in different regimes to remain unchanged, or the (short- or long-run) impact of a shock on a specific variable to be the same across regimes.\footnote{Actually, the restrictions in Eq.s (\ref{eq:G})-(\ref{eq:GQ}) are much more versatile, but these cases are the most interesting ones in empirical applications.} 

\vspace{0.5cm}
\noindent
\textit{Sign and ranking restrictions}
\vspace{0.5cm}

\noindent
An alternative, or in addition, to equality restrictions could be to impose more restrictive sign restrictions than normalization ones. When sign restrictions to impulse responses or to structural parameters are imposed alone, as in \cite{Uhlig2005}, we obtain identified sets rather than point (global or local) identification. When they complement (zero) equality restrictions, as in \cite{ARW18} and \cite{GK18}, they allow to tighten the identified sets. 

Moreover, in a recent contribution, \cite{AHD21} propose to consider ranking restrictions on impulse responses or on structural parameters of SVARs in order to narrow the identified sets generally obtained through sign restrictions. Such ranking restrictions can be found in previous investigations, or come from micro data on heterogeneous industries or households.\footnote{\cite{AHD21} introduce heterogeneous and elasticity restrictions, that essentially impose a ranking on the relative magnitude of impulse responses or elasticities of certain variables. Moreover, they consider slope restrictions, that consist in ranking the response of the same variable over different horizons.} Our SVAR-WB framework offers a natural extension to their approach. In fact, other than imposing ranking restrictions within each regime, exactly as advocated by the two authors, we open the possibility to ranking restrictions across regimes. As an example, we can have information that, at the same horizon, the response of a variable to a particular shock is higher or lower in a regime relative to another one. 

As for the equality restrictions, sign and ranking restrictions can be seen as further constraints on the columns of $Q_1,\ldots,Q_s$ matrices, jointly. Moreover, when extending to ranking restrictions across regimes, the constraints involve the columns of all $Q_1,\ldots,Q_s$, jointly. The way we can write such constraints is
%\begin{equation}
%	\label{eq:SignRestr_jp}
%	\Lj\left(\begin{array}{c}G(\Sigma_{1,tr}^{-1\prime}\,,\,B_{1+}^\prime)\:Q_1\\ \vdots\\ 
%	G\left(\Sigma_{s,tr}^{-1\prime}\,,\,B_{s+}^\prime\right)
%	Q_s\end{array}\right)e_j\,\geq\,\textbf{0}
%	\hspace{0.3cm}\text{for}\hspace{0.3cm}1\leq j\leq n, 
%\end{equation}
\begin{equation}
	\label{eq:SignRestr_jp}
	\Lj\left(\begin{array}{c}G(\Sigma_{1,tr}^{-1\prime}\,,\,B_{1+}^\prime)\:Q_1\\ \vdots\\ 
	G\left(\Sigma_{s,tr}^{-1\prime}\,,\,B_{s+}^\prime\right)
	Q_s\end{array}\right)e_j\,\geq\,\textbf{0}
	\hspace{0.3cm}\Longleftrightarrow \hspace{0.3cm}
	\mathbfcal{L}_j(\phi,Q)\,\geq\, 0,
	\hspace{0.3cm}\text{for}\hspace{0.3cm}1\leq j\leq n, 
\end{equation}
denoting a restriction on the sign or on the ranking of the function $G(\cdot)$ of the parameters, subject to the \textit{j}-th shock on the different regimes. The selection matrix $\Lj$  will be made of $1$, $-1$, or possible known parameters $\lambda$, in order to respect the particular sign denoted in Eq. (\ref{eq:SignRestr_jp}). 

\vspace{0.5cm}
\noindent
\textit{Restrictions on the Forecast Error Variance (FEV)}
\vspace{0.5cm}

For the generic regime $p$, following \cite{V22}, the FEV at horizon $\tilde{h}$ is
\begin{equation}
\label{eq:FEVp}
	FEV_p(\tilde{h})=\sum_{h=0}^{\tilde{h}}IR_p^h\,IR_p^{h\prime}.  
\end{equation}
Focusing on the contribution of shock $j$ to the $FEV$ of variable $i$ at horizon $\tilde{h}$, and for the regime $p$, we have
\begin{eqnarray}
	CFEV_{p,ij}(\tilde{h})& = & \frac{\sum_{h=0}^{\tilde{h}}\left(e_i^\prime IR_p^h e_j\right)^2}
													{\sum_{z=1}^{n}\sum_{h=0}^{\tilde{h}}\left(e_i^\prime IR_p^h e_z\right)^2}
													%\nonumber\\
												%&   & \nonumber\\
												%& = & 
												=
												\frac{\sum_{h=0}^{\tilde{h}}\left(e_i^\prime C_{ph}\Sigma_{p,tr}Q_p e_j\right)^2}
													{\sum_{z=1}^{n}\sum_{h=0}^{\tilde{h}}\left(e_i^\prime C_{ph}\Sigma_{p,tr}Q_p e_z\right)^2}\nonumber\\
												&   & \nonumber\\
												& = & q_{p,j}^\prime \Upsilon_{p,i}^{\tilde{h}}(\phi_p)q_{p,j}, \label{eq:CFEVp} 
\end{eqnarray}
with $\phi_p$ denoting the set of reduced-form parameters in the regime $p$, and where
$\Upsilon_{p,i}^{\tilde{h}}(\phi_p)=\frac{\sum_{h=0}^{\tilde{h}}\left(e_i^\prime C_{ph}\Sigma_{p,tr}\right)^\prime \left(e_i^\prime C_{ph}\Sigma_{p,tr}\right)} 
		{\sum_{h=0}^{\tilde{h}}\left(e_i^\prime C_{ph}\Sigma_{p,tr}\right)\left(e_i^\prime C_{ph}\Sigma_{p,tr}\right)^\prime}$, as shown in \cite{U04}. \cite{V22} proposes bound restrictions on the forecast error variance decomposition that, in our framework, within each regime, read as
\begin{equation}
\label{eq:fevRestBasic}
	\underline{k}_{p,ij}^{\tilde{h}}\leq q_{p,j}^\prime\Upsilon_{p,i}^{\tilde{h}}(\phi_p)q_{p,j}\leq\overline{k}_{p,ij}^{\tilde{h}},
\end{equation}
where $\underline{k}_{p,ij}^{\tilde{h}}$ and $\overline{k}_{p,ij}^{\tilde{h}}$ do represent the lower and upper \textit{known} bounds. 

A natural extension to SVAR-WBs is to jointly consider the relative effect of the shock $j$ on the FEV of variable $i$ in two or more regimes. As an example, we can impose a restriction of the form
\begin{equation}
\label{eq:fevRestEx}
	q_{p_1,j}^\prime\Upsilon_{p_1,i}^{\tilde{h}}(\phi_{p_1})q_{p_1,j} - q_{p_2,j}^\prime\Upsilon_{p_2,i}^{\tilde{h}}(\phi_{p_2})q_{p_2,j} \geq 0,
\end{equation}
indicating that, at horizon $\tilde{h}$, the effect of shock $j$ on the FEV of variable $i$ is larger in regime $p_1$ than in regime $p_2$. A very general way of writing this kind of restrictions is
\begin{gather}
\label{eq:fevRestGen}
	\underline{k}_{p,ij}^{\tilde{h}}\leq e_j^\prime \left(\begin{array}{c}Q_1\\ \vdots\\Q_s\end{array}\right)^\prime
	\left(\begin{array}{ccc}
	\gamma_{1,1}\Upsilon_{1,i}^{\tilde{h}}& &\\ &\ddots&\\&&\gamma_{1,s}\Upsilon_{s,i}^{\tilde{h}}
	\end{array}\right)
	\left(\begin{array}{c}Q_1\\ \vdots\\Q_s\end{array}\right)e_j \nonumber\\
	\hspace{5cm}-e_r^\prime \left(\begin{array}{c}Q_1\\ \vdots\\Q_s\end{array}\right)^\prime
	\left(\begin{array}{ccc}
	\gamma_{2,1}\Upsilon_{1,i}^{\tilde{h}}& &\\ &\ddots&\\&&\gamma_{2,s}\Upsilon_{s,i}^{\tilde{h}}
	\end{array}\right)
	\left(\begin{array}{c}Q_1\\ \vdots\\Q_s\end{array}\right)e_r 
	\leq\overline{k}_{p,ij}^{\tilde{h}}\nonumber
\end{gather}
with $1\leq j,r\leq n$, and where the $\gamma$'s are known selection scalars.\footnote{We have omitted $(\phi_p)$ from $\Upsilon_{p,i}^{\tilde{h}}(\phi_p)$ for the sake of simplicity.}
Finally, let 
\begin{equation}
\label{eq:fevRest}
	\underline{K}_i\leq \mathbfcal{F}_i(\phi,Q) \leq \overline{K}_i, \text{ for}\hspace{0.3cm}1\leq i\leq n
\end{equation}
be the compact way to represent the set of all the FEV restrictions on the variable $i$ within or across the regimes. 

Such restrictions do generalize those originally presented in \cite{V22} in two directions: (a) we allow for restrictions involving different shocks, as in the spirit of \cite{CV22}, within the same regime and, (b) we allow for restrictions regarding FEVs in different regimes. This latter, is totally new in the literature. However, the price to pay when extending to multiple shocks and multiple regimes is that deriving analytical results for checking non-emptiness of the identified sets is extremely complicated, as discussed in \cite{AHD21}, for which analytical conditions are derived only for tri-variate SVARs for ranking restrictions. For FEV restrictions, \cite{CV22} provide conditions for existence and uniqueness of a solution in the generalization of the Max Share identification approach allowing for simultaneous identification of a multiplicity of shocks. Their approach, however, does not consider restrictions across the regimes, that is a peculiarity of our methodology. The solution that we propose in this paper is to check for non-emptiness only numerically, that, of course, is less precise, although potential errors are expected to reduce for a large number of checks.

From hereafter, with the wording \textit{inequality restrictions} we generally mean to sign, normalization, ranking and FEV restrictions.
 
\vspace{0.7cm}
\noindent
\textit{Restricted parameters and identified set}
\vspace{0.4cm}

\noindent
Based on all the kinds of restrictions developed here before, for any reduced-form parameter $\phi=\Btot\in \mathbb{P}^R$ we can now define the set of admissible parameters or, equivalently, the admissible restrictions, of the SVAR-WB as follows:
\begin{eqnarray}
		\Ar & \equiv & \Bigg\{\Atot\in \mathbb{P}^S\cap N,\; \Bigg|
		\mathbfcal{R}_j(\phi,Q)=0,\;\mathbfcal{L}_j(\phi,Q)\,\geq\, 0,\;\underline{K}_i\leq \mathbfcal{F}_i(\phi,Q) \leq \overline{K}_i,\nonumber\\ 
		& & \hspace{2cm}i,j\in\{1,\ldots,n\},\; p\in\{1,\ldots,s\}\Bigg\}.\label{eq:R}
\end{eqnarray}
Importantly, according to the previous Conditions \ref{cond:Gequal} to \ref{cond:StrRegular}, if the transformation is admissible and (strongly) regular, then the restrictions in Eq. (\ref{eq:R}) will be admissible and (strongly) regular, too. 

Equivalently, but from a different perspective, it could be interesting to look at the set of all the admissible matrices $(Q_1,\ldots,Q_s)$ that, conditional on the reduced-form parameters, satisfy the inequality restrictions:
\begin{eqnarray}
	\Qr & \equiv & \Bigg\{\left(\begin{array}{ccc}Q_1&&\\&\ddots&\\&&Q_s\end{array}\right),\;Q_p\in\O,p\in\{1,\ldots,s\}\Bigg|
		\mathbfcal{R}_j(\phi,Q)=0,\;\mathbfcal{L}_j(\phi,Q)\,\geq\, 0,\;\nonumber\\ 
	& & \underline{K}_i\leq \mathbfcal{F}_i(\phi,Q) \leq \overline{K}_i,\; \big(\Sigma_{p,tr}^{-1\prime}Q_p D\,,\,B_{p+}^\prime Q_p D\big)\in N_p,\; i,j\in\{1,\ldots,n\},\; p\in\{1,\ldots,s\}\Bigg\}.\nonumber\label{eq:Qrestr}
\end{eqnarray}
Moreover, once defined the set $\Qr$, it could be of interest to collect the admissible orthogonal matrices operating in each regime. In this respect, for $p\in\{1,\ldots,s\}$ we define
\begin{equation}
\label{eq:Qp}
	\Qrp\equiv \bigg\{\big(e_p \otimes I_n\big)^\prime Q \big(e_p \otimes I_n\big)\:\bigg|\:Q\in \Qr\bigg\}\nonumber.
\end{equation}
with $e_p$ being the \textit{p}-th column of the identity matrix $I_n$.

Finally, as one might be interested not simply on the structural parameters but on transformations of them, like impulse response functions, it is also important to define what we call the identified set. Specifically for the \textit{p}-th regime, 
we define the identified set
\begin{equation}
\label{eq:ISp}
%	\IS\equiv \bigcup\limits_{p=1}^{s}\Bigg\{\eta(\phi_p,Q_p)\:\Bigg|\:
%								\left(\begin{array}{ccc}Q_1&&\\&\ddots&\\&&Q_s\end{array}\right)\in\Qr\Bigg\}\nonumber.
	\ISp\equiv \bigg\{\eta(\phi,Q_p)\:\Bigg|\:Q_p\in\Qrp\bigg\}\nonumber.
\end{equation}
with $\eta(\phi,Q_p)$ being the transformation of the structural parameters one is interested in, in the \textit{p}-th regime. 

Furthermore, considering the SVAR-WB as a whole. we define the set 
\begin{equation}
\label{eq:IS}
	\IS\equiv \bigg\{\Big(\eta(\phi,Q_1),\ldots,\eta(\phi,Q_s)\Big)\:\Bigg|\:\eta(\phi,Q_p)\in\ISp,\:p\in\{1,\ldots,s\}\bigg\}.
	\nonumber
\end{equation}

%%%%%%%%%%%%%%%%%%%%%%%%%%%  Subsection: Global vs Local identification %%%%%%%%%%%%%%%%%%%%%%%%%%%%%%%%%%%%%%%%%%%

\subsection{Global vs local identification in SVAR-WBs}
\label{sec:GlvsLoc}

The work by RWZ represents a milestone in the identification of SVARs as it extends the notion of global identification to this class of models too, that were considered to be simply locally identifiable as presenting restrictions on the covariance matrix (see, for example, \citeauthor{Magnus_Neudecker_2007}, \citeyear{Magnus_Neudecker_2007}, Chapter 16). A key element in their contribution is to allow for zero-restrictions only, and to avoid any forms of constraint either across shocks or across equations. Moreover, they show that the SVAR is globally identified if and only if such zero-restrictions follow a sort of recursive scheme (according to their Theorem 7). This result is extremely powerful for different reasons. First, given the parameters of the reduced form, there exists a unique orthogonal matrix mapping to unique structural parameters coherent with the imposed restrictions; second, such a matrix always exists; third, there is a very simple algorithm for deriving this matrix. 

\cite{BK19}, instead, investigate all the cases in which the kind of restrictions imposed leads to local identification. Interestingly, they also show, through a very simple bi-variate example, that SVAR-WBs, as originally proposed by \cite{BacchiocchiFanelli15}, might face the feature of local identification.

Investigating the problem more in depth, this result is not surprising and, furthermore, can be generalized to all SVAR-WBs where we have gains in terms of identification. In fact, the gains can be obtained when there are restrictions across equations, or shocks, belonging to different regimes. This situation is not compatible with the kind of constraints considered in RWZ. Instead, they are coherent with the \citeauthor{BK19}'s (\citeyear{BK19}) framework, providing evidence of local identification.\footnote{This does not mean that SVAR-WBs cannot be globally identified; however, as we will see in details in the next section, this case can only happen when each regime is globally identified, meaning that there will be no relationships among the regimes and the analysis can be performed regime-by-regime independently.} For this reason, in the following sections, the identification analysis of SVAR-WBs is pursued, at best, by focusing on the local identification case.

\subsection{Conditions for identification in SVAR-WB}
\label{sec:Id}

In Section \ref{sec:ident_def} we have seen that lack of global identification arises when, at least in one regime, there is an orthogonal matrix $Q_p$ that makes the sets of admissible parameters $(A_{p0},A_{p+})$ and $(Q_pA_{p0},Q_pA_{p+})$ observationally equivalent. When focusing on local identification, as we are obliged to do in this paper, such orthogonal matrix takes the form of an infinitesimal rotation $Q_p=(I_n+H_p)$. Such infinitesimal rotation has the feature that $H_p$ is a skew-symmetric (or sometimes denoted by emi-symmetric or anti-symmetric or anti-metric) matrix, i.e. $H_p^\prime=-H_p$.\footnote{For more details, See \cite{Lucchetti06ET} or \cite{BL18}.} 

\begin{defin}[Admissible infinitesimal rotations]
	\label{def:AdmRot} The infinitesimal rotations $\Qp=(I_n+H_p)$, $p=\{1,\ldots,s\}$, are admissible if 
	\begin{equation}
		\label{eq:AdmRot}
		\Rj\left(\begin{array}{c}G\left(A_{10}^\prime Q_1,A_{1+}^\prime Q_1\right)\\ 
		\\ G\left(A_{s0}^\prime Q_s,A_{s+}^\prime Q_s\right)\end{array}\right)e_j=0,
		\hspace{0.3cm}\text{for}\hspace{0.3cm}1\leq j\leq n. 
	\end{equation}
\end{defin}

In other words, we consider an orthogonal matrix that locally rotates the structural parameters of the \textit{p}-th regime. If such an infinitesimal rotation is admissible according to Definition \ref{def:AdmRot}, i.e. the new parameters continue to satisfy 
the restrictions, then the SVAR-WB cannot be locally identified. Overall, if we consider the parameter point $\Atot\in R$, the restrictions should guarantee that, for each regime, the unique infinitesimal admissible rotation matrix be the identity matrix $I_n$. 
In what follows, we discuss on how to check about the existence of such infinitesimal admissible rotation matrices. We first introduce a necessary order condition and then move to more restrictive sufficient rank conditions. 

\newpage
%\vspace{0.5cm}
\noindent
\textit{Necessary order condition for identification of SVAR-WB}
\vspace{0.5cm}

\noindent
The Rothenberg order condition for a SVAR without any breaks says that the number of restrictions must be at least equal to $n(n-1)/2$. The other side of the same coin states that the number of free parameters to estimate cannot exceed $n(n+1)/2+mn$ (where $m=nl+1$ and $l$ is the number of lags).

The contributions by \cite{Rigobon03REStat} and \cite{LanneLutkepohl08JMCB} show that, if a) there is heteroskedasticity in the data generating two (or more) volatility regimes and b) the contemporaneous parameters of the structural form do not change across the regimes ($A_{10}=A_{20}=A_0$ in our notation), than all the parameters of the SVAR can be identified (up to an appropriate normalization rule) without the need to impose any restriction. It clearly emerges that in the case of two volatility regimes with the specific parametrization by \cite{Rigobon03REStat} and \cite{LanneLutkepohl08JMCB} the Rothenberg's order condition is no longer necessary.

However, it is also clear that under this parametrization, the SVARs in the two regimes produce the same impulse responses, once standardized for the shock to have the same magnitude. \cite{BacchiocchiFanelli15}, instead, provide a different parametrization of the 
model that allows for a change in the contemporaneous parameters of the structural form before and after the break. When all the structural parameters, thus, are allowed to change across the regimes, a new order condition is now necessary. 

The next theorem, already presented in \cite{BacchiocchiFanelli15} (Proposition TS1, Supplementary Material), provides a necessary order condition for the case of a SVAR-WB with $s$ volatility regimes.

\medskip

\begin{theo}
  \label{theo:OrdCond}
	Consider an SVAR-WB as in Eq. (\ref{eq:SVARWB}) with $s$ regimes and admissible parameters represented by $\Ar$. A necessary condition for the identification of the parameters is that $f\geq s\,\nt$, where $\nt=n(n-1)/2$ .
\end{theo}

\medskip
 
\begin{proof} 
	The formal proof will be discussed in the next section, as it directly comes out as a by-product of the proof of Theorem \ref{theo:SuffCond}. Here below we simply provide the intuition behind the result of 
	Theorem \ref{theo:OrdCond}. As largely discussed before, lack of identification is caused by the presence of orthogonal admissible matrices $Q_1,\ldots,Q_s$ that generate new admissible structural parameters. 
	In the case of local identification, such admissible matrices take the form of $Q_p = (I_n+H_p)$, $p=\{1,\ldots,s\}$, with $H_1,\,\ldots\,,H_p$ skew-symmetric matrices. However, being skew-symmetric, 
	each of them is characterized by $\nt=n(n-1)/2$ distinct elements. As we will see here below, a sufficient condition for identification is based on the solution of a homogeneous system of equations for the 
	elements in $H_1,\,\ldots\,,H_p$, where the number of equations is given by the number of restrictions. If the unique solution to such a system is represented by the null vector, then the SVAR-WB is 
	locally identified. However, a necessary condition for this to occur is that the number of equations to be at least equal to the number of unknowns, i.e. $f\geq s\nt$, as stated by Theorem \ref{theo:OrdCond}.
\end{proof}
	
%The previous theorem generalizes the Rothenberg's order condition for traditional SVARs to the case of a SVAR-WBs with $s$ regimes. Furthermore, when there are just two regimes and the restrictions are imposed on the responses on impact ($A_{01}^{-1}$ and $A_{02}^{-1}$), as expected, although the different specification of the two models, the previous result corresponds to the order condition provided by \cite{BacchiocchiFanelli15} (see their main identification results reported in Proposition 1, pages 768-769).

\newpage
%\vspace{0.5cm}
\noindent
\textit{Sufficient rank condition for identification of SVAR-WB}
\vspace{0.5cm}

\noindent
In this section we introduce a condition that is sufficient for identification and that considers not only quantitatively but also qualitatively the restrictions imposed on the parameters. Before going to the sufficient condition, however, we need to discuss about an alternative way to write the restrictions. Let, for simplicity, $G_p \equiv G(A_{p0}^\prime,A_{p+}^\prime)$, $p=1,\ldots,s$. The restrictions defined in Eq. (\ref{eq:G}), that we indicate as ``implicit form'', can be compacted as
\begin{eqnarray}
		\underset{f_{j}\times sg}{\left(\begin{array}{ccc}
		R_{11,j} & \cdots & R_{1s,j}\\
		\vdots  & \ddots & \vdots\\        
		R_{s1,j} & & R_{ss,j}
		\end{array}\right)}
		\underset{sg\times n}{\left(\begin{array}{c}
		G_1\\
		\vdots\\
		G_s
		\end{array}\right)}
		\underset{n \times 1}{e_j}=
		\underset{f_{j}\times 1}{\left(\begin{array}{c}
		0\\
		\vdots\\
		0
		\end{array}\right)} 
		&\Longleftrightarrow&
		\underset{f_{j}\times sg}{\left(\begin{array}{ccc}
		R_{1,j}^* & \cdots & R_{s,j}^*
		\end{array}\right)}
		\underset{sg\times n}{\left(\begin{array}{c}
		G_1\\
		\vdots\\
		G_s
		\end{array}\right)}
		\underset{n \times 1}{e_j}=
		\underset{f_{j}\times 1}{\left(\begin{array}{c}
		0\\
		\vdots\\
		0
		\end{array}\right)}\nonumber\\
		&\Longleftrightarrow& \hspace{1cm} \Rj \left(\begin{array}{c}
		G_1\\
		\vdots\\
		G_s
		\end{array}\right) e_j = 0 \label{eq:ImpForm}
\end{eqnarray}
where $R_{pp,j}$ are the restrictions in each of the regimes, while $R_{pq,j}$ are the across-regime restrictions, all referring to the \textit{j}-th structural shock. More compactly, we use the notation $R_{p,j}^*\equiv\big(R_{1p,j}^\prime,\ldots,R_{sp,j}^\prime\big)^\prime$ collecting all the restrictions imposed on the \textit{j}-th shock in the \textit{p}-th regime. An equivalent way for writing these restrictions, that we indicate as ``explicit form'', is
\begin{eqnarray}
		\underset{sg\times 1}{\left(\begin{array}{c}
		G_1\\
		\vdots\\
		G_s
		\end{array}\right)}\underset{n \times 1}{e_j}
		=\underset{sg \times \tau_j}{\left(\begin{array}{ccc}
		S_{11,j} & \cdots & S_{1s,j}\\
		\vdots  & \ddots & \vdots\\        
		S_{s1,j} & & S_{ss,j}
		\end{array}\right)}
		\underset{\tau_j\times 1}{\theta_j}
		&\Longleftrightarrow &
		\underset{sg\times 1}{\left(\begin{array}{c}
		G_1\\
		\vdots\\
		G_s
		\end{array}\right)}\underset{n \times 1}{e_j}
		=\underset{sg \times \tau_j}{\left(\begin{array}{c}
		S_{1,j}^*\\
		\vdots\\
		S_{s,j}^*
		\end{array}\right)}
		\underset{\tau_j\times 1}{\theta_j}\nonumber\\
		&\Longleftrightarrow & \hspace{1cm}\left(\begin{array}{c}
		G_1\\
		\vdots\\
		G_s
		\end{array}\right) e_j = \Sj\,\theta_j
		\label{eq:ExpForm}
\end{eqnarray}
where $\theta_j$ is the vector containing the free parameters associated to the \textit{j}-th shock in all the regimes, $S_{p,j}^*\equiv\big(S_{p1,j}\ldots,S_{ps,j}\big)$, and $\Sj\equiv \textbf{R}_{j,\perp}$, i.e. the orthogonal complement of $\Rj$. The number of free parameters contributing to the identification of the \textit{j}-th shock in all the regimes is equal to $\tau_j=sg-f_j$, and $\tau=\tau_1+\ldots+\tau_n$. 

For $j,k\in\{1,\ldots,n\}$ and $p\in\{1,\ldots,s\}$, let
\begin{equation}
\label{eq:Vjik}
	V_{j,p,k}\equiv\left(\,R_{p,j}^*\cdot S_{p,k}^*\,\right)
\end{equation}
and 
\begin{equation}
\label{eq:Vjtheta}
	\setlength{\dashlinegap}{1pt}
	\underset{f_j\times s(n-j)}{\Vtj} \equiv \left[\begin{array}{c:c:c:c}
	V_{j,1,j+1}\theta_{j+1}\;|\ldots|\;V_{j,1,n}\theta_n & V_{j,2,j+1}\theta_{j+1}\;|\ldots|\;V_{j,2,n}\theta_n 
	& \ldots & V_{j,s,j+1}\theta_{j+1}\;|\ldots|\;V_{j,s,n}\theta_n\end{array}\right].
\end{equation}

The next theorem provides a sufficient condition for the SVAR-WB to be locally identified at a certain point in the parameter space.

\medskip

\begin{theo}[Sufficient rank condition for identification]
  \label{theo:SuffCond}
	Consider an SVAR-WB as in Eq. (\ref{eq:SVARWB}) with $s$ regimes and admissible parameters represented by $\Ar$. 
	The SVAR-WB is locally identified at the parameter point $\Atot\in \Ar$ if, for $j\in\{1,\ldots,n-1\}$, 
	\begin{equation}
		\label{eq:RkSuff}
		\rk\:\bigg( \,\textbf{V}_j(\theta_0) \,\bigg) \:=\:s(n-j),
	\end{equation}
	where $\theta_0=\big(\theta_{0,1}^\prime,\theta_{0,2}^\prime,\cdots,\theta_{0,n}^\prime\big)^\prime$ is such that 
	$\bigg(\textbf{G}\Atot\bigg)e_j=\textbf{S}_j\theta_{0,j}$, for every $j\in\{1,\ldots,n\}$.	
\end{theo}

\medskip
 
\begin{proof} 
	See the appendix.
\end{proof}

\cite{BacchiocchiFanelli15}, in their Proposition TS1 (Supplementary Material), propose a necessary and sufficient condition for local identification of SVAR-WBs. There are two main differences with respect to our Theorem \ref{theo:SuffCond}. The former is on the fact that we deal with the identification shock-by-shock, while they treat the identification issue on the SVAR-WB as a whole. The advantage, in our strategy, is that we can focus on the identification of the shocks of interest, only, leaving the SVAR-WB, as a whole, not identified. This point is better discussed in the next section, focusing on partial identification in SVAR-WBs. The latter, instead, is about the kind of restrictions used for identifying the model. In the present paper we consider a much wider family of restrictions, including impulse responses at different horizons. \cite{BacchiocchiFanelli15}, instead, focus on restrictions on the response on impact, only. However, they allow for non-homogeneous restrictions, while we only deal with zero restrictions. The advantage of our way of imposing restrictions, as we will see, is that it naturally allows for extending to the set identification case.

\begin{corol}
  \label{corol:OrderCond}
	Consider an SVAR-WB as in Eq. (\ref{eq:SVARWB}) with $s$ volatility regimes, with admissible parameters represented by $\Ar$. 
	If the sufficient rank condition of Theorem \ref{theo:SuffCond} is satisfied, then $f_{j}\geq s(n-j)$, for $j=1,\ldots,n-1$.
\end{corol}

\medskip
 
\begin{proof} 
	The proof immediately follows from Theorem \ref{theo:SuffCond}. In the system of equations (\ref{eq:RadoSys}) (in Appendix \ref{app:Proofs}), the number of equations is at least as large as the number of the unknowns, i.e. the 
	below diagonal elements of the \textit{j}-th column of $H_1$, $H_2$, $\ldots$, $H_s$.
\end{proof}

The corollary provides an interesting result that can be seen as a necessary condition for identification of \textit{recursive} SVAR-WB. Put differently, if the condition of the corollary is met, then we can use Theorem \ref{theo:SuffCond} to effectively check whether the SVAR-WB is locally identified. On the contrary, if the condition of the corollary does not hold, it does not mean that the model is not identified, it simply states that it is not recursive, and other strategies have to be used for studying identification, as we will see in the following Theorem \ref{theo:NecSuff}. % (and Theorem \ref{theo:NecSuffStruct} in Appendix \ref{app:StructCond}).

\vspace{0.5cm}
\noindent
\textit{Partial identification in SVAR-WB}
\vspace{0.5cm}

\noindent
The result obtained in Theorem \ref{theo:SuffCond} easily allows to check whether just a subset of structural shocks are identified. Precisely, a formal definition of what we mean for partial identification is stated as follows. 
 
\begin{defin}[Partial identification]
\label{def:PartIdent}
	The \textit{j}-th structural shock is locally identified at the parameter point $\Atot\in \Ar$ if and only if there does not exist in an open neighborhood about $\Atot$ another observationally equivalent parameter point 
	$\left(\tilde{\mathcal{A}}_0,\tilde{\mathcal{A}}_+\right)$ such that $A_{p0}^\prime e_j\neq \tilde{A}_{p0}^\prime e_j$, $A_{p+}^\prime e_j\neq \tilde{A}_{p+}^\prime e_j$, $\forall \; p\in \AllReg$,
	where $e_j$ is the \textit{j}-th column of the identity matrix. 
\end{defin}

The following theorem provides a sufficient rank condition for a subset of shocks to be locally identified at a certain point in the parameter space $\Ar$.

\medskip
 
\begin{theo}[Partial identification]
  \label{theo:RankPartial}
	Consider an SVAR-WB as in Eq. (\ref{eq:SVARWB}) with $s$ regimes and admissible parameters represented by $\Ar$. The j-th structural shock is locally identified at the parameter point $\Atot\in \Ar$ if, for $i\in\{1,\ldots,j\}$, 
	\begin{equation}
		\label{eq:RkSuffPartial}
		\rk\:\bigg( \,\textbf{V}_i(\theta_0) \,\bigg) \:=\:s(n-i),
	\end{equation}
	where $\theta_0=(\theta_{0,1}^\prime,\theta_{0,2}^\prime,\cdots,\theta_{0,n}^\prime)^\prime$ is such that $\bigg(\textbf{G}\Atot\bigg)e_j=\textbf{S}_j\theta_{0,j}$, for every $j\in\{1,\ldots,n\}$.	
\end{theo}

\medskip

\begin{proof} 
	The proof follows directly from the proof of Theorem \ref{theo:SuffCond}.
\end{proof}

\vspace{0.5cm}
\noindent
\textit{Local identification almost everywhere}
\vspace{0.5cm}

\noindent
The previous sufficient conditions presented in Theorems \ref{theo:SuffCond} and \ref{theo:RankPartial} focus on a specific point in the parametric space $\Ar$. However, it is possible to prove that, if the conditions hold in one parameter point, then they hold almost everywhere in the parameter space. The following theorem formalizes this result for the SVAR-WBs. However, we need first to define the set, 
\begin{equation}
	\label{eq:k}
	K = \big\{\Atot\in \Ar \hspace{0.3cm}\left| \hspace{0.3cm} \text{the sufficient rank condition holds}\right. \big\}.
\end{equation}
The set $K$ collects all the parameter points satisfying the admissible restrictions and for which the sufficient rank condition in Theorem \ref{theo:SuffCond} holds. This set is clearly open, implying that if the sufficient condition holds for a certain point in 
the parameter space $\Ar$, it will also holds in the neighborhood of such point. Furthermore, as reported in the next theorem, if the rank condition is satisfied for a point in the parameter space $\Ar$, the model will be locally identified in almost all points in 
the parameter space.

\medskip
 
\begin{theo}
  \label{theo:IdentEverywhere}
	Consider an SVAR-WB with admissible and regular restrictions represented by $\Ar$. Either the set $K$ is empty, or the complement of $K$ in $\Ar$ is of measure zero in $\Ar$.
\end{theo}

\medskip
 
\begin{proof} 
	See the appendix.
\end{proof}

The theorem generalizes to SVAR-WB the results obtained by RWZ for the global identification of SVARs, by \cite{BL18} for the local identification in SVARs, and by \cite{Johansen95} for simultaneous equation models and cointegrating systems. This result is extremely powerful in that, as originally discussed by \cite{Giannini92} for checking the local identification of SVAR models, the rank condition can be checked at any random value of the parameter space. If it is verified in that single point, it will be verified 
almost everywhere in $\Ar$. 

Importantly, Theorem \ref{theo:IdentEverywhere} adds important ingredients to practically check local identification in SVAR-WBs. The sufficient condition, in fact, can be easily checked directly on $\Vtj$, for $j\in\{1,\ldots,n\}$. The next algorithm provides the details.

\begin{algo}
\label{algo:IdRecursive}
	Consider an SVAR-WB with admissible and regular restrictions represented by $\Ar$. Let the ordering of the equations be as indicated in Eq. (\ref{eq:ordering}).
	\begin{enumerate}
		\item For $j\in\{1,\ldots,n\}$, let $\theta_j\sim N(0,I_{\tau_j})$ be a draw of a $\tau_j$-variate standard normal random variable.
		\item Set $j=1$.
		\item If, for $i\in\{1,\ldots,j\}$, $\rk\big(\textbf{V}_i(\theta)\big)=s(n-i)$, with $\textbf{V}_i(\theta)$ defined as in Eq. (\ref{eq:Vjtheta}), then all the shocks for $i\in\{1,\ldots,j\}$ 
					are locally identified. If $j=n-1$, then the SVAR-WB is locally identified and STOP; otherwise set $j=j+1$ and repeat Step 3. 
					If instead $\rk\big(\textbf{V}_i(\theta)\big)<s(n-i)$ for some $i\in\{1,\ldots,j\}$, go back to Step 1.
		\item Repeat Steps 1 to 3 for a maximum of $N$ times. If for all times $\rk\big(\textbf{V}_i(\theta)\big)<s(n-i)$ for some $i\in\{1,\ldots,j\}$, we conclude that the imposed	
					identifying restrictions do not achieve local identification. STOP.
	\end{enumerate}
\end{algo}

In the previous algorithm, if the SVAR-WB is identified, $N=1$ is sufficient. If, instead, the model is not identified, the rank condition should be checked for a large $N$. Even in this latter case, it continues to be very fast. The implementation of the algorithm is extremely simple and allows to check both the local identification of a subset of structural shocks, or the local identification of the SVAR-WB as a whole.

\vspace{0.5cm}
\noindent
\textit{Necessary and sufficient rank condition for identification of SVAR-WB}
\vspace{0.5cm}

The sufficient rank condition reported in Theorem \ref{theo:SuffCond} attacks the identification issue by investigating the model shock-by-shock, although jointly for all the regimes. This is allowed, as stated in Corollary \ref{corol:OrderCond}, by the fact that 
the number of restrictions is enough to potentially guarantee the matrix $\Vtj$ to have full column rank, for $j\in\{1,\ldots,n-1\}$. The sufficient condition, thus, cannot be applied to SVAR-WBs in which the imposed restrictions do not follow the particular ``recursive'' scheme described in Corollary \ref{corol:OrderCond}, although the overall number of restrictions satisfies the order condition in Theorem \ref{theo:OrdCond}.

In this section we provide a new necessary and sufficient condition that can be used in all identification schemes. It reads as the generalization of the rank condition for identification in SVARs originally proposed by \cite{AGbook} to the case of SVAR-WBs, and it is based on calculating the rank of a particular matrix that depends on the restrictions and the structural parameters of the model. Before presenting the condition, however, we need to define some more elements. Starting from the matrix $V_{j,i,k}$ introduced in Eq. (\ref{eq:Vjik}), we define the matrix
\begin{equation}
\label{eq:tVjtheta}
	\setlength{\dashlinegap}{1pt}
	\underset{f_j\times sn}{\tVtj} = \left[\begin{array}{c:c:c:c}
	V_{j,1,1}\theta_{1}\;|\ldots|\;V_{j,1,n}\theta_n & V_{j,2,1}\theta_{1}\;|\ldots|\;V_{j,2,n}\theta_n & \ldots & V_{j,s,1}\theta_{1}\;
	|\ldots|\;V_{j,s,n}\theta_n\end{array}\right]
\end{equation}
that, differently from $\Vtj$, includes for each regime $p$, all the columns $(V_{j,p,1}\theta_{1}\;|\ldots|\;V_{j,p,n}\theta_n)$, 
and not just those from $j+1$ to $n$, i.e. $\big(V_{j,p,j+1}\theta_{j+1}\;|\ldots|\;V_{j,p,n}\theta_n\big)$, as instead in Eq. (\ref{eq:Vjtheta}).
Moreover, we collect all such $\tVtj$, $j\in \{1,\ldots,n\}$, in the block diagonal matrix 
\begin{equation}
\label{eq:tVtheta}
	\underset{f\times sn^2}{\tVt}=\left(\begin{array}{cccc}
	\tilde{\textbf{V}}_1(\theta) & & &\\
	& \tilde{\textbf{V}}_2(\theta) & &\\
	& & \ddots & \\
	& & & \tilde{\textbf{V}}_n(\theta)
	\end{array}\right),
\end{equation}
where $\theta=(\theta_1^\prime,\theta_2^\prime,\cdots,\theta_n^\prime)^\prime$ and each non-zero element contains information (restrictions and free structural parameters) for the identification of the \textit{j}-th shock in all the regimes. Finally, we introduce the selection matrix, completely known and of dimension $sn^2\times s\nt$, with $\nt=n(n-1)/2$, 
\begin{equation}
\label{eq:Ttt}
	\underset{sn^2\times s\nt}{\Ttt} = \left(\underset{sn\times sn}{\tilde{T}_{n,s}}\otimes \underset{n\times n}{I_n}\right)
																								 \left(\underset{s\times s}{I_s}\otimes \underset{n^2\times \nt}{\tilde{D}_{n}}\right)
\end{equation}
where the $n^2 \times \nt$ full-column rank matrix $\tilde{D}_n$, defined in \cite{magnus88}, is such that for any $\nt$-dimensional vector $h$, it holds $\tilde{D}_n\,h = \text{vec }(H)$, with $H$ a $n\times n$ skew-symmetric matrix ($H = -H^\prime$), 
while the $sn\times sn$ matrix $\tilde{T}_{n,s}$ is defined as
\begin{equation}
	\label{eq:Tt}
	\tilde{T}_{n,s} = \left(\begin{array}{c}I_s\otimes e_1^\prime\\ \vdots\\ I_s\otimes e_n^\prime\end{array}\right)
\end{equation}
with $e_j$ being the \textit{j}-th column of $I_n$. Intuitively, if we define by $h_p$ the $\nt\times 1$ vector containing the distinct elements in the skew-symmetric matrix $H_p$, $p\in\{1,\ldots,s\}$, the matrix $\Ttt$ in Eq. (\ref{eq:Ttt}) allows to transform 
$\ve \left(\begin{array}{c}H_1\\\vdots\\H_s\end{array}\right)$, of dimension $sn^2\times 1$, into $\left(\begin{array}{c}h_1\\\vdots\\h_s\end{array}\right)$, of dimension $s\nt\times 1$. The necessary and sufficient condition for the identification of the SVAR-WB is defined in the next theorem and is based on the following matrix
\begin{equation}
\label{eq:VttTtt}
	\underset{f\times s\nt}{\Vtt} = \tVtz \: \Ttt,
\end{equation}	
that must have full column rank for local identification.

\medskip

\begin{theo}[Necessary and sufficient rank condition for identification]
  \label{theo:NecSuff}
	Consider an SVAR-WB as in Eq. (\ref{eq:SVARWB}) with $s$ volatility regimes, with admissible parameters represented by $\Ar$. The SVAR-WB is locally identified at the parameter point $\Atot\in \Ar$, if and only if 
	\begin{equation}
		\label{eq:RkNecSuff}
		\rk\:\bigg( \,\Vttz \,\bigg) \:=\:s\nt,
	\end{equation}
	where $\nt=n(n-1)/2$ and $\theta_0=(\theta_{0,1}^\prime,\theta_{0,2}^\prime,\cdots,\theta_{0,n}^\prime)^\prime$ is such that 
	$\bigg(\textbf{G}\Atot\bigg)e_j=\textbf{S}_j\theta_{0,j}$, for every $j\in\{1,\ldots,n\}$.	
\end{theo}

\medskip
 
\begin{proof} 
	See the appendix.
\end{proof}

%The necessary and sufficient rank condition, as can be seen in Eq. (\ref{eq:RkNecSuff}), depends on the true structural parameters $\theta_0$, that are unknown before the estimation process.
Theorem \ref{theo:NecSuff} can be compared to Proposition TS1 in \cite{BacchiocchiFanelli15}. The main difference is on the kind of restrictions considered. We allow for zero restrictions on the structural parameters and on functions of them, like impulse responses at different horizons. In their proposition, instead, only (zero and non-zero) restrictions on the responses on impact are considered.

The next theorem extends Theorem \ref{theo:IdentEverywhere} and adds practicality to the rank condition in Eq. (\ref{eq:RkNecSuff}). We first define the set
\begin{equation}
	\label{eq:kt}
	\tilde{K} = \big\{\Atot\in \Ar \hspace{0.3cm}\left| \hspace{0.3cm} 
	\text{the necessary and sufficient rank condition holds}\right. \big\}.
\end{equation}
 
\begin{theo}
  \label{theo:IdentEverywhereRank}
	Consider a SVAR-WB with admissible and regular restrictions represented by $\Ar$. Either the set $\tilde{K}$ is empty, 
	or the complement of $\tilde{K}$ in $\Ar$ is of measure zero in $\Ar$.
\end{theo}

\medskip
 
\begin{proof} 
	See the appendix.
\end{proof}

The previous theorem allows to check the rank condition for random values of $\theta$. If it is met for at least one $\theta$, then, according to Theorem \ref{theo:IdentEverywhereRank}, the SVAR-WB is identified for almost all structural
parameters in the parametric space. When the condition is not met, instead, one should repeat the check for different random $\theta$s in order to be more convinced that, effectively, the model to be not identified everywhere in the parametric space. This result is completely new in the literature on SVAR-WBs, and is surely an improvement relative to Proposition TS1 in \cite{BacchiocchiFanelli15}.

All the steps for the practical implementation of Theorem \ref{theo:NecSuff} are reported in the next algorithm, that mimics what already presented in the previous Algorithm \ref{algo:IdRecursive}. Given the non-recursive structure of the restrictions, however, the identification is checked for the SVAR-WBs as a whole.

\begin{algo}
\label{algo:IdGeneral}
	Consider an SVAR-WB with admissible and regular restrictions represented by $\Ar$. 
	\begin{enumerate}
		\item Define the matrix $\Ttt$ as described in Eq. (\ref{eq:Ttt}).
		\item For $j\in\{1,\ldots,n\}$, let $\theta_j\sim N(0,I_{\tau_j})$ be a draw of a $\tau_j$-variate standard normal random variable, and
					based on such $\theta_j$ build the matrix $\tVtj$, as described in Eq. (\ref{eq:tVjtheta}). 
		\item Collect all the matrices $\tVtj$, $j\in\{1,\ldots,n\}$, in the block-diagonal matrix $\tVt$ as in Eq. (\ref{eq:tVtheta}), and 
					calculate the matrix $\Vtt = \tVtz \: \Ttt$ as in Eq. (\ref{eq:VttTtt}). 
		\item If $\rk\big(\Vtt\big)=s\nt$, then the SVAR-WB is locally identified and STOP; 
					If instead $\rk\big(\Vtt\big)<s\nt$, go back to Step 2.
		\item Repeat Steps 2 to 4 for a maximum of $N$ times. If for all times $\rk\big(\Vtt\big)<s\nt$, we conclude that the imposed	
					identifying restrictions do not achieve local identification. STOP.
	\end{enumerate}
\end{algo}

The same considerations as for Algorithm \ref{algo:IdRecursive} do apply. 

\newpage
%\vspace{0.5cm}
\noindent
\textit{Just-identified SVAR-WBs}
\vspace{0.5cm}

\noindent
The idea of \textit{just identification} we consider in this section is substantially a generalization of the notion of exact identification proposed by RWZ for SVAR models. In that case, the exact identification is strictly connected to the existence of a unique orthogonal matrix mapping the reduced-form parameters to the restricted structural ones. An equivalent definition is in general precluded for SVAR-WBs in that, as discussed in Section \ref{sec:GlvsLoc}, an SVAR-WB has gains in terms of identification only when restrictions across regimes are imposed. However, this feature generates local identification that, according to \cite{BK19}, leads to a bounded number of admissible orthogonal matrices to potentially exist, and not just one, as in the global identification case. For SVAR-WBs, just identification mainly refers to the number of restrictions we impose in each equation and for each regime. The definition is as follows.

\begin{defin}[Just identification]
	\label{def:JustIdent}
	An SVAR-WB as in Eq. (\ref{eq:SVARWB}) with $s$ volatility regimes is just identified if it is locally identified and characterized, 
	in each regime, by exactly $n(n-1)/2$ admissible restrictions represented by $\Ar$. 
\end{defin}

The next definition, instead, considers a more stringent condition for just identification, that, given its particular recursive pattern of restrictions, allows investigating the parameters shock by shock, as already discussed in Theorem \ref{theo:SuffCond}.

\begin{defin}[Recursive just identification]  
\label{def:RecJustIdent}
	A just identified SVAR-WB is also said recursively just identified if there exists a particular ordering of the shocks such that 
	the number of restrictions is $f_{p,j}=n-j$, with $j\in\{1,\ldots,n\}$ and $p\in\{1,\ldots,s\}$, and thus $f_j=s(n-j)$.
\end{defin}

\medskip

A just identified SVAR-WB is characterized by $f=s\,n(n-1)/2$ restrictions, i.e. the minimum required for identification, as stated in Theorem \ref{theo:OrdCond}. Recursive just identification, instead, requires the restrictions to be organized such that each column of $Q_{p,j}$, in each of the $s$ regimes, is characterized by the same number of restrictions $f_{p,j}=n-j$. Clearly, given that we allow for restrictions across regimes, they must be handled with care and recorded in the proper way to satisfy the 
condition of Definition \ref{def:RecJustIdent}. 
%This last notion will be used in the next sections to simplify the estimation procedure.

\vspace{0.5cm}
\noindent
\textit{Examples}
\vspace{0.5cm}

\noindent
The implementation of all the techniques discussed in this section are presented through a set of examples in Appendix \ref{app:examples}. Specifically, in the appendix we show how to impose the restrictions discussed in Section \ref{sec:ident_restr} and how to implement the identification check, depending on the kind of constraints characterizing the SVAR-WB.

%%%%%%%%%%%%%%%%%%%%%%%%%%%%%%%%%%%%%%%%%%%%%%%%%%%%%%%%%%%%%%%%%%%%%%%%%%%%%%%
%														 SECTION ESTIMATION																%
%%%%%%%%%%%%%%%%%%%%%%%%%%%%%%%%%%%%%%%%%%%%%%%%%%%%%%%%%%%%%%%%%%%%%%%%%%%%%%%

\section{Estimating the structural parameters of locally-identified SVAR-WBs}
\label{sec:Estim}

In Section \ref{sec:identification} we have learned that an SVAR-WB is in general at most locally identified, implying that there will be a set of observationally equivalent isolated points on the parametric space, all supporting the equality and sign restrictions. As for locally-identified SVARs, \cite{BacchiocchiFanelli15} propose to estimate the unknown coefficients of SVAR-WBs through maximum likelihood (ML) conditional on some reduced-form parameters. Specifically, the estimation procedure consists in maximizing the likelihood and stopping once the maximum is reached, irrespectively of other observationally equivalent admissible maxima. \cite{BK19} show how such procedure can conduct to misleading results in the case of more than one admissible solution for locally-identified SVARs. The same conclusions naturally apply to the case of SVAR-WBs that, by their nature, are at most locally identified.\footnote{As said before, globally-identified SVARs only occurs when all the parameters are globally identified in each of the regime, precluding however restrictions across-regimes. This case, however, is not interesting for the purposes of the present paper, as the model can be estimated within each regime.}

In this section we propose an algorithm that can be used for estimating all the structural parameters of the model and obtain the identified set $\IS$ as defined in Section \ref{sec:ident_restr}, that represents the set of all admissible objects of interest for the empirical research. The algorithm, rather than focusing directly on the structural parameters, is based on searching for all admissible orthogonal matrices forming the set $\Qr$, and generalizes the algorithm proposed in \cite{BK19} to estimate locally-identified SVARs.  

\subsection{A general algorithm for the estimation of just-identified SVAR-WBs}
\label{sec:Algorithm}

The algorithm presented in this section is very general and allows to estimate all the admissible orthogonal matrices collected in $\Qr$. The strategy used by the algorithm is to solve a non-linear (quadratic) system of equations represented by the equality restrictions, as well as the orthogonality and length restrictions characterizing the columns in the $Q_p$ matrices, $p\in\{1,\ldots,s\}$. Among all the solutions of the non-linear system of equations, some of them are potentially to be discarded according to the normalization and, maybe, inequality restrictions. All these steps are formalized in the next algorithm. 

\begin{algo}
	\label{algo:EstimGen}
	Consider a just identified SVAR-WB with the equality restrictions as in Eq. (\ref{eq:GQ}) and inequality restrictions as in 
	Eq.s (\ref{eq:SignRestr_jp}) and (\ref{eq:fevRest}). Let $\phi=\Btot=\big[B_{1+}\,,\,\Sigma_1\,,\,\cdots\,,\,B_{s+}\,,\,\Sigma_{s}\big]$
	be the reduced-form parameter point.
	\begin{enumerate}
	\item Using the equality restrictions form the system of equations
					\begin{equation}
					\label{eq:SysQ}
						\left\{\begin{array}{rcl}
						\Rj\left(\begin{array}{c}
						G(\Sigma_{1,tr}^{-1\prime}\,,\,B_1^\prime)\:Q_1\\ \vdots\\ G\left(\Sigma_{s,tr}^{-1\prime}\,,\,B_s^\prime\right)Q_s
						\end{array}\right)e_j & = & 0,\hspace{1cm}j\in\{1,\ldots,n\}\\
						Q_1^\prime\,Q_1 & = & I_n\\
						& \vdots &\\
						Q_s^\prime\,Q_s & = & I_n\\
						\end{array}\right.\nonumber
					\end{equation}
	\item solve the previous system for $Q_1,\ldots,Q_s$; 
	\item if the set of all the solutions is non-empty, then retain only those satisfying the inequality restrictions 
				and obtain $\Qr$; Set $M(\phi)$ as the number of admissible solutions, then STOP;
	\item if, instead, the set of all the solutions is empty, then there are no admissible structural parameters associated to $\phi$. 
	\end{enumerate}
\end{algo}

The previous algorithm returns with the non-empty set $\Qr$, containing all the admissible orthogonal matrices related to each of the regimes given the restrictions and the reduced-form parameters $\phi$. The crucial point in the algorithm relates to the way the system in Step 1 is solved. Given the non-linearity provided by the length and orthogonality conditions, the simplest way to find the solutions is to solve it numerically. Matlab, for example, provides the commands \texttt{vpasolve} and \texttt{solve}, that search for all the solutions of a non-linear system of equations as the one in the previous algorithm. Moreover, as guaranteed by the Matlab developers, and to the best of our experience, in the case of systems of polynomial equations, the \texttt{vpasolve} and \texttt{solve} commands return with the complete set of solutions (both real and complex, where these lasts are immediately discarded).

%%%%%%%%%%%%%%%%%%%%%%%%%%%%%%%%%%%%%%%%%%%%%%%%%%%%%%%%%%%%%%%%%%%%%%%%%%%%%%%
%														 SECTION INFERENCE 																%
%%%%%%%%%%%%%%%%%%%%%%%%%%%%%%%%%%%%%%%%%%%%%%%%%%%%%%%%%%%%%%%%%%%%%%%%%%%%%%%

\section{Drawing inference in locally- or set-identified SVAR-WBs}
\label{sec:Inference}

The standard approach for estimating the parameters and doing inference in locally-identified SVAR-WBs is the Gaussian-based maximum likelihood estimator, which also suggests a standard classical inference on the estimated parameters and impulse responses. 
Examples are \cite{ABCF17}, \cite{BCF17}, \cite{Bacchiocchi17OBES} and \cite{BacchiocchiFanelli15} for exogenously-determined regimes, and \cite{PodVel18} for regimes determined endogenously through a Markov-Switching process. However, \cite{BK19} criticize this approach and suggest three alternative ones based on the joint analysis of all the admissible solutions collected in the identified set that, in the case of local identification, is represented by isolated points coherent with the restrictions and compatible with the reduced-form parameters. Specifically, the approaches are: a Bayesian approach, a frequentist-valid approach, and a robust Bayes approach. In the following sections we extend their proposals to locally-identified SVAR-WBs, maintaining the assumption of exogenously determined break dates. Moreover we also propose a classical and a robust Bayes approach for drawing inference on set-identified SVAR-WBs that, to the best of our knowledge, have never been treated in the literature of SVAR-WBs.

%%%%%%%%%%%%%%%%%%%%%%%%%%%%%%%%%%%%%%%%%%%%%%%%%%%%%%%%%%%%%%%%%%
\subsection{Drawing inference in locally-identified SVAR-WBs}
\label{sec:InferenceLocId}

\subsubsection{Bayesian Inference}
\label{sec:BI}

In the standard Bayesian inference, in the presence of local identification, the posterior of the structural parameters and impulse responses can have multiple modes. This aspect can generate computational challenges as the commonly used Markov Chain Monte Carlo (MCMC) methods can fail to well explore the multi-modal posterior. In this respect, \cite{BK19} propose to combine the constructive algorithm for computing $\IS$, our Algorithm \ref{algo:EstimGen} (or Algorithms \ref{algo:EstimRec} and \ref{algo:EstimSeqRec} in Appendix \ref{sec:AlgorithmSeqRec}), with the sampling algorithm for the reduced-form parameters. They show that this approach can get around issues encountered in MCMC applied to the posterior for the structural parameters. 
  
We are interested in approximating the posterior for a scalar impulse response $\eta_p=\eta(\phi,Q_p)$, for the generic \textit{p}-th regime. Conditional on the values of reduced-form parameters yielding nonempty $\Qrp$, let $\ISp$ consist of $M_p(\phi)\geq 1$ number of distinct points, 
\begin{equation}
	IS_{\eta,p}(\phi) = \Big\{\eta\big(\phi,Q_p^{(1)}\big), \eta\big(\phi,Q_p^{(2)}\big), \dots, \eta\big(\phi,Q_p^{(M_p(\phi))}\big) \Big\},
	\label{IS_Mphi}
\end{equation}
where, for simplicity, $\eta\big(\phi,Q_p^{(1)}\big)<\eta\big(\phi,Q_p^{(2)}\big)<\dots<\eta\big(\phi,Q_p^{(M_p(\phi))}\big)$. 

The intuition is to follow the ``agnostic'' Bayesian approach by \citet{Uhlig2005}, where the posterior for $\eta_p$ is induced by the posterior for $\phi$, $\pi_{\phi|Y}$, supported on $\tilde{\Phi} \equiv \big\{\phi: \Qrp \neq \emptyset \big\}$, and the uniform 
prior for $Q_p$ given $\phi \in \tilde{\Phi}$ supported only on the admissible set of rotation matrices $\Qrp$. Specifically, we assign uniform weights over the admissible rotation matrices, implying equal weights assigned over the points in $\ISp$. The posterior for $\eta_p$, thus, can be expressed as:
\begin{equation}
	\pi_{\eta_p | Y}(\eta_p \in A) \propto E_{\phi|Y} \left[ \sum_{m=1}^{M_p(\phi)} \mathbbm{1} \Big\{ \eta(\phi,Q_p^{(m)}) \in A \Big\} \right],
	\label{eta posterior}
\end{equation}
for $A \subset \mathbb{R}$. Since the likelihood of the reduced-form is uni-modal and concentrated around the maximum likelihood estimate, the MCMC sampling algorithms will well perform to get the random draws from $\pi_{\phi|Y}$. Hence, the posterior sampler for $\phi$ combined with the algorithm for computing $\big\{ \eta(\phi,Q_p^{(m)}): m= 1, \dots, M_p(\phi) \big\}$ allow us to approximate the posterior in Eq. (\ref{eta posterior}) reliably by its Monte Carlo empirical analogue.

\subsubsection{Frequentist-valid inference}
\label{sec:FVI}

The Bayesian procedure discussed in the previous section specifies an allocation of the prior belief over observationally equivalent impulse responses that, given $\phi$, is not updated by the data and, as a consequence, the posterior remains sensitive to how it is specified. For locally-identified models \cite{BK19} propose an asymptotically valid frequentist inference procedure for the impulse response identified set that can draw inferential statements robust to the choice of prior weights. 

The idea is to project the asymptotically valid frequentist confidence set for the reduced-form parameters $\phi$ through the identified set mapping $IS_{\eta}(\phi)$.\footnote{See, among others, the 2011 working paper version of \cite{MS12}, \cite{NT14}, \cite{KT13}, for similar approaches where the identified set is a connected interval with positive width.} The projection approach, in general, yields conservative but asymptotically valid confidence sets even when the identified set consists of discrete points. \cite{BK19} also provide a strategy to solve the computational challenge given by the need to compute the projection via the identified set of discrete points based on the finite number of draws or grid points of $\phi$ from their confidence set. 

For the generic \textit{p}-th regime, let $CS_{\phi_p,\alpha}$ be an asymptotically valid confidence set for $\phi_p$ with coverage probability $\alpha \in (0,1)$. If the maximum likelihood estimator $\hat{\phi_p}$ is $\sqrt{T}$-asymptotically normal, the acceptance region of the likelihood ratio test with a critical value set to the $\alpha$-th quantile of the $\chi^2\big(dim(B_{p+})+n(n-1)/2\big)$-distribution provides the likelihood contour set as $CS_{\phi_p,\alpha}^p$.\footnote{As an example, the MCMC confidence set for $\phi_p$ in \cite{CCT18} yields the likelihood contour set with asymptotically valid coverage.} Moreover, if the posterior for $\sqrt{T}(\phi_p - \hat{\phi}_p)$ asymptotically coincides with the sampling distribution of the maximum likelihood estimator, $CS_{\phi_p,\alpha}^p$ can be obtained through the Bayesian highest density posterior region with credibility $\alpha$. 

Now, let $\phi_{(k)}=\big(\phi_{(k),1},\ldots,\phi_{(k),s}\big)\in \mathbb{P}^R$ be a generic draw of $\phi$, and, as in the previous subsection, let $IS_{\eta,p}(\phi_{(k)}) = \Big\{\eta\big(\phi_{(k)},Q_p^{(1)}\big), \eta\big(\phi_{(k)},Q_p^{(2)}\big), \dots, 
\eta\big(\phi_{(k)},Q_p^{\big(M_p(\phi_{(k)})\big)}\big) \Big\}$, with the points ordered for convenience in increasing order and where $M(\phi_{(k)})$ is the number of distinct points in the identified set corresponding to $\phi=\phi_{(k)}$. Finally, let $\bar{M}=\max_k M(\phi_{(k)})$ be the largest cardinality of $IS_{\eta,p}(\phi_{(k)})$ among $k\in\{1, \dots, K\}$. $\bar{M}$ indicates the largest number of disconnected intervals that the projected confidence set, denoted by $CS^p_{\eta,\alpha}$, can have. \cite{BK19} propose to form clusters and check, for each draw $\phi_{(k)}$, $k\in\{1,\dots, K\}$, if any point of $IS_{\eta,p}(\phi_{(k)})$ can be associated with the $\tilde{m}$-th cluster, with $\tilde{m}\in\{1,\ldots,\bar{M}\}$. The main technical difficulty is represented by the approximation of the projected confidence set through a finite number of Monte Carlo draws or grid points from $CS_{\phi_p,\alpha}^p$. In this respect, \cite{BK19} propose two different strategies that can be extended to the present framework in order to obtain, for each regime, frequentist-valid confidence sets for the structural parameters, or more commonly, for the impulse responses. Specifically, given a Monte Carlo draw $\phi=\phi_{(k)}$, the two approaches allow to assign each distinct element in $IS_{\eta,p}(\phi_{(k)})$ to a specific cluster. Once obtained the cluster assignment for every $k\in\{1,\dots, K\}$, we construct, for each cluster $\tilde{m} \in \{ 1, \dots, \bar{M}\}$, an interval $C_{\tilde{m}}^p$ delimiting all the elements within the cluster. We then form an approximate of the projection confidence set by taking the union of all $C_{\tilde{m}}^p$: 
\begin{equation}
	\widehat{CS}_{\eta,\alpha}^p \equiv \bigcup_{\tilde{m}=1}^{\bar{M}} C_{\tilde{m}}^p. \notag 
\end{equation}
The two approaches in \cite{BK19} differ in the way they label the clusters. The \textit{switching-label projection confidence sets} approach allows the labels indexing observationally equivalent impulse responses to vary across horizons, while the 
\textit{fixed-label projection confidence sets} approach maintains unique labels. Both approaches allow to capture multi-modality of the posterior distribution at each horizon.

It is worth nothing that $\widehat{CS}_{\eta,\alpha}^p$ includes all the $IS_{\eta,p}(\phi_{(k)})$'s, $k=1, \dots, K$, and at the same time, it can yield a collection of disconnected intervals. Furthermore, under rather general conditions and denoting with $\phi_0$ the vector of true parameters of the reduced form, it can be shown that $\widehat{CS}_{\eta,\alpha}^p$ converges to $IS_{\eta,p}(\phi_0)$ in the Hausdorff metric, and thus $\widehat{CS}^p_{\eta,\alpha}$ can consistently uncover the true identified set consisting of potentially multiple points.  

Finally, if $\big\{ \phi_{(k)} : k\in\{1,\dots, K\} \big\}$ are draws from the credible region with credibility of the posterior distribution for $\phi$, then $\widehat{CS}_{\eta,\alpha}^p$ can be seen as an approximation of the set $C_{\eta,\alpha}^p$ satisfying
\begin{equation}
 \pi_{\phi_p|Y}\big(IS_{\eta,p}(\phi) \subset C_{\eta,\alpha}^p\big) \geq \alpha. \notag
\end{equation}

%%%%%%%%%%%%%%%%%%%%%%%%%%%%%%%%%%%%%%%%%%%%%%%%%%%%%%%%%%%%%%%%%%%%%%%%%%%%%%%%%
\subsection{Drawing inference in set-identified SVAR-WBs}
\label{sec:SetInference}

In this section we extend the analysis to set-identified SVAR-WBs. To the best of our knowledge, this is completely new in the literature. First of all, we define the kind of set-identified SVAR-WBs allowed, and then move to two possible strategies for doing inference on the parameters or on the identified sets of such models. Specifically, in Section \ref{sec:SetId} we will present a pure Bayesian approach, while in Section \ref{sec:RobBayesSVARWB} we present a robust Bayesian approach.

\subsubsection{Set-identified SVAR-WBs}
\label{sec:SetId}

The kind of set identification we allow in the present paper is an extension of the proposals by \citet{GK18} and \citet{ARW18}. In their setup the starting point is the well known condition in SVARs for global identification, where the number of restrictions $f_j$ must follow the recursive pattern $f_j=n-j$, with $j\in\{1,\ldots,n\}$. As a consequence, set identification arises when, at least for one $j$, $f_j<n-j$. Moreover, \citet{GK18} provide conditions for the restrictions to generate a convex identified set. 

Our notion of set identification for SVAR-WB considers departures from Definition \ref{def:RecJustIdent}, and specifically, reads as follows:

\begin{defin}[Set identification]  
\label{def:SetIdent}
	An SVAR-WB as in Eq. (\ref{eq:SVARWB}) with zero restrictions as in Eq. (\ref{eq:GQ}) and inequality restrictions as in 
	Eq.s (\ref{eq:SignRestr_jp}) and (\ref{eq:fevRest}) is set identified if at least for one $j\in\{1,\ldots,n\}$ or one $p\in\{1,\ldots,s\}$, $f_{p,j}<n-j$,
	and thus $f_j < s(n-j)$.
\end{defin}

In the case of no break, it clearly reduces to departures from the previous \citeauthor{RWZ10RES}'s condition. Moreover, given the stability restrictions that could make an SVAR-WB attractive for empirical researchers, the resulting local identification 
does not justify the search for conditions for convex identified sets, as they will be generally made of disjoint sets. As a consequence, the suggestion by \citet{GK18} of summarizing the identified set through the ``posterior mean bounds'' is extremely
conservative as they act as a hull of the identified set, where, however, some regions are not admissible within. 
 
\subsubsection{Bayesian inference in set-identified SVAR-WBs}
\label{sec:BayesSVARWB}

Here below we provide an algorithm to numerically implement our Bayesian approach. The first step consists in estimating a Bayesian VAR-WB. As the parameters of the reduced form are completely unrestricted across regimes, they can be treated as standard Bayesian VARs within each regime. Thus, once specified the prior distribution of the reduced-form parameters $\pi_\phi$, it becomes standard to obtain the posterior $\pi_{\phi|Y}$. From this posterior, we will obtain the draws of $\phi$. The second step, instead, consists in rotating the reduced-form parameters through the admissible orthogonal matrices, randomly generated from a uniform distribution. 

\begin{algo}
\label{algo:BayesSetId}
Consider a set identified SVAR-WB as in Definition \ref{def:SetIdent}. 
Let $\phi=\Btot=\big[B_{1+}\,,\,\Sigma_1\,,\,\cdots\,,\,B_{s+}\,,\,\Sigma_{s}\big]$ be a generic reduced-form parameter point.
\begin{enumerate}
	\item Specify $\pi_{\phi_p}$, $p\in\{1,\ldots,s\}$. Estimate a Bayesian reduced-form VAR to obtain the posterior $\pi_{\phi_p|Y}$, 
		$p\in\{1,\ldots,s\}$, and thus $\pi_{\phi|Y}$.
	\item Draw $\phi$ from $\pi_{\phi|Y}$.
	\item Set $j=1$. 
  \item Based on the restrictions as in Eq. (\ref{eq:ImpForm}), define the $\big([f_j+s(j-1)]\times ns\big)$ matrix $\Gamma_j$ as 
		\begin{equation}
    \label{eq:Gammaj}
				\Gamma_j=\left(\begin{array}{ccc}
				R_{1,j}^{*}G_1&\ldots&R_{s,j}^{*}G_s\\&\tilde{Q}_1^\prime&\\&\vdots&\\&\tilde{Q}_{j-1}^\prime&\end{array}\right).
		\end{equation}
		If $j=1$, then $\Gamma_j = \big(\begin{array}{ccc}R_{1,j}^{*}G_1&\ldots&R_{s,j}^{*}G_s\end{array}\big)$.
	\item Let $\Gamma_{j\perp}$ be the space of all $(ns\times 1)$ vectors orthogonal to $\Gamma_j$, with 
		$\mathrm{dim}\,(\Gamma_{j\perp})=ns-[f_j+s(j-1)]=\gamma_j>s$. For each $p\in \{1,\ldots,s\}$ find the vector $\tilde{q}_{j,p}\in \Gamma_{j,\perp}$ of the form
		\footnote{Such a vector can be obtained in the following way. Let $B_j$ a basis for the space $\Gamma_{j\perp}$. The generic vector belonging to $\Gamma_{j\perp}$ is of the form
		\begin{equation}
		\label{eq:genVector}
			v_j = B_j\,\lambda_j = 
			\left(\begin{array}{c}\underset{n\times \gamma_j}{B_j^1}\\ \vdots\\ \underset{n\times \gamma_j}{B_j^s}\end{array}\right)\lambda_j.\nonumber
		\end{equation}
		Let $\bar{B}_j^p$ be the $n(s-1)\times\gamma_j$ matrix obtained by removing $B_j^p$ from $B_j$. Then, find $\hat{\lambda}_j$ as a vector orthogonal to the rows of $\bar{B}_j^p$ 
		(for instance, using the Matlab command \texttt{orth}). The vector $\tilde{q}_{j,p}\in \Gamma_{j,\perp}$ will be given by $\tilde{q}_{j,p}\in \Gamma_{j,\perp}=B_j\hat{\lambda}_j$,
		normalized to have unit length.}
		\begin{equation}
		\label{eq:qjpt}
			\tilde{q}_{j,p}=\left(\underset{1\times n}{0},\ldots,\underset{1\times n}{q_{j,p}^\prime},\ldots,\underset{1\times n}{0}\right)^\prime.\nonumber
		\end{equation}
		Update the admissible partial orthogonal matrices
		\begin{equation}
			\label{eq:PartOrtMat}
					Q_{j,p}=\big[Q_{j-1,p}\,|\,q_{j,p}\big]
		\end{equation}
		with $p\in\{\,1\ldots,s\}$. 
		\begin{itemize} 
			\item[a.] If $j=n$, calculate the structural parameters and check for the sign restrictions. If the inequality restrictions are met, add 
					$\big(Q_{j,1},\ldots,Q_{j,s}\big)$ to the set of admissible orthogonal matrices $\Qt$ and add the reduced-form parameters $\phi$
					to the set of $\tilde{\Phi}$. If the inequality restrictions are not met, go back to (Step 2).
			\item[b.] If $j<n$, based on the obtained $Q_{j,p}$ in Eq. (\ref{eq:PartOrtMat}), define the matrices
				\[
					\tilde{Q}_1=\left(\begin{array}{cccc}
					Q_{1,1}&0&\cdots&0\\
					0&Q_{2,1}&\cdots&0\\
					\vdots&\vdots&\ddots&\vdots\\
					0&0&\cdots&Q_{s,1}
					\end{array}\right)
					\hspace{0.3cm}\cdots\hspace{0.3cm}
					\tilde{Q}_{j}=\left(\begin{array}{cccc}
					Q_{1,j}&0&\cdots&0\\
					0&Q_{2,j}&\cdots&0\\
					\vdots&\vdots&\ddots&\vdots\\
					0&0&\cdots&Q_{s,j}
					\end{array}\right),
				\]
				update $j=j+1$ and go back to (Step 4).
		\end{itemize}
	\item Iterate (Step 2)-(Step 5) $N$ times.
	\item If $\Qt=\{\hspace{0.3cm}\}$, then no admissible orthogonal matrix found. Otherwise, $\Qt$ collects all the admissible orthogonal
			matrices for all regimes, and $\tilde{\Phi}$ collects the associated draws for the reduced-form parameters. 
\end{enumerate}
\end{algo}

\begin{remark}
\label{rem:BayesAlg}
	Some remarks are in order. The algorithm, in synthesis, generalizes to SVAR-WBs the original proposal in \cite{Uhlig2005} for sign-restricted 
	SVARs. In this set-up, we extend both in the direction of zero and stability restrictions. In the case of no breaks and no zero restriction
	it reduces to the Uhlig's pure-sign-restriction approach. As already discussed in the previous sections, we exploit the orthogonal 
	reduced-form parametrization, in the \citet{ARW18}'s terminology, to obtain draws from the posterior distribution of the structural parameters.
	However, given the different characteristic of the model with respect to traditional SVARs, as well as the presence of stability restrictions,
	we do not have any analytical results allowing to generate draws from a normal-generalized-normal (NGN) distribution over the structural 
	parametrization conditional to zero and sign restrictions, as in \citet{ARW18}. Rather, Algorithm \ref{algo:BayesSetId} corresponds to a variation
	in their Algorithm 2, allowing for SVAR-WBs characterized by stability restrictions. Providing new analytical results for drawing directly from 
	a NGN distribution is beyond the scope of the present paper, but belongs to our future research agenda.
\end{remark}

The implementation of the algorithm is straightforward and relatively computationally fast. If the object of interest is represented by specific impulse responses, a natural way to present the results could be the plot, for each horizon, of the highest posterior density regions for different credibility levels. This way of reporting the results allows to present potential multimodality in the posterior distribution of the object of interest. This strategy is implemented in the empirical application presented in Section \ref{sec:EmpApp}.

\subsubsection{Robust Bayesian inference in set-identified SVAR-WBs}
\label{sec:RobBayesSVARWB}

In standard set-identified SVARs, Bayesian techniques received criticisms when used for drawing inference on the identified sets for impulse responses. \cite{BH15}, among others, prove the effect of the choice of prior does not disappear asymptotically as instead 
it does in the point-identified case. This drawback clearly remains in the set-identified SVAR-WBs discussed in the present paper.

As an alternative to the standard Bayesian inference, for our set-identified SVAR-WB we consider the robust Bayesian approach of \citet{GK18}, that, being based on multiple priors, does eliminate the drawback caused by the informativeness of the unique selected prior. 
%Originally proposed for a class of set-identified models with identified set $IS_{\eta}(\phi)$ of positive Lebesgue measure in $\mathbb{R}$, it could also be thought for identified sets made of isolated points, with zero Lebesgue measure. In this respect, viewing the set of locally-identified parameter values of the SVAR-WB (or, equivalently, the admissible impulse responses of interest) as the $\IS$, our extension of the robust Bayes inference represents a valid alternative to the two approaches presented in Section \ref{sec:BI} and Section \ref{sec:FVI} for locally-identified SVAR-WBs. See also the recent discussion in \cite{BK19}. 

The first step is the same as the pure Bayesian inference discussed in Section \ref{sec:BayesSVARWB}, and consists in obtaining the posterior $\pi_{\phi|Y}$. The robust Bayesian inference does not specify a conditional prior for $Q$ given $\phi=(B_{1+},\Sigma_1,\ldots,B_{s+},\Sigma_s)$, say $\pi_{Q|\phi}$, but allows for arbitrary probability distributions for it, in order to make the posterior inference free from the choice of $\pi_{Q|\phi}$. Let $\Pi_{Q|\phi}$ denote a collection of conditional priors $\pi_{Q|\phi}$. The class of conditional priors that impose equality and/or inequality restrictions on the SVAR-WB is defined as
\[
	\Pi_{Q|\phi}=\big\{\pi_{Q|\phi}\:\big|\:\pi_{Q|\phi}\big(\Qr\big)=1,\:\pi\text{-almost surely}\big\}.
\]
From the class of prior distributions $\Pi_{Q|\phi}$ it is possible to obtain the class of conditional priors $\pi_{Q_p|\phi_p}$ for the admissible $Q_p$ matrices in each single regime, denoted by $\Pi_{Q_p|\phi}$, $p\in\{1,\ldots,s\}$. For each regime, thus, the posterior for $\phi_p$, combined with the prior class $\Pi_{Q_p|\phi}$, generates the class of joint posteriors for $(\phi_p,Q_p)$
\[
	\Pi_{\phi_p Q_p|Y}=\big\{\pi_{\phi_p Q_p|Y}=\pi_{\phi_p |Y}\pi_{Q_p|\phi_p}\:\big|\:\pi_{Q_p|\phi_p}\in \Pi_{Q_p|\phi_p}\big\}.
\]
Substantially, the main difference with respect to the original approach by \citet{GK18}, thus, is in the way we treat the arbitrary distributions for $Q_p$, that are related across the regimes and must be considered jointly, instead of specifically in each regime as a single SVAR. 

Finally, the class of posteriors $\Pi_{\phi_p Q_p|Y}$ induces the class of posteriors for impulse responses or any transformations of the structural parameters in each regime, $\eta_p=\eta(\phi_p,Q_p)$. Specifically, focusing on the impulse response of interest $\eta_p$, let $A \subset \mathbb{R}$, then
\[
	\Pi_{\eta_p|Y}\equiv\bigg\{\pi_{\eta_p|Y}(A)=\int_{}^{}\pi_{\phi_p Q_p|Y}\big(\eta_p\in A\big)\ud\pi_{\phi_p|Y}\:\bigg|\:
	\pi_{\phi_p Q_p|Y}\in \Pi_{\phi_p Q_p|Y}\bigg\}.
\]
By applying Theorem 1 in \citet{GK18}, the class of posterior distributions $\Pi_{\eta_p|Y}$ can be summarized by defining the posterior lower probability $\pi_{\eta_p|Y \ast} (A)$ and the posterior upper probability $\pi_{\eta_p|Y}^{\ast} (A)$ for any event 
$\{\eta_p \in A \}$. Specifically, the range of posterior probabilities for $\{\eta_p \in A \}$ is given by the convex interval 
\begin{equation}
	\pi_{\eta_p|Y}(A) \in \left[ \pi_{\eta_p|Y \ast} (A), \pi_{\eta_p|Y}^{\ast} (A) \right] \equiv \Big[ \pi_{\phi_p|Y} \big(\ISp \subset A\big), 
	\pi_{\phi_p|Y} \big(\ISp \cap A \neq \emptyset\big) \Big] \label{posterior prob range}.
\end{equation}
%The two extremes of the interval have a very interesting interpretation when considering set identification from the Bayesian perspective. In fact, fixing an hypothesis of interest, say $\{\eta_p \in A \}$, the posterior lower probability states that, independently of the prior distribution considered, the posterior credibility for $\{\eta_p \in A \}$ is at least equal to $\pi_{\eta_p|Y \ast} (A)$. These definitions will play a relevant role in summarizing the results of the robust Bayesian approach and in conducting global sensitivity analysis.

If we want to summarize the information in the posterior class without specifying $A$, \citet{GK18} propose to report the set of posterior means of the identified set. 
%For a scalar impulse response in traditional SVARs, it is straightforward to compute the range of posterior means. 
%Let $\ell(\phi) =\min\{ \eta \in IS_{\eta}(\phi) \}$ and $u(\phi) = \max\{ \eta \in IS_{\eta}(\phi) \}$. Theorem 2 in \citet{GK18} shows that the range of posterior means is given by the Aumann expectation of the convex hull of the identified set $\big[E_{\phi|Y}\big(\ell(\phi)\big), E_{\phi|Y}\big(u(\phi)\big)\big]$, whose empirical counterpart can be obtained by random sample of $\phi$ drawn from $\pi_{\phi|Y}$. 
In the case of SVAR-WBs, the range of posterior means has to be calculated for each regime, $\big[E_{\phi|Y}\big(\ell_p(\phi)\big), E_{\phi|Y}\big(u_p(\phi)\big)\big]$, $p\in\{1,\ldots,s\}$, taking into account, however, that the stability constraints make the regimes interconnected.\footnote{The stability and inequality restrictions are taken into account in the calculation of the admissible $Q_p$ matrices.}

An alternative strategy to summarize the results is represented by the robust credible regions, that can be read as the robust Bayesian counterpart of the highest posterior density region generally reported in standard Bayesian inference. For our SVAR-WB, 
fixing $\alpha\in (0,\,1)$, in the generic $p$-th regime, we focus on the set $C_{\alpha,p}$ such that the posterior lower probability associated to this set is at least equal to $\alpha$, i.e.
\[
	\pi_{\eta_p|Y \ast}(C_{\alpha,p})=\pi_{\phi_p|Y}\big(\ISp\subset C_{\alpha,p}\big)\geq \alpha
\]
whatever the choice of posterior within the class $\Pi_{\eta_p|Y}$. 
%In general, obtaining $C_{\alpha,p}$ can be problematic when the object of interest is a vector. However, \citet{GK18} show that if $\eta$ is a scalar and the class of admitted sets (reducing to intervals in the specific case) consists of closed connected intervals, $C_{\alpha,p}$ can then be computed by solving a simple optimization problem.

Given the characteristics of SVAR-WBs constrained by stability restrictions, deriving results about the convexity of the identified sets is much more complicate than in SVARs featuring only zero or sign restrictions (see \citeauthor{GK18}, \citeyear{GK18}). As a consequence, it is reasonable to expect the identified set to be potentially not convex. However, \citet{GK18} provide two important results for the set of posterior means and robust credible regions from the frequentist perspective, even in the case of non-convex identified sets. Firstly, they show that the set of posterior means converges to the hull of the true identified set. Secondly, the robust credible region has the correct asymptotic coverage for the true identified set, regardless the identified set to be convex or not.

%Clearly, when the attention is on impulse responses,
%the set $C_{\alpha,p}$ can be given by disconnected intervals.
%The practical implementation of the robust credible region consists in calculating such probability for each draw $\phi_p$ from 
%the posterior distribution of the reduced-form parameters and then calculate the union of the obtained sets. 
%In the case of an impulse response in the $p$-th regime, for each draw $\phi_p$ we will obtain the intervals $C_{\alpha,p}(\phi_p)$ such that 
%$\pi_{\phi_p|Y}\big(\ISp\subset C_{\alpha,p}(\phi_p)\big)\geq \alpha$. The robust credible region $C_{\alpha,p}$ will be given
%by the union of all the intervals obtained for each $\phi_p$, i.e.
%\begin{equation}
%\label{eq:RobCredReg}
%	C_{\alpha,p}=\left\{\bigcup_{\phi_p}C_{\alpha,p}(\phi_p)\:\Bigg|\pi_{\phi_p|Y}\big(\ISp\subset C_{\alpha,p}(\phi_p)\big)\geq \alpha,
%	\:\:\pi_{\phi_p|Y}\in \Pi_{\phi_p|Y}\right\}
%\end{equation}	

The next algorithm allows to numerically approximate the set of posterior means and the robust credible regions for set identified SVAR-WBs.

\begin{algo}
\label{algo:RobSetId}
Consider a set identified SVAR-WB as in Definition \ref{def:SetIdent}. 
Let $\phi=\Btot=\big[B_{1+}\,,\,\Sigma_1\,,\,\cdots\,,\,B_{s+}\,,\,\Sigma_{s}\big]$ be a generic reduced-form parameter point. 
Moreover, for each regime $p\in\{1,\ldots,s\}$, let $\eta_p = c_{p,ih}^\prime(\phi) q_{p,j^*}$ be the impulse response of interest.
\begin{enumerate}
	\item Specify $\pi_{\phi_p}$, $p\in\{1,\ldots,s\}$. Estimate a Bayesian reduced-form VAR to obtain the posterior $\pi_{\phi_p|Y}$, 
		$p\in\{1,\ldots,s\}$, and thus $\pi_{\phi|Y}$.
	\item Draw $\phi$ from $\pi_{\phi|Y}$. Given the draw of $\phi$, check whether the admissible matrix $Q$ obtained through 
		(Step 3)-(Step 5) of Algorithm \ref{algo:BayesSetId} satisfies the inequality restrictions. If so, then retain this $Q$ and go to (Step 3). 
		Otherwise, repeat (Step 3)-(Step 5) of Algorithm \ref{algo:BayesSetId} a maximum of $N$ times (e.g. N = 3000) or until $Q$ is obtained 
		satisfying the inequality restrictions. If none of the $L$ draws of $Q$ satisfies the inequality restrictions, approximate $\Qr$ as being empty 
		and return to Step 2 to obtain a new draw of $\phi$.
	\item Given $\phi$ obtained in (Step 2), iterate (Step 3)-(Step 5) of Algorithm 
		\ref{algo:BayesSetId} $L$ times and let $\Big(Q^{(l)}\in \Qr:\,l=1,\ldots,\tilde{L}\Big)$ be the draws that satisfy the 
		inequality restrictions. Thus, for each $p\in\{1,\ldots,s\}$, let $q_{p,j^*}^{(l)}$, $l=1,\ldots,\tilde{L}$
		be the $j^*$-th column vector of $Q_p^{(l)}$. Approximate $\big[\ell_{p}(\phi);\, u_{p}(\phi)\big]$ by 
		$\big[\mathrm{min}_l \:c_{p,ih}^\prime(\phi) q_{p,j^*}^{(l)}\:,\:\mathrm{max}_l \:c_{p,ih}^\prime(\phi) q_{p,j^*}^{(l)}\big]$.
	\item Repeat (Step 2)-(Step 3) $N$ times to obtain, for each $p\in\{1,\ldots,s\}$,
		$\big[\ell_{p}(\phi_n);\, u_{p}(\phi_n)\big]$, $n=1,\ldots,N$. Approximate the set of posterior means by the sample averages of
		$\big(\ell_{p}(\phi_n);\: n=1,\ldots,N\big)$ and $\big(u_{p}(\phi_n);\: n=1,\ldots,N\big)$. 
	\item To obtain an approximation of the smallest robust credible region with credibility $\alpha\in(0,\,1)$, define, for each $p\in\{1,\ldots,s\}$, 
		$d(\eta_p,\phi)=\mathrm{max}\big\{|\eta_p-\ell_p(\phi)|\,,\, |\eta_p-u_p(\phi)|\big\}$, and let $\hat{z}_\alpha(\eta_p)$ be the sample 
		$\alpha$-th quantile of $\big(d(\eta_p,\phi_n)\,:\, n=1,\ldots,N\big)$. An approximated smallest robust credible region for $\eta_p$
		is an interval centered at $\mathrm{arg\:\:}\mathrm{min}_{\eta_p}\:\hat{z}_\alpha(\eta_p)$ with radius 
		$\mathrm{min}_{\eta_p}\:\hat{z}_\alpha(\eta_p)$.
	\item The proportion of drawn $\phi$'s that pass (Step 2) is an approximation of the posterior probability of having a nonempty 
		identified set.
\end{enumerate}
\end{algo}

\begin{remark}
\label{rem:RobBayesAlg}
	Some remarks are in order. First, the present algorithm adapts Algorithms 1 and 2 in \citet{GK18} to our SVAR-WB. The way we draw orthogonal $Q$ matrices subject to zero and sign restrictions in (Step 2) is also common to \citet{ARW18}.
	Second, in \citet{GK18} the optimization step (Step 3) is a non-convex optimization problem whose argument is a scalar impulse response. For this reason, in their Algorithm 1 they propose to solve a non-linear minimization (maximization) problem with respect 
	to the orthogonal matrix $Q$. In our framework, instead, the minimization (maximization) problem is no-longer a scalar, as it should be solved jointly for the impulse responses of all the regimes. For this reason, we substitute the non-linear optimization 
	problem by a numerical search. This approach, although potentially less precise, should guarantee much more stability, as convergence of gradient-based optimization methods, based on all the orthogonal matrices $Q_1,\ldots,Q_s$, is even more problematic than in 
	\citeauthor{GK18}'s Algorithm 1. Our solution is in line with their Algorithm 2 and exactly reduces to this latter in the case of no breaks. 
	As this alternative way of deriving the bounds of the posterior mean interval is not based on a minimization or maximization strategy, a drawback is that it could provide a smaller $\ISp$ for each draw $\phi$. However, as the number of admissible $\Qr$ 
	matrices $\tilde{L}$ goes to infinity, the algorithm provides consistent estimator for the identified set at each regime. Finally, as for standard SVARs, as originally proposed by \cite{BK19}, the robust Bayesian approach can be implemented also for 
	locally-identified SVAR-WBs, where the identified set is represented by a finite number of isolated elements in $\IS$, instead of a set of positive Lebesgue measure.
\end{remark}

%%%%%%%%%%%%%%%%%%%%%%%%%%%%%%%%%%%%%%%%%%%%%%%%%%%%%%%%%%%%%%%%%%%%%%%%%%%%%%%
%												 SECTION EMPIRICAL APPLICATION  											%
%%%%%%%%%%%%%%%%%%%%%%%%%%%%%%%%%%%%%%%%%%%%%%%%%%%%%%%%%%%%%%%%%%%%%%%%%%%%%%%

\section{Empirical application}
\label{sec:EmpApp}

Among many others, \cite{BG06RESTATS} provide empirical evidence and theoretical explanations of differences in the conduct of the monetary policy by the Fed during the Great Moderation with respect to the previous decades. Within this framework we propose some 
empirical examples using SVAR-WBs as those developed in the present paper.

Our empirical analysis is based on the vector $y_{t}=(\tilde{y}_{t},\pi _{t},R_{t})^{\prime }$ ($n$=3), where $\tilde{y}_{t}$ is a measure of the output gap, $\pi _{t}$ the inflation rate and $R_{t}$ a nominal policy interest rate.\footnote{The measure of real activity, $\tilde{y}_{t}$, is the Congressional Budget Office (CBO) output gap, constructed as percentage log-deviations of real GDP with respect to CBO potential output. The measure of inflation, $\pi _{t}$, is the annualized quarter-on-quarter GDP deflator inflation rate, while the policy instrument, $R_{t}$, is the Federal funds rate (average of monthly observations). The data were collected from the website of the Federal Reserve Bank of St. Louis.} The observations are quarterly over the sample 1954.q3-2008.q3. 
Coherently with \cite{BG06RESTATS}, we divide the postwar period 1954.q3-2008.q3 into two sub-samples: the `pre-Volcker' period, 1954.q3-1979.q2, and the `Great-Moderation' period, 1979.q3-2008.q3. We consider, thus, one structural break and two regimes. The reduced form is specified by six lags for both the first and second regime.\footnote{\cite{BacchiocchiFanelli15} provide strong statistical evidence that the two sub-periods can be regarded as two periods characterized by different volatility and different dynamics in the reduced-form parameters.} 

In order to show how our methodology works we propose five alternative SVAR-WBs, characterized by different identification schemes, combining both equality and inequality restrictions. Model I and Model II are characterized by the same equality restrictions. The only difference is that the latter also presents a set of sign restrictions. 

\vspace{0.7cm}
\noindent
\textit{Equality restrictions: zero and stability restrictions}
\vspace{0.4cm}
%\begin{itemize}
%\itemsep0em 
%\item monetary policy shocks (MP) have no contemporaneous effect on inflation in the first regime;
%\item demand shocks (D) have no contemporaneous effect on inflation in the two regimes;
%\item supply shocks (S) have the same contemporaneous effect on inflation in the two regimes;
%\item supply shocks (S) have the same contemporaneous effect on output gap in the two regimes;
%\item demand shocks (D) have the same contemporaneous effect on output gap in the two regimes.
%\end{itemize}

\noindent
The transformation functions and the identifying restrictions for both the regimes are represented in Table \ref{tab:models}, upper panel. The symbol ``$\times$'' stands for unrestricted coefficient, while ``$\encircled{\times}$'' (in different colors) means that the two coefficients are restricted to be equal. The stability restrictions considered in this empirical application refer to the response of inflation to a demand shock and the response of output gap to both a demand and a supply shock, that contribute identifying the model, while allowing both the reaction function of the Fed and the effect of monetary policy shock on output to be regime-specific. 
%\footnote{In this respect, in order to guarantee the same magnitude of a monetary policy shock in the two regimes, the stability restriction is imposed by standardizing for the ratio between the standard deviations of the residuals of the interest rate equations in the two regimes.} 
Moreover, in line with many contributions based on a Cholesky identification scheme, like \cite{BG06RESTATS} and \cite{CEE05JPE}, we assume the monetary policy shock to influence inflation only after one quarter. This zero restriction, however, is limited to the first regime. This choice is justified by the intention to capture the larger efforts made by Governor Paul Volcker and the Federal Open Market Committee, the policy making committee of the Fed, starting in October of '79 to break the inflation cycle of the great inflation period. Finally, the last two zero restrictions assume a lag in the transmission of demand shocks to price dynamics, as a consequence of menu costs and other forms of price rigidities. Interestingly, in both regimes we allow the monetary policy shocks to have a non-zero contemporaneous effect on output as suggested by many DSGE models.\footnote{See, among others, \cite{DS04} and \cite{DSSW07}.} This characteristic is precluded in recursive identification schemes \textit{\`{a} la} \cite{CEE05JPE} and \cite{BG06RESTATS}. 

First of all, the necessary order condition for identification is met. Specifically, given the particular pattern of the restrictions, it is possible to see that the SVAR-WB in Model I does not meet the requirements of Corollary \ref{corol:OrderCond} and, as a consequence, identification is guaranteed by the rank condition in Theorem \ref{theo:NecSuff}. Finally, differently from standard SVARs, the identification benefits of the stability restrictions across the two regimes. 

\begin{table}
	\centering
	\caption{Restrictions for Model I to Model IV.}
	\label{tab:models}
		\begin{tabular}{lrccc}
		  \hline\hline
		  &&& \text{first regime} & \text{second regime}\\
			&&& $G_1\equiv G(A_{10}^\prime,A_{10}^\prime)$ & $G_2\equiv G(A_{20}^\prime,A_{20}^\prime)$\\
			\hline
			&&& $\begin{array}{ccc} \text{MP} & \text{S} & \text{D} \end{array}$ &  $\begin{array}{ccc} \text{MP} & \text{S} & \text{D} 
			\end{array}$ \\
			\text{model I}: & $\left[IR_{0}\right] = $ &
			$\begin{array}{c} i_t \\ \pi_t \\ \tilde{y}_t \end{array} $ & 
 			$\left[\begin {array}{ccc} \times&\times&\times\\ 0&$\Circled[outer color=green]{$\times$}$&0 \\\times&$\Circled{$\times$}$&$\Circled[outer color=red]{$\times$}$ \end{array}\right]$&
			$\left[\begin {array}{ccc} \times&\times&\times\\ 0&$\Circled[outer color=green]{$\times$}$&\times \\ \times&$\Circled{$\times$}$&$\Circled[outer color=red]{$\times$}$ \end{array}\right]$\\
			\hline
			\rule{0pt}{2ex}\\
			\text{model II}: & \multicolumn{4}{c}{model I + sign restrictions}\\ 
			\rule{0pt}{2ex}\\
			\hline
			&&& $\begin{array}{ccc} \text{MP} & \text{S} & \text{D} \end{array}$ & $\begin{array}{ccc} \text{MP} & \text{S} & \text{D} 
			\end{array}$\\
			\text{model III}: & $\left[IR_{0}\right] = $ &
			$\begin{array}{c} i_t \\ \pi_t \\ \tilde{y}_t \end{array} $ & 
 			$\left[\begin {array}{ccc} \times&\times&\times\\ \times&$\Circled[outer color=green]{$\times$}$&\times \\ \times&\times&$\Circled[outer color=red]{$\times$}$ \end{array}\right]$&
			$\left[\begin {array}{ccc} \times&\times&\times\\ \times&$\Circled[outer color=green]{$\times$}$&\times \\ \times&\times&$\Circled[outer color=red]{$\times$}$ \end{array}\right]$\\
			& \multicolumn{4}{c}{+ sign restrictions}\\ 
			\hline
			\rule{0pt}{2ex}\\
			\text{model IV}: & \multicolumn{4}{c}{$
			\begin{array}{lr}
			\text{equality restrictions:}& \text{NO}\\ 
			\text{sign restrictions:}& \text{YES}\\
			\text{ranking restrictions:}& \text{NO}\\
			\text{FEV restrictions:}& \text{NO}
			\end{array}$}
			\rule{0pt}{2ex}\\
			\hline
			\rule{0pt}{2ex}\\
			\text{model V}: & \multicolumn{4}{c}{$
			\begin{array}{lr}
			\text{equality restrictions:}& \text{NO}\\ 
			\text{sign restrictions:}& \text{YES}\\
			\text{ranking restrictions:}& \text{YES}\\
			\text{FEV restrictions:}& \text{YES}
			\end{array}$}
			\rule{0pt}{2ex}\\
			\hline\hline
		\end{tabular}
\end{table}

About the estimation and inference, given the reduced-form parameters $\phi$, the admissible orthogonal matrices for both the regimes can be estimated through Algorithm \ref{algo:EstimGen}. Being Model I locally identified, more than one admissible orthogonal matrix is expected to satisfy the restrictions and be coherent with the reduced-form parameters. In this respect, the standard strategy to focus on just one single admissible orthogonal matrix, as in \cite{BacchiocchiFanelli15}, could lead to misleading results. The inference on the parameters of Model I, thus, has been produced by implementing the three procedures presented in Section \ref{sec:Inference}.

We focus on the response of output gap to a restrictive monetary policy shock in the two regimes for each of the presented models. Concerning Model I, the responses are reported in Figure \ref{fig:ModelI}. In both panels, with different nuances of gray, we show the highest $90\%$, $75\%$, $50\%$, $25\%$ and $10\%$ quantiles of the posterior distribution of the identified set (impulse response of output gap to a monetary policy shock) for each horizon $h$.\footnote{The posterior distribution of the identified set is obtained through a Kernel smoothing function calculated over $100$ points (Matlab command \textit{ksdensity}). Increasing the number of points provides practically indistinguishable results.} The frequentist-valid inference is obtained by retaining the $90\%$ draws of the reduced-form parameters with highest value of the posterior density function.\footnote{Very similar results have been obtained through the maximization of the value of the likelihood function or the AIC and BIC information criterion. All the results can be obtained from the authors upon request.} The results are shown as the intervals delimited by the dotted-circle lines. Finally, we also show the results of the robust Bayesian approach in terms of the set of posterior means (dotted lines) and the upper and lower bounds of the robust credible regions with credibility $90\%$ (solid lines). 

The first interesting result is that, once considering both equality and normalization restrictions, the distribution of the effect in both the regimes is clearly bi-modal, at least for the first horizons after the shock. Although the probability mass of the impulse responses is mainly concentrated on the negative responses of the output gap to a restrictive monetary policy, there is also evidence of unexpected potential positive responses. The robust credible bounds, as expected, do represent a sort of hull for the standard Bayesian results, being the two robust credible regions very close to the upper and lower bounds of the highest $90\%$ probability of the posterior distribution of the identified set. In all cases, the frequentist-valid inference reveals to be extremely conservative and produces wide intervals. The other two approaches, pure Bayesian and robust Bayesian, instead, produce much more interesting results that emphasize some of the theoretical findings proposed in the paper. 

As already mentioned in the previous sections, the number of admissible solutions can be potentially reduced by the inclusion of sign restrictions. In this respect, Model II has the same equality restrictions as Model I, but includes some very mild sign restrictions, that are common to Model III to Model V, too. Specifically, the restrictions we consider, for both the regimes, are:

\vspace{0.7cm}
\noindent\textit{Sign restrictions in both first and second regime}
\begin{itemize}
\itemsep0em 
	\item interest rate responds non-negatively to monetary policy shocks (MP) on impact;
	\item interest rate responds non-negatively to supply shocks (S) on impact;
	\item interest rate responds non-negatively to demand shocks (D) on impact;
	\item inflation responds non-positively to monetary policy shocks (MP) on impact and after one quarter ($h=0,1$);
\end{itemize}

The first three sign restrictions are in line with all standard reaction functions of the central bank to demand and supply (price) shocks. The responses for output gap related to Model II are reported in Figure \ref{fig:ModelII}. Although Models I and II feature the same set of equality restrictions, imposing sign restrictions in the SVAR-WB, easily justifiable by theoretical grounds and commonly used in applied macroeconomic contributions, enables us to rule out one of the two admissible solutions obtained by the equality restrictions alone, and returns a completely different scenario. The reason is simply due to the local nature of the identification. Sign restrictions help discarding very unlikely responses, that uselessly contributed to widen the confidence bands of impulse responses. 

In Model III we relax some of the equality restrictions considered in Model I and Model II, and move to the set-identified SVAR-WB. In particular we maintain all the sign restrictions of Model II, while reduce the equality restrictions to just two stability constraints, as shown in Table \ref{tab:models}, third panel. Specifically, in Model III we relax the assumption of full price rigidity in both the regimes, as well as the zero restriction on the response of inflation to a monetary policy shock in the great inflation period and the stability restriction of the output gap to a supply shock. The results, in terms of the response of output gap to a monetary policy shock, are reported in Figure \ref{fig:ModelIII}. Although the model becomes set identified, rather than locally identified, the results show a significant recessionary reaction of output in the second regime, while becomes insignificant in the first. 

Some recent literature provides evidence of instabilities on the effect of demand and supply shocks on inflation and the business cycle (see, for instance, \citeauthor{DLPT20}, \citeyear{DLPT20}). If this represents a push for the relevance of our SVAR-WBs, on the other side it could be viewed as a deterrent to impose stability restrictions as those used in Model I to III. In this respect, however, our machinery continues to offer interesting tools consisting in imposing sign, ranking and FEV restrictions, within and across the regimes.

In this respect, in Model IV and V, to be more agnostic, we completely abandon equality (stability and zero) restrictions in favor of less stringent inequality restrictions, both on impulse responses and FEVs. Specifically, Model IV serves as a reference, where we only impose sign restrictions within each regime, as in the spirit of traditional SVARs. The results are shown in Figure \ref{fig:ModelIV}. Model V, instead, incorporates inequality constraints within and across the regimes, in terms of ranking and FEV restrictions.

\vspace{0.7cm}
\noindent\textit{Ranking restrictions in both first and second regime}
\begin{itemize}
\itemsep0em 
	\item on impact response of interest rate to demand shock (D) larger in the second regime;
	\item on impact response of interest rate to supply shock (S) larger in the second regime;
	\item on impact response of inflation to demand shock (D) larger in the first regime.
\end{itemize}

\vspace{0.7cm}
\noindent\textit{FEV restrictions in both first and second regime}
\begin{itemize}
\itemsep0em 
	\item FEV of output gap to demand shock (D) larger than all the other shocks in both regimes at $h=0,1$;
	\item FEV of inflation to supply shock (S) larger than all the other shocks in both regimes at $h=0,1$;
	\item FEV of interest rate to monetary policy shock (MP) larger than all the other shocks in both regimes at $h=0,1$;
	\item FEV of inflation to demand shock (D) larger in the first regime at $h=0$.
\end{itemize}

The first two ranking restrictions are coherent with the view, largely recognized by many economists, of a more aggressive and credible monetary policy to fight inflation during the great moderation.\footnote{See, for example, Lecture 2 by Nobel Prized and Fed Governor Ben Bernanke on ``The Federal Reserve after World War II'' (\citeauthor{B12}, \citeyear{B12}).} The last ranking restriction, instead, is inspired by some recent findings on the flattening of the Phillips curve, as advocated, among others by \cite{HMS20RinE}. According to this evidence, we impose a larger reaction of inflation to demand shocks during the great inflation than the great moderation. 

About the FEV restrictions, the first three are within regimes, and read as the Max Share Identification approach for a multiplicity of shocks, as deeply discussed and formalized by \cite{CV22}. Specifically, each of the three shocks is expected to explain, at least for $h=0,1$, the largest part of FEVs of one of the three variables. Finally, the last FEV restriction is a constraint across regimes and is totally in line with the last ranking restriction, but in terms of the FEVs, rather than of the impulse responses. Model V includes all these restrictions, jointly with the traditional sign restrictions. The results are reported in Figure \ref{fig:ModelV}.
%Finally, in Model IV, other than removing the full rigidity of prices in the second regime, we relax the two zero restrictions that allowed to identify a supply shock in Model I and Model II. The new results are shown in Figure \ref{fig:ModelIV}. As expected,  the response of output becomes much more uncertain, but seems to go in the same direction as the evidence obtained for the previous models.

Comparing the results of all the models, it clearly emerges that for the locally-identified specification in Model II (with sign restrictions to remove unreasonable results) a restrictive monetary policy shock produces, in both the regimes, a significant negative response of the output gap, a bit more pronounced during the great inflation, see Figure \ref{fig:ModelII}. However, when removing some of the equality restrictions, leading to a set-identified SVAR-WB, the evidence changes. While the monetary policy shock is recessionary in the second regime, it becomes highly non significant in the first regime, as shown in Figure \ref{fig:ModelIII}. Such results are partially confirmed when using only sign restrictions within each regime, that wouldn't justify the joint analysis of all regimes, as instead advocated in this paper (Model IV, Figure \ref{fig:ModelIV}). In this respect, however, it is worth noting the rather large robust credible regions, that differently from the Bayesian analysis, would suggest a non significant response of output to a monetary policy shock in the two regimes. Model V, instead, emphasizes the potential benefits of our methodology when combining traditional sign restrictions within each regime, to inequality constraints, in terms of ranking and FEV restrictions, across the regimes. In fact, as shown in Figure \ref{fig:ModelV}, while the Bayesian approach produces similar results than Model IV, including ranking and FEV restrictions across the regimes allows to enormously reduce the width of the robust credible regions for the impulse responses. Definitely, using our methodology, especially for set-identified SVAR-WBs, there is strong evidence of a different impact of monetary policy shock on output in the two regimes. 

Providing a detailed and thorough investigation on the US monetary policy transmission mechanism is beyond the scopes of the present empirical application, that mainly serves as a guideline for describing the implementation and highlighting the potentiality of our new methodology. A more sophisticated application, as in the spirit of \cite{SimsZha06AER} and \cite{BG06RESTATS} is in our research agenda. Some comments and comparisons, however, are worth stressing. 

\cite{SimsZha06AER} propose a Markov-switching SVAR with both structural coefficients and variances to be potentially time-varying over the different regimes. Interesting, they also allow some of the parameters to remain unchanged, without, however, treating these as restrictions to be used in the identification issue. They propose a set of different specifications and find that the best fit occurs with time variation in the variances of the disturbances only, without, however, closing the doors to a possible regime-specific monetary policy reaction function. Their evidence about a different conduct of the monetary policy over time, however, is not as strong as in \cite{BG06RESTATS}. 

Related to this point, the main difference with respect to \cite{SimsZha06AER} is about the definition of the reaction function of the Fed. As a consequence of their identification strategy, the behavior of the Fed is modeled through a negative relation between the policy rate and a monetary aggregate (M2 divisia). Our identification strategy, instead, allows for a potential immediate reaction, although with different strength over the regimes, to demand and supply shocks. 

Figure \ref{fig:ModelV_DS} reports, for the agnostic identification scheme of Model IV, the systematic response of the Fed, in terms of the policy rate, to these two kinds of shocks during the great inflation and great moderation periods. In the upper panels we show that effectively the response of the Fed to a price shock is immediate, especially during the great inflation period. Furthermore, the response of the Fed is not limited to an immediate reaction but is long-lasting, particularly during the great moderation. Moving the attention to the response of the Fed to demand shocks, in the lower panels of Figure \ref{fig:ModelV_DS} we highlight the differences, if any, in the conduct of the monetary policy in the two regimes. In both cases the Fed significantly responds positively to inflationary demand shocks, and shows, substantially, a similar conduct in the two regimes. Definitely, if the Fed contributed to maintain more stable inflation and output, it seems to be due to the more aggressive and long-lasting response to price shocks during the great moderation regime. Such evidence is even stronger for the identification scheme of Model V.

Another point that we want to emphasize is about the potential response of output to a monetary policy shock. Our identification strategy in Model II, that benefits of the restrictions across the regimes, leaves completely unrestricted the reaction of the economic activity to unexpected monetary policy interventions, that is precluded both in \cite{BG06RESTATS} and \cite{SimsZha06AER}. Focusing on both panels in Figure \ref{fig:ModelII}, we show that in both the regimes the response is recessionary and significant on impact and for seven/eight quarters after a monetary squeeze. 

As already said, \cite{BG06RESTATS} also have evidence of different effects of monetary policy on the economy. The results reported in Figures \ref{fig:ModelII} to \ref{fig:ModelV} are, in some sense, in line with their findings. However, differently from these authors, our results show a more pronounced response of output to monetary policy shocks during the great moderation. \cite{BG06RESTATS}, who obtain their results under the rather strong assumption of a Cholesky identification scheme in both regimes, explain their findings by stating that ``monetary policy has more successfully managed to moderate the effects of exogenous disturbances since the early 1980s, possibly by systematically responding more decisively to fluctuations in economic conditions'' (\citeauthor{BG06RESTATS}, \citeyear{BG06RESTATS}, page 446). In this respect, our view is essentially the same, but we change the perspective in interpreting this thought. In fact, we completely agree with the view of a different reaction function of the central bank in the two regimes, but if the Fed \textit{systematically} responds more aggressively to demand and supply shocks since the early 1980s, an unexpected monetary policy shock should have more pronounced effects on the real economy, being this last shock orthogonal to demand and supply shocks. In other words, if the Feb has gained more credibility since Paul Volcker started his mandate, economic agents do expect the Fed to automatically react to both demand and supply shocks.\footnote{See \cite{CGG00QJE}, \citeauthor{CS05RED} (\citeyear{CS01NBER}, \citeyear{CS05RED}), and \cite{Boivin06JMCB}, for further evidence on the more aggressive response of the Fed to inflation since the early 1980s.} Short term interest rate will incorporate these changes according to this \textit{known} reaction function. What remains in the hands of the Fed, is thus a more powerful monetary instrument that should affect more prices and output. This alternative interpretation is completely in line with our findings.

\begin{figure}[H]
        \caption{Impulse response functions of output gap to a monetary policy shock: Model I (locally-identified SVAR-WB: equality restrictions, only).}
  \label{fig:ModelI}
\begin{center}
                \subfigure[{Model I: first regime}]{
                \includegraphics[angle=270,origin=c,scale=0.24]{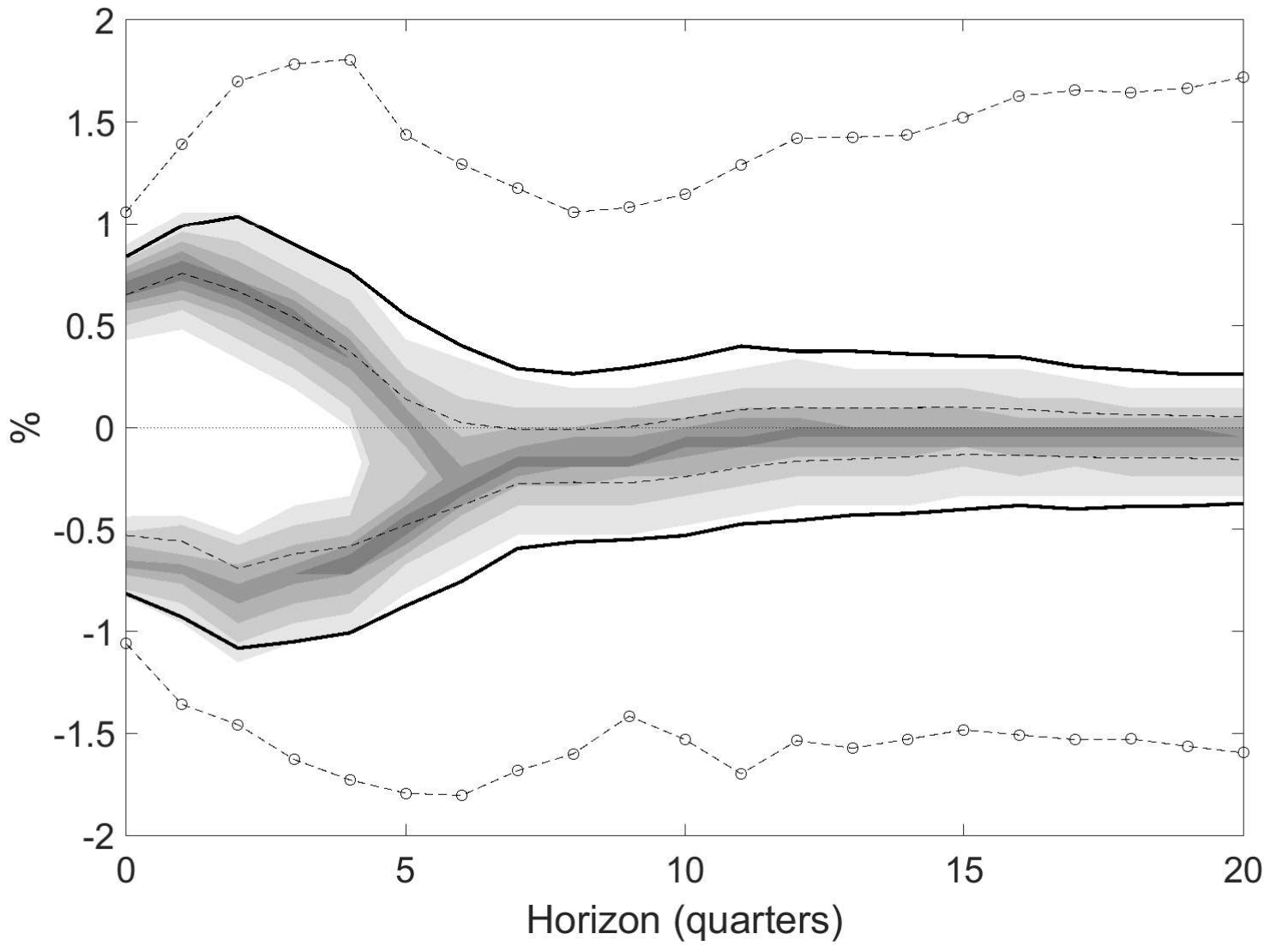}
                }
                \subfigure[{Model I: second regime}]{
                \includegraphics[angle=270,origin=c,scale=0.24]{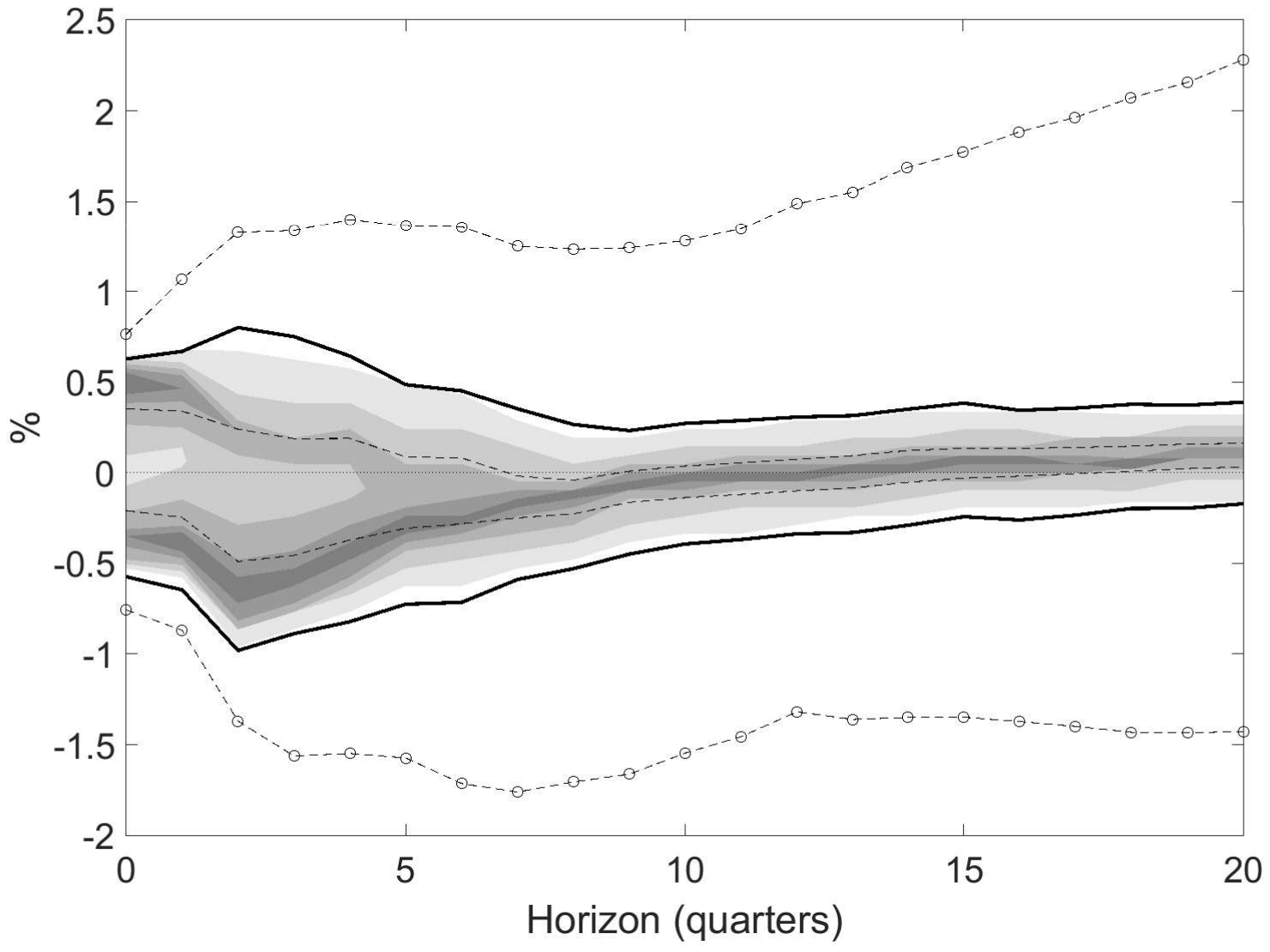}
                }
\end{center}
\begin{minipage}{\textwidth}%
{\scriptsize{\textit{Notes}: In both figures, in different scales of gray, we show the highest $90\%$, $75\%$, $50\%$, $25\%$ and $10\%$ quantiles of the distribution, for each horizon $h$, obtained through the Bayesian approach. Frequentist-valid intervals, obtained by retaining the $90\%$ draws of the reduced-form parameters with highest value of the posterior density function, are shown as the intervals delimited by the dotted-circle lines. In dotted lines we report the set of posterior means while in solid lines the upper and lower bounds of the robust credible regions with credibility $90\%$ obtained through the robust Bayesian approach.\par}}
\end{minipage}
\end{figure}

\begin{figure}[H]
        \caption{Impulse response functions of output gap to a monetary policy shock: Model II (locally-identified SVAR-WB: equality and sign restrictions).}
  \label{fig:ModelII}
\begin{center}
                \subfigure[{Model II: first regime}]{
                \includegraphics[angle=270,origin=c,scale=0.24]{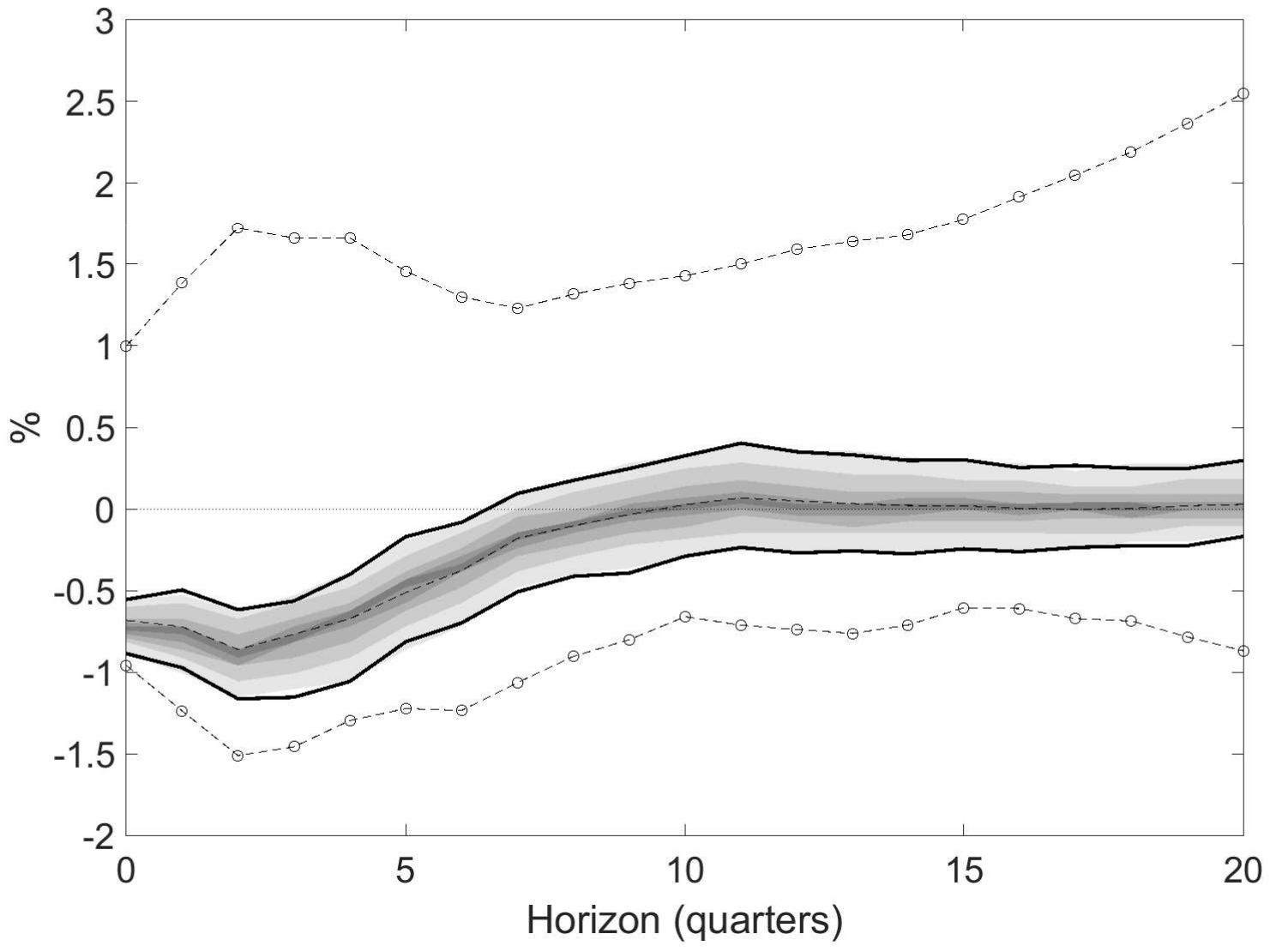}
                }
                \subfigure[{Model II: second regime}]{
                \includegraphics[angle=270,origin=c,scale=0.24]{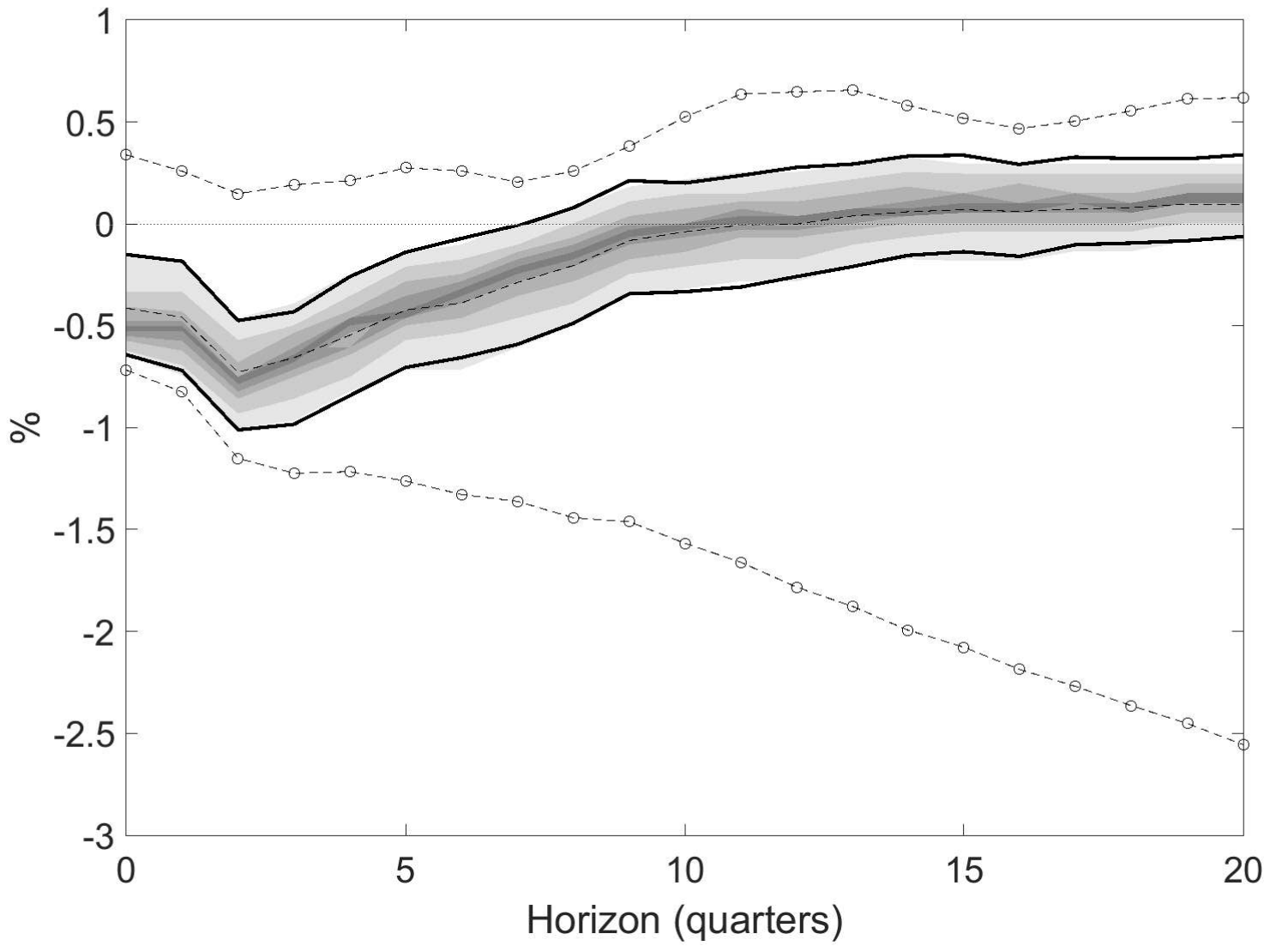}
                }
\end{center}
\begin{minipage}{\textwidth}%
{\scriptsize{\textit{Notes}: In both figures, in different scales of gray, we show the highest $90\%$, $75\%$, $50\%$, $25\%$ and $10\%$ quantiles of the distribution, for each horizon $h$, obtained through the Bayesian approach. Frequentist-valid intervals, obtained by retaining the $90\%$ draws of the reduced-form parameters with highest value of the posterior density function, are shown as the intervals delimited by the dotted-circle lines. In dotted lines we report the set of posterior means while in solid lines the upper and lower bounds of the robust credible regions with credibility $90\%$ obtained through the robust Bayesian approach.\par}}
\end{minipage}
\end{figure}

\begin{figure}[H]
        \caption{Impulse response functions of output gap to a monetary policy shock: Model III (set-identified SVAR-WB: equality and sign restrictions).}
  \label{fig:ModelIII}
\begin{center}
                \subfigure[{Model III: first regime}]{
                \includegraphics[angle=270,origin=c,scale=0.24]{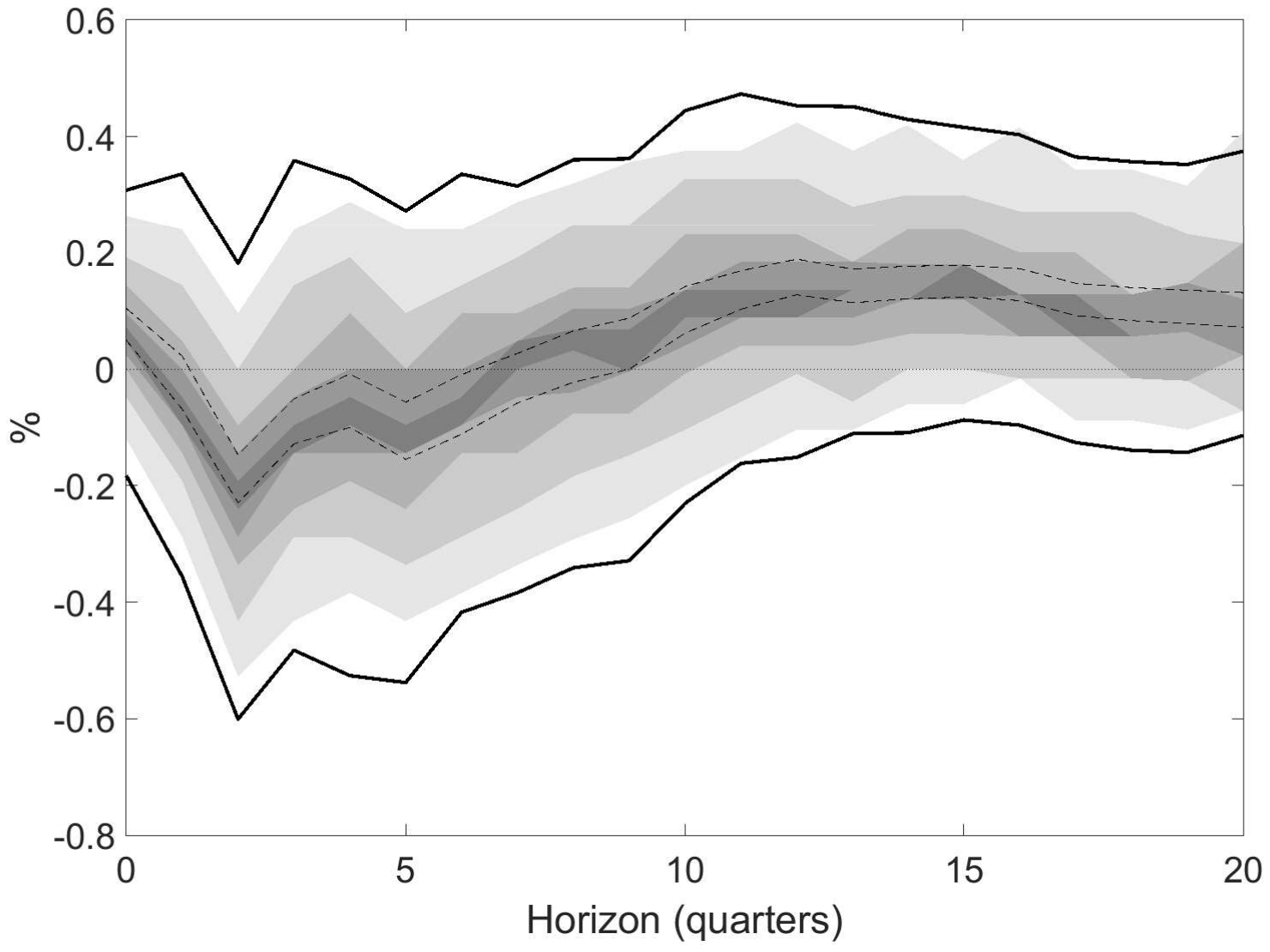}
                }
                \subfigure[{Model III: second regime}]{
                \includegraphics[angle=270,origin=c,scale=0.24]{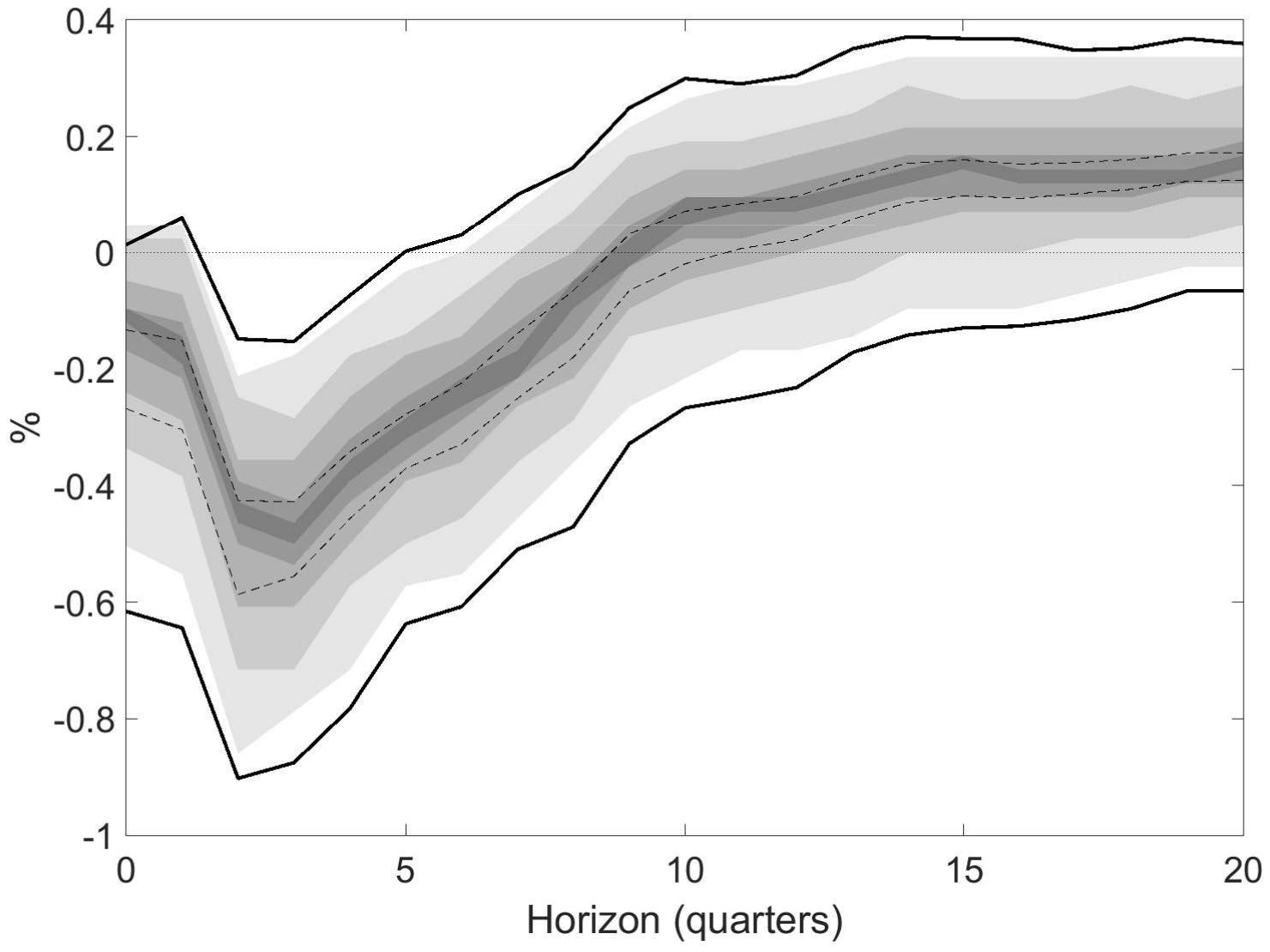}
                }
\end{center}
\begin{minipage}{\textwidth}%
{\scriptsize{\textit{Notes}: In both figures, in different scales of gray, we show the highest $90\%$, $75\%$, $50\%$, $25\%$ and $10\%$ quantiles of the distribution, for each horizon $h$, obtained through the Bayesian approach. In dotted lines we report the set of posterior means while in solid lines the upper and lower bounds of the robust credible regions with credibility $90\%$ obtained through the robust Bayesian approach.\par}}
\end{minipage}
\end{figure}

\begin{figure}[H]
        \caption{Impulse response functions of output gap to a monetary policy shock: Model IV (set-identified SVAR-WB: sign restrictions within regimes, only).}
  \label{fig:ModelIV}
\begin{center}
                \subfigure[{Model IV: first regime}]{
                \includegraphics[angle=270,origin=c,scale=0.24]{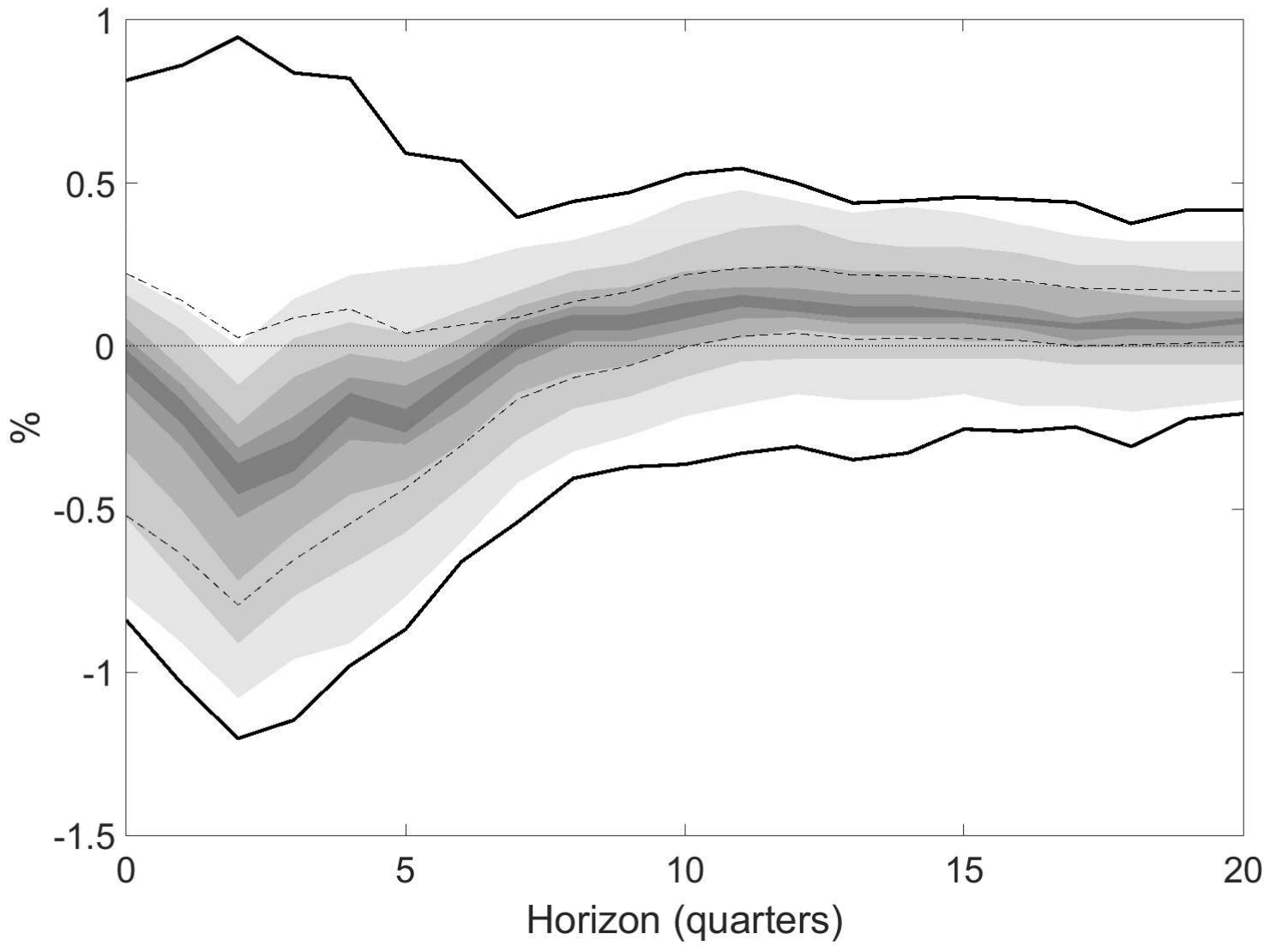}
                }
                \subfigure[{Model IV: second regime}]{
                \includegraphics[angle=270,origin=c,scale=0.24]{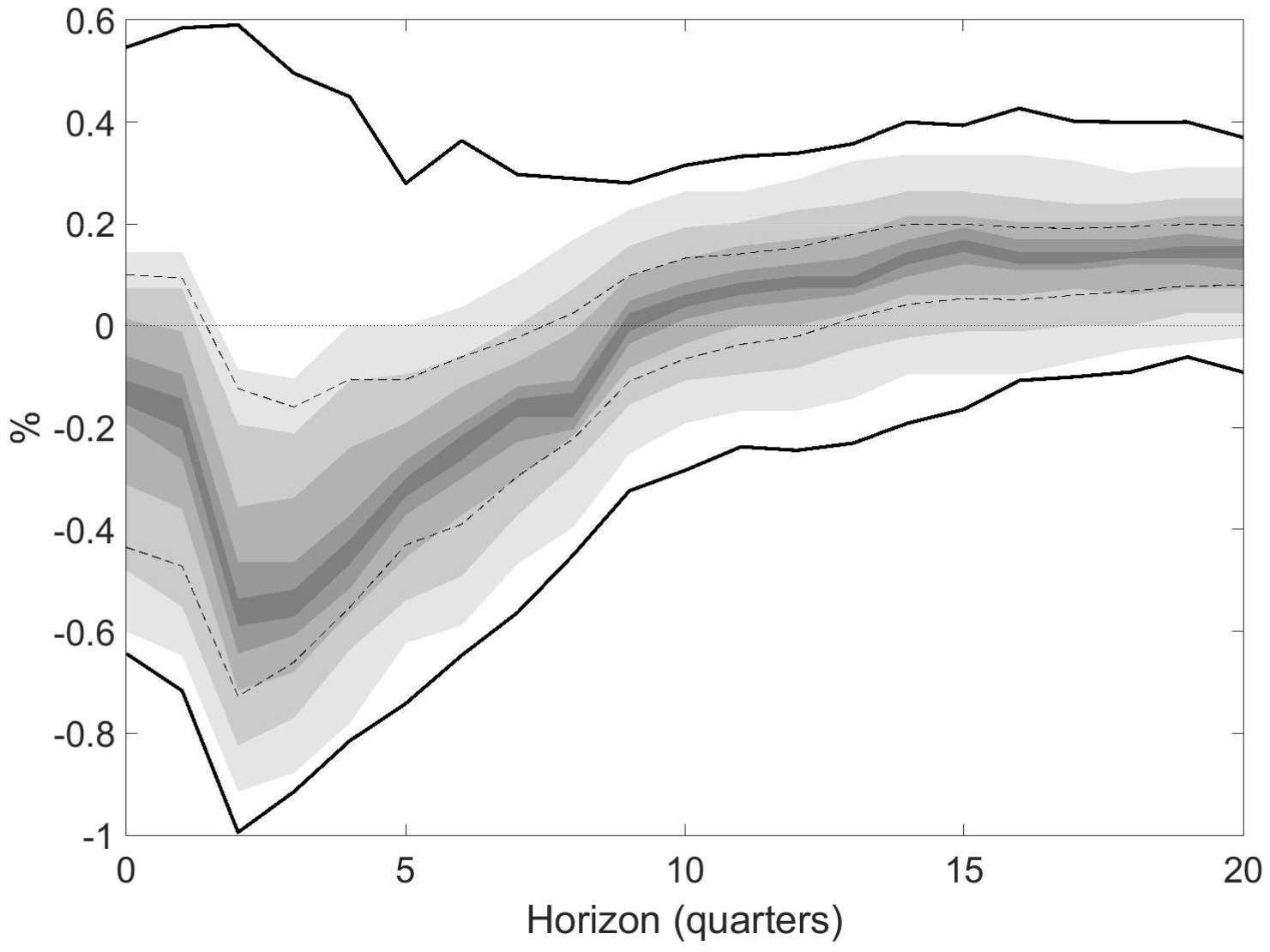}
                }
\end{center}
\begin{minipage}{\textwidth}%
{\scriptsize{\textit{Notes}: In both figures, in different scales of gray, we show the highest $90\%$, $75\%$, $50\%$, $25\%$ and $10\%$ quantiles of the distribution, for each horizon $h$, obtained through the Bayesian approach. In dotted lines we report the set of posterior means while in solid lines the upper and lower bounds of the robust credible regions with credibility $90\%$ obtained through the robust Bayesian approach.\par}}
\end{minipage}
\end{figure}

\begin{figure}[H]
        \caption{Impulse response functions of output gap to a monetary policy shock: Model V (set-identified SVAR-WB: sign restrictions + ranking restrictions + FEV restrictions).}
  \label{fig:ModelV}
\begin{center}
                \subfigure[{Model V: first regime}]{
                \includegraphics[angle=270,origin=c,scale=0.24]{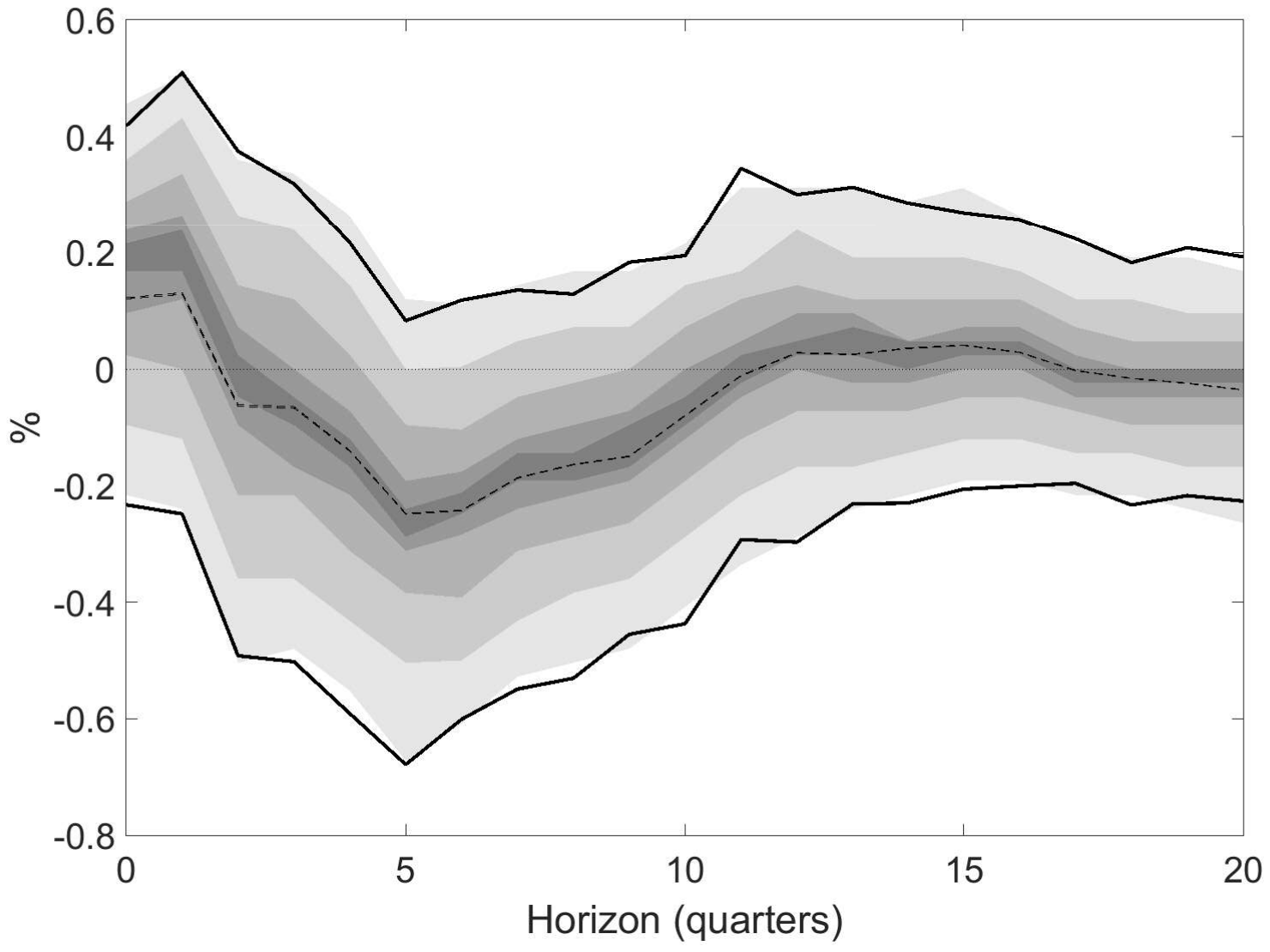}
                }
                \subfigure[{Model V: second regime}]{
                \includegraphics[angle=270,origin=c,scale=0.24]{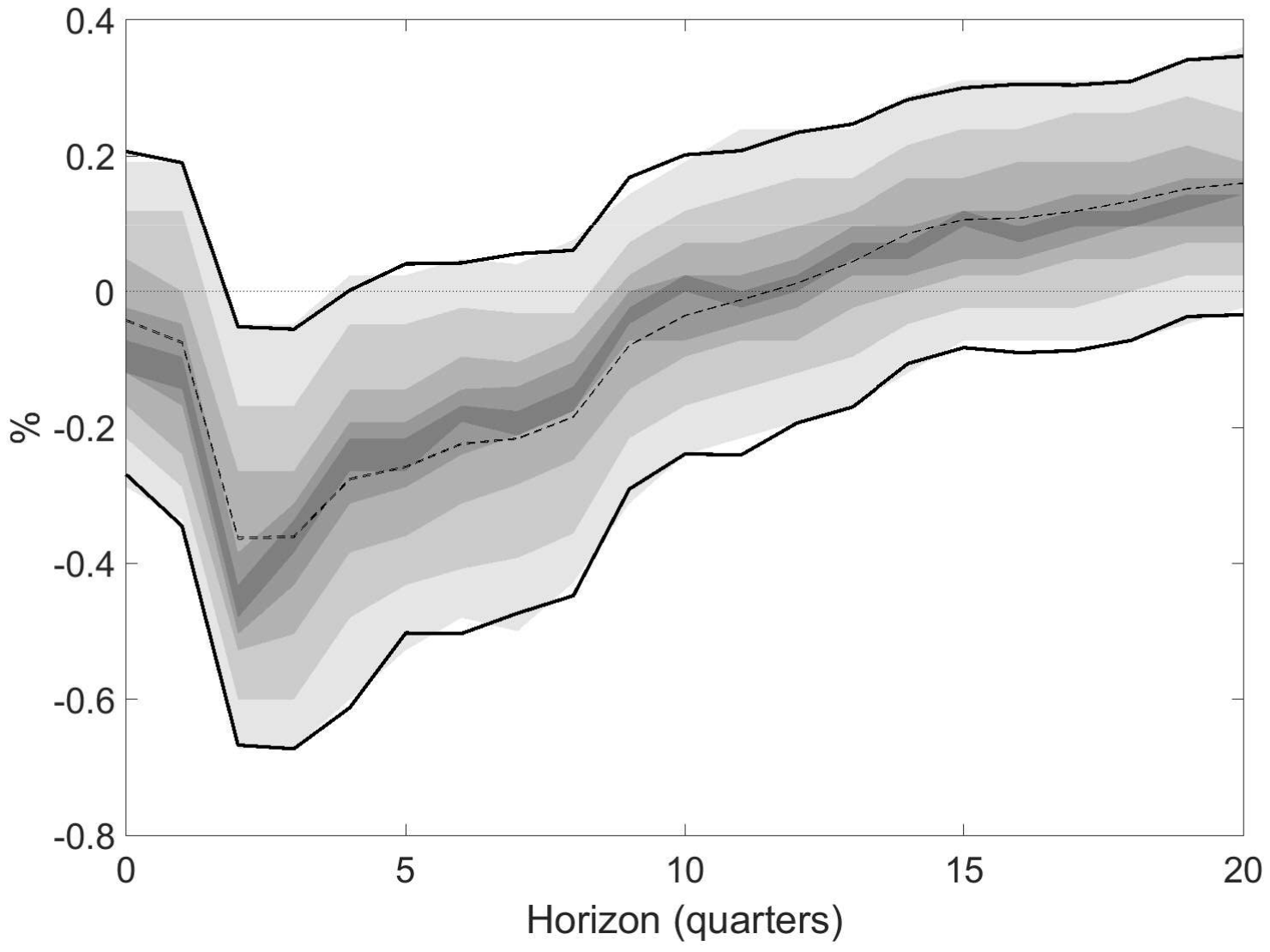}
                }
\end{center}
\begin{minipage}{\textwidth}%
{\scriptsize{\textit{Notes}: In both figures, in different scales of gray, we show the highest $90\%$, $75\%$, $50\%$, $25\%$ and $10\%$ quantiles of the distribution, for each horizon $h$, obtained through the Bayesian approach. In dotted lines we report the set of posterior means while in solid lines the upper and lower bounds of the robust credible regions with credibility $90\%$ obtained through the robust Bayesian approach.\par}}
\end{minipage}
\end{figure}

\begin{figure}[H]
        \caption{Impulse response functions of interest rate to supply and demand shocks: Model IV (set-identified SVAR-WB: sign restrictions within regimes, only).}
  \label{fig:ModelV_DS}
\begin{center}
								\subfigure[{Response to supply shocks: first regime}]{
                \includegraphics[angle=270,origin=c,scale=0.24]{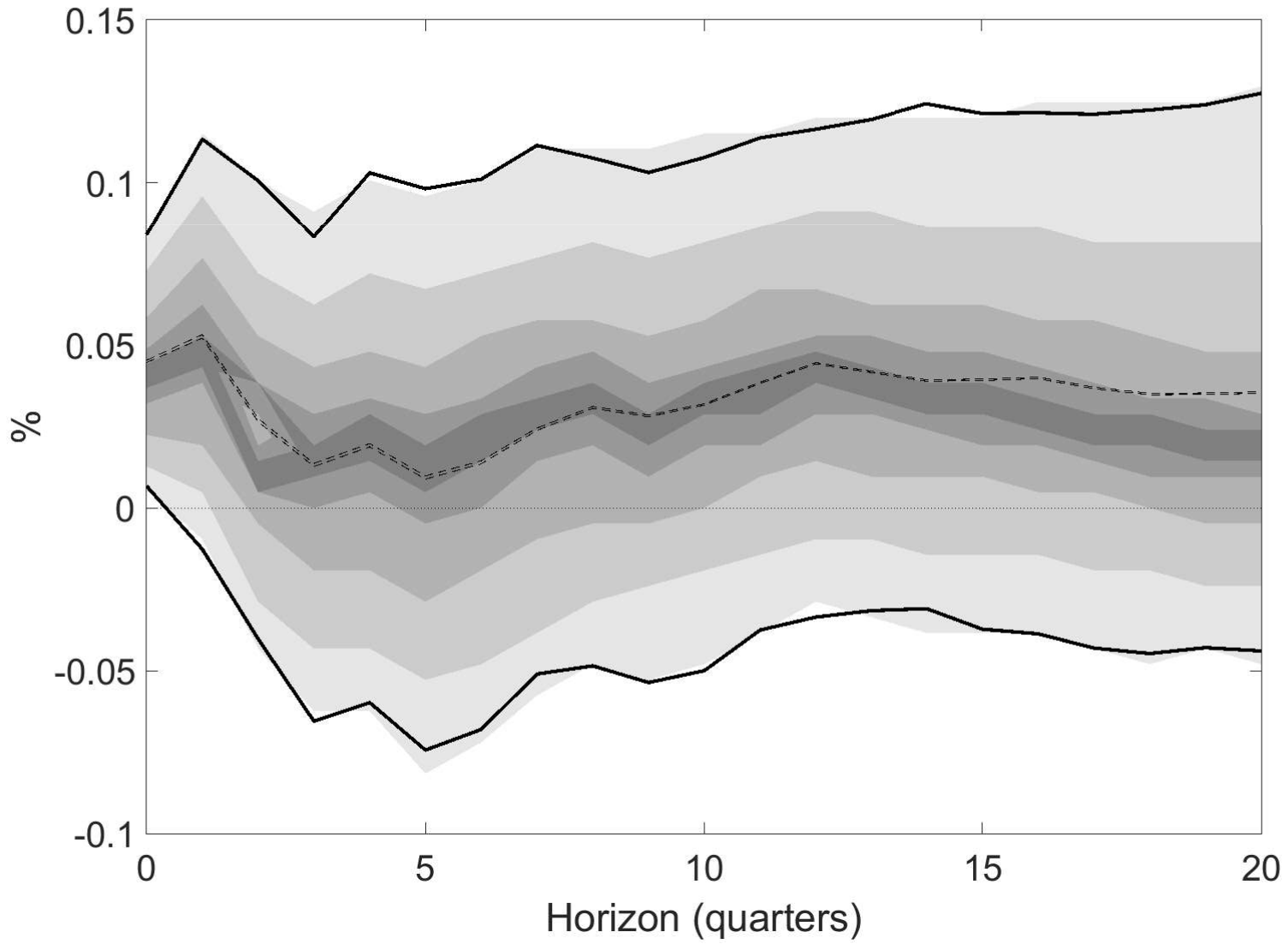}
                }
                \subfigure[{Response to supply shocks: second regime}]{
                \includegraphics[angle=270,origin=c,scale=0.24]{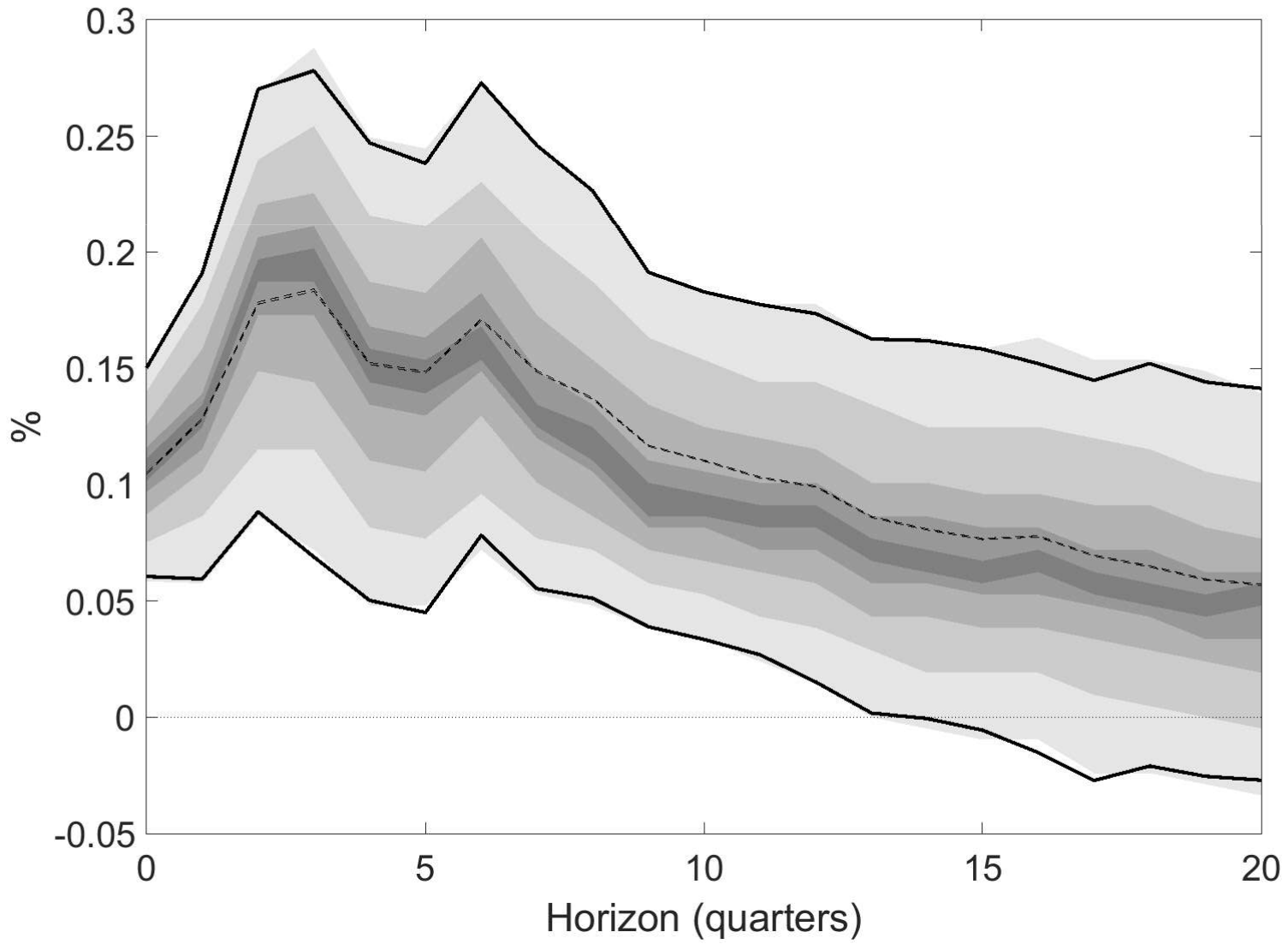}
                }\\
                \subfigure[{Response to demand shocks: first regime}]{
                \includegraphics[angle=270,origin=c,scale=0.24]{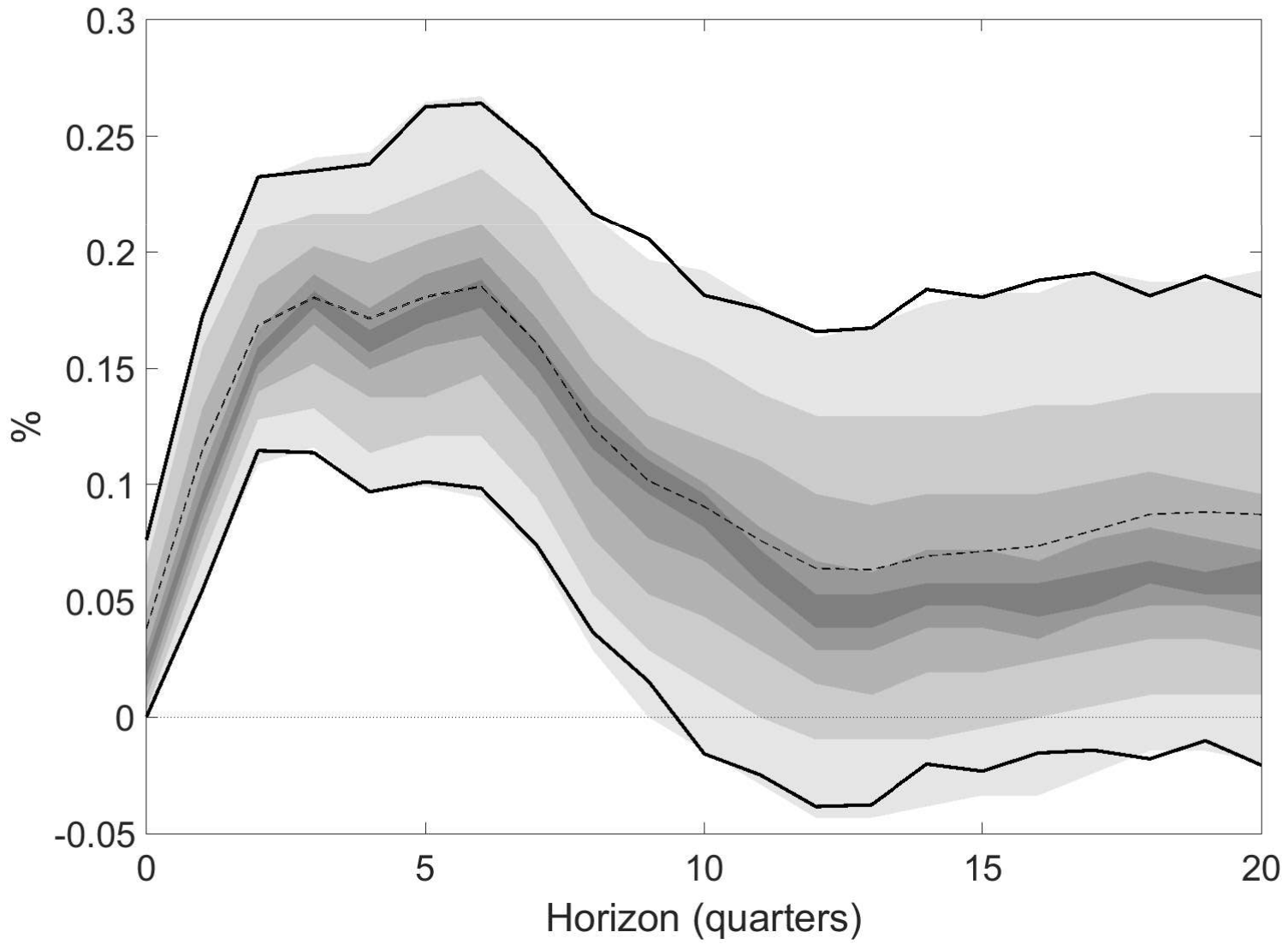}
                }
                \subfigure[{Response to demand shocks: second regime}]{
                \includegraphics[angle=270,origin=c,scale=0.24]{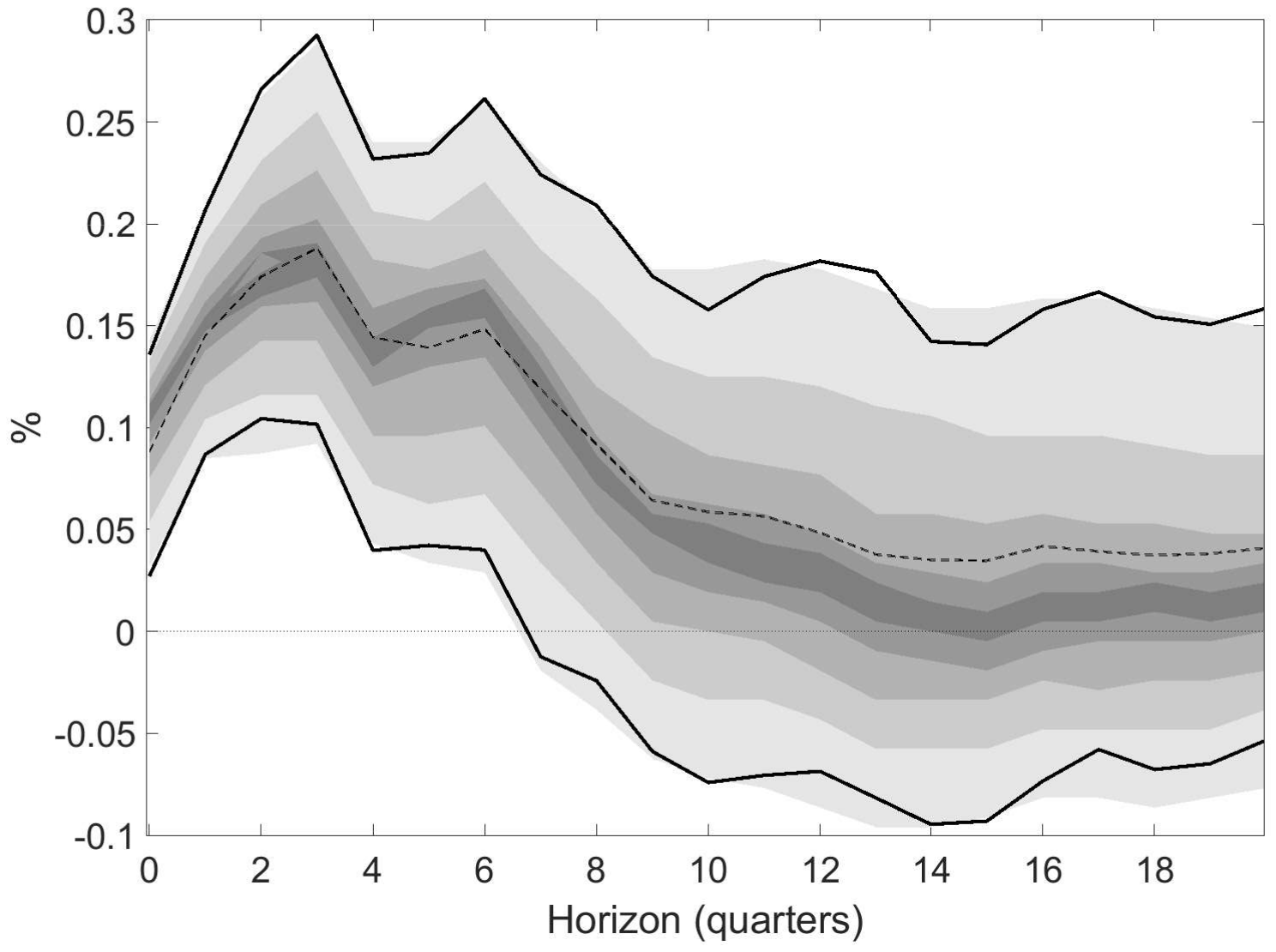}
                }
\end{center}
\begin{minipage}{\textwidth}%
{\scriptsize{\textit{Notes}: In all figures, in different scales of gray, we show the highest $90\%$, $75\%$, $50\%$, $25\%$ and $10\%$ quantiles of the distribution, for each horizon $h$, obtained through the Bayesian approach. Frequentist-valid intervals, obtained by retaining the $90\%$ draws of the reduced-form parameters with highest value of the posterior density function, are shown as the intervals delimited by the dotted-circle lines. In dotted lines we report the set of posterior means while in solid lines the upper and lower bounds of the robust credible regions with credibility $90\%$ obtained through the robust Bayesian approach.\par}}
\end{minipage}
\end{figure}

%%%%%%%%%%%%%%%%%%%%%%%%%%%%%%%%%%%%%%%%%%%%%%%%%%%%%%%%%%%%%%%%%%%%%%%%%%%%%%%
%														 SECTION CONCLUSION																%
%%%%%%%%%%%%%%%%%%%%%%%%%%%%%%%%%%%%%%%%%%%%%%%%%%%%%%%%%%%%%%%%%%%%%%%%%%%%%%%

\pagebreak

\section{Conclusion}
\label{sec:conclusion}

This paper deals with SVAR models with structural breaks and offers five main contributions. After presenting the representation of SVAR-WBs, we first study the identification theory and propose a set of results for easily checking whether the model is identified or not. The identification of the structural parameters can be obtained through a set of restrictions on the parameters, and new in the framework of SVAR-WBs, on functions of them. It is also shown that one can have gains in terms of identification when such restrictions are imposed across-regimes as stability restrictions. This characteristic, however, generates local identification, instead of the more suitable global one. The number of observationally equivalent identified points in the parametric space can be reduced by the inclusion of mild sign restrictions on the parameters or on functions of them, like impulse responses. 

Second, given the presence of local identification, we provide a new strategy for obtaining all the solutions of the identification issue, i.e. all the isolated and observationally equivalent parameters satisfying the imposed equality and inequality restrictions. 

Third, we provide three strategies for doing inference on the structural parameters that jointly consider all the admissible parameters. These approaches for estimating and doing inference on the parameters are completely new in this literature, specifically when compared to the ML approach generally used in SVARs and SVAR-WBs, that reduces to focus on just one of the multiple observationally equivalent admissible parameters. 

Fourth, we also consider the case of only set-identified SVAR-WB. Inequality constraints can be imposed, both on impulse responses and FEVs, within each regime or across the regimes. While the formers essentially coincide with traditional inequality restrictions in the SVAR literature, the latters are completely new, and offer an innovative tool for practitioners dealing with SVARs characterized by potential structural breaks on the parameters. For set-identified SVAR-WBs, a Bayesian and robust Bayesian approach have been proposed to perform inference on the structural parameters or on the impulse responses of interest.

Finally, all the theoretical results developed in the paper are implemented through an empirical analysis about the conduct of the monetary policy pursued by the Federal Reserve over the great inflation and great moderation periods. In this respect, we provide evidence about substantial different reaction functions over the two regimes, in particular characterized by a more aggressive and persistent response to price shocks during the great moderation. Such a different conduct leads to a more pronounced response of output to monetary policy shocks during this latter period.

\newpage

\renewcommand{\thesection}{\Alph{section}}\setcounter{section}{0}

\setlength{\baselineskip}{12pt} 
\bibliographystyle{ecta}
%\bibliography{acompat,SVARhetero}
\bibliography{SVAR_WB}

\newpage

%%%%%%%%%%%%%%%%%%%%%%%%%%%  APPENDIX %%%%%%%%%%%%%%%%%%%%%%%%%%%%%%%%%%%%%%%%%%%
\appendix
%%%%%%%%%%%%%%%%%%%%%%%%%%%%%%%%%%%%%%%%%%%%%%%%%%%%%%%%%%%%%%%%%%%%%%%%%%%%%%%%%

%%%%%%%%%%%%%%%%%%%%%%%%%%%  APPENDIX Proofs %%%%%%%%%%%%%%%%%%%%%%%%%%%%%%%%%%%%%%%%%%%
\section{Proofs} 
\label{app:Proofs}

\begin{proof}[Proof of Theorem \ref{theo:SuffCond}] 
	If both the transformation functions $G_p(\cdot)$ and the infinitesimal rotations $Q_p$ are admissible for any $p\in\{1,\ldots,s\}$, and 
	the restrictions are organized as discussed in Section \ref{sec:ident_restr}, and specifically, with the equations ordered according to 
	Eq. (\ref{eq:ordering}), we have
	\[
	\begin{array}{lcl}
	(R_{11,j}\,|\ldots|\,R_{1s,j})
	\left(\begin{array}{c}
	G_1\,Q_1\\
	\vdots\\
	G_s\,Q_s
	\end{array}\right)e_j=(0) 
	& \Longleftrightarrow & 
	(R_{11,j}\,|\ldots|\,R_{1s,j})
	\left(\begin{array}{c}
	G_1(I_n+H_1)\\
	\vdots\\
	G_2(I_n+H_s)
	\end{array}\right)e_j=(0) \\
	&\vdots&\\
	(R_{s1,j}\,|\ldots|\,R_{ss,j})
	\left(\begin{array}{c}
	G_1\,Q_1\\
	\vdots\\
	G_s\,Q_2
	\end{array}\right)e_j=(0) 
	& \Longleftrightarrow & 
	(R_{s1,j}\,|\ldots|\,R_{ss,j})
	\left(\begin{array}{c}
	G_1(I_n+H_1)\\
	\vdots\\
	G_2(I_n+H_2)
	\end{array}\right)e_j=(0) 
	\end{array}
	\]
	where, as before, $G_p = G(A_{p0}^\prime,A_{p+}^\prime)$, $\forall\;p\in\{1,\ldots,s\}$. Given that the infinitesimal transformation 
	is admissible, these quantities can be compacted as
	\[
	\begin{array}{lcl}
	\left(\begin{array}{ccc}
	R_{11,j} &\ldots& R_{1s,j}\\
  \vdots & \ddots & \vdots\\
	R_{s1,j} &\ldots& R_{ss,j}\\
  \end{array}\right)
	\left(\begin{array}{c}
	G_1\,H_1\\
	\vdots\\
	G_s\,H_s
	\end{array}\right)e_j=(0) 
	& \Longleftrightarrow & 
	\left(\begin{array}{ccc}
	R_{11,j} &\ldots& R_{1s,j}\\
  \vdots & \ddots & \vdots\\
	R_{s1,j} &\ldots& R_{ss,j}\\
	\end{array}\right)
	\left(\begin{array}{c}
	G_1\,H_{1,j}\\
	\vdots\\
	G_s\,H_{s,j}
	\end{array}\right)=(0) 
	\end{array}
	\]
	where the generic $H_{p,j}$ is the j-th column of $H_p$, $p\in\{1,\ldots,s\}$. 
	Using the properties of the vec operator and the Kronecker product we obtain that the previous relation can be written as
	\begin{equation}
	\label{eq:SysTwoReg}
	\begin{array}{l}
	\left(\begin{array}{ccc}
	R_{11,j} &\ldots& R_{1s,j}\\
  \vdots & \ddots & \vdots\\
	R_{s1,j} &\ldots& R_{ss,j}\\
	\end{array}\right)
	\left(\begin{array}{c}
	(H_{1,j}^\prime\otimes I_g)\ve G_1\\
	\vdots\\
	(H_{s,j}^\prime\otimes I_g)\ve G_s
	\end{array}\right)=(0)\\ 
	\\
	\hspace{3cm}\Longleftrightarrow 
	\underset{f_{j}\times sgn}{\left(\begin{array}{ccc}
	R_{1,j}^*(H_{1,j}^\prime\otimes I_g) &\ldots& R_{s,j}^*(H_{s,j}^\prime\otimes I_g)\\
  \end{array}\right)}
	\underset{sgn\times 1}{\left(\begin{array}{c}
	\ve G_1\\
	\vdots\\
	\ve G_s
	\end{array}\right)}=
	\underset{f_{j}\times 1}{(\;0\;)} 
	\end{array}
	\end{equation}
	Moreover, the explicit form of the restrictions introduced in Eq. (\ref{eq:ExpForm}) allows to write the constraints on all 
	$j\in\{1,\ldots,n\}$ as
	\begin{eqnarray}
		\left(\begin{array}{c} \ve G_1\\\vdots\\\ve G_s  \end{array}\right)=
		\left(\begin{array}{c}
		S_{1,1}^*\theta_{1}\\
		S_{1,2}^*\theta_{2}\\
		\vdots\\
		S_{1,n}^*\theta_{n}\\
		\hdashline
		\vdots\\
		\hdashline
		S_{s,1}^*\theta_{1}\\
		S_{s,2}^*\theta_{2}\\
		\vdots\\
		S_{s,n}^*\theta_{n}
		\end{array}\right)
		& \Longleftrightarrow &
		\left(\begin{array}{c} \ve G_1\\\vdots\\\ve G_s  \end{array}\right)=
		\left(\begin{array}{cccc}
		S_{1,1}^*& & &\\
		& S_{1,2}^*& &\\
		& & \ddots &\\
		& & & S_{1,n}^*\\
		\hdashline
		\vdots\\
		\hdashline
		S_{s,1}^*& & &\\
		& S_{s,2}^*& &\\
		& & \ddots &\\
		& & & S_{s,n}^*\\
		\end{array}\right)
		\left(\begin{array}{c}
		\theta_{1}\\ \theta_{2}\\ \vdots \\ \theta_{n}
		\end{array}\right)\nonumber\\
		& \Longleftrightarrow &
		\left(\begin{array}{c} \ve G_1\\\vdots\\\ve G_s  \end{array}\right)=
		S^*\,\theta \label{eq:vecG1Gs}
  \end{eqnarray}	
	where $\theta_j$ is the vector collecting the $\tau_j$ free parameters of the \textit{j}-th shock.
	Furthermore, remembering that $H_1,H_2,\ldots,H_s$ are skew-symmetric, the \textit{j}-th column of each of them can be written as
	\[
	H_{p,j}=
	\left(\begin{array}{c}
	h_{p,1j}\\ h_{p,2j}\\ \vdots\\ h_{p,nj}
	\end{array}\right)
	\hspace{0.5cm}\text{with}\hspace{0.3cm}p\in\{1,\ldots,s\}
	\]
	with, by definition, $h_{1,jj}=h_{2,jj}=\ldots=h_{s,jj}=0$, for $j\in\{1,\ldots,n\}$. Hence, for any $p\in\{1,\ldots,s\}$ and 
	$j\in\{1,\ldots,n\}$
	\[
	(H_{p,j}^\prime \otimes I_g ) = (h_{p,1j}I_g\,|\,h_{p,2j}I_g\,|\cdots|\;0\,I_g\;|\cdots|\;h_{p,nj}I_g), 
	\]
	that leads to
	\[
	R_{p,j}^*(H_{p,j}^\prime \otimes I_g ) = (h_{p,1j}R_{p,j}^*\,|\,h_{p,2j}R_{p,j}^*\,|\cdots|\;0\,R_{p,j}^*\;|\cdots|\;h_{p,nj}R_{p,j}^*),
	\hspace{0.1cm}\text{with }p\in\{1,\ldots,s\}, 
	\hspace{0.1cm}j\in\{1,\ldots,n\}.
	\]
	As already mentioned in Section \ref{sec:ident_restr}, let the equations be ordered according to Eq. (\ref{eq:ordering}). 
	Starting from Eq. (\ref{eq:SysTwoReg}) and considering the first shock ($j=1$), i.e. the most restricted one, if we combine with Eq. 
	(\ref{eq:vecG1Gs}) we obtain the following system of equations
	\begin{equation}
	\label{eq:RadoSys}
		\begin{array}{rcl}
		0\;R_{1,1}^*S_{1,1}^*\theta_{1}+h_{1,21}R_{1,1}^*S_{1,2}^*\theta_{2}+\ldots+h_{1,n1}R_{1,1}^*S_{1,n}^*\theta_{n}+\hspace{1cm}&&\\
		+0\;R_{2,1}^*S_{2,1}^*\theta_{1}+h_{2,21}R_{2,1}^*S_{2,2}^*\theta_{2}+\ldots+h_{2,n1}R_{2,1}^*S_{2,n}^*\theta_{n}+&&\\
		\vdots\hspace{3cm}&&\\
		+0\;R_{s,1}^*S_{s,1}^*\theta_{1}+h_{s,21}R_{s,1}^*S_{s,2}^*\theta_{2}+\ldots+h_{s,n1}R_{s,1}^*S_{s,n}^*\theta_{n}\hspace{0.35cm}& = & 0
		\end{array}
		\end{equation}
		or, using the definitions in Eq.s (\ref{eq:Vjik})-(\ref{eq:Vjtheta}), more compactly,
		\begin{equation}
		\label{eq:RadoSysComp}
		\setlength{\dashlinegap}{1pt}
		\underset{f_{1}\times s(n-1)}{\left[
		\begin{array}{c:c:c:c}
		V_{1,1,2}\theta_{2}\;|\ldots|\;V_{1,1,n}\theta_n\; & \;V_{1,2,2}\theta_{2}\;|\ldots|\;V_{1,2,n}\theta_n\; & \ldots & 
		\;V_{1,s,2}\theta_{2}\;|\ldots|\;V_{1,s,n}\theta_n\\
		\end{array}\right]}
		\underset{s(n-1)\times 1}{\left(\begin{array}{c}
		h_{1,21}\\\vdots\\ h_{1,n1}\\\hdashline h_{2,21}\\\vdots\\ h_{2,n1}\\\hdashline \vdots\\\hdashline h_{s,21}\\\vdots\\ h_{s,n1}
		\end{array}\right)} = \underset{f_{1}\times 1}{(\;0\;).}
	\end{equation}
	If the quantity under the square brackets, that we denoted as $\Vtj$ and that depends on the restrictions (through the matrices $V_{1,i,k}$) 
	and on the unknown parameters (through the vectors $\theta_k$), has full column rank, then the 
	only admissible solution is the null vector, leading to the conclusion that the first column of $H_1,H_2,\ldots,H_s$ is equal to zero. 
	The first shock, in each of the regimes, will be thus locally identified.
		
	Checking for the identification of the second shock ($j=2$) can be easily done by checking whether the second column of $H_1,H_2,\ldots,H_s$ 
	is a zero vector, too. A system of equations of the kind of (\ref{eq:RadoSys}) can be written where the unknowns are represented by 
	$H_{1,2},H_{2,2},\ldots,H_{s,2}$. It is interesting to note, however, that if the first equation of the model is identified in all the
	regimes, then we already have that $h_{1,21}=h_{2,21}=\ldots=h_{s,21}=0$, which implies that $h_{1,12}=h_{2,12}=\ldots=h_{s,12}=0$.
	The system of equations, thus, reduces to
	\begin{equation}
		\label{eq:RadoSys2}
		\begin{array}{rcl}
		h_{1,32}R_{1,2}^*S_{1,3}^*\theta_{3}+\ldots+h_{1,n2}R_{1,2}^*S_{1,n}^*\theta_{n}+\hspace{1cm}&&\\
		+h_{2,32}R_{2,2}^*S_{2,3}^*\theta_{3}+\ldots+h_{2,n2}R_{2,2}^*S_{2,n}^*\theta_{n}+&&\\
		\vdots\hspace{3cm}&&\\
			+h_{s,32}R_{s,2}^*S_{s,3}^*\theta_{3}+\ldots+h_{s,n2}R_{s,2}^*S_{s,n}^*\theta_{n}\hspace{0.35cm}& = & 0
		\end{array}
		\end{equation}
		or, more compactly,
		\begin{equation}
		\label{eq:RadoSysComp2}
		\setlength{\dashlinegap}{1pt}
		\underset{f_{2}\times s(n-2)}{\left[
		\begin{array}{c:c:c:c}
		V_{2,1,3}\theta_{3}\;|\ldots|\;V_{2,1,n}\theta_n\; & \;V_{2,2,3}\theta_{3}\;|\ldots|\;V_{2,2,n}\theta_n\; & \ldots & 
		\;V_{2,s,3}\theta_{3}\;|\ldots|\;V_{2,s,n}\theta_n\\
		\end{array}\right]}
		\underset{s(n-2)\times 1}{\left(\begin{array}{c}
		h_{1,32}\\\vdots\\ h_{1,n2}\\\hdashline h_{2,32}\\\vdots\\ h_{2,n2}\\\hdashline \vdots\\\hdashline h_{s,32}\\\vdots\\ h_{s,n2}
		\end{array}\right)} = \underset{f_{2}\times 1}{(\;0\;).}
	\end{equation}
	We can use, thus, the same strategy as before to check whether the unique solution to the previous system is the null vector. In the same
	way we can proceed and check the identification of all the equations of the SVAR-WB with $s$ regimes. This proves the theorem.
	
	\vspace{0.5cm}
	
	Moreover, going back to system (\ref{eq:RadoSysComp}), it admits the null vector as the unique solution only if, at least, the number of
	equations is not less than the number of unknowns, i.e. $f_1\geq s(n-1)$. Similarly, for the system (\ref{eq:RadoSysComp2}), it must 
	hold that $f_2\geq s(n-2)$, and similarly for all the equations of the SVAR-WB. This proves our Corollary \ref{corol:OrderCond}.
	
	\vspace{0.5cm}
	Furthermore, considering jointly all the shocks of the SVAR-WB, it must hold that 
	$f=f_1+f_2+\ldots+f_n\geq s(n-1)+s(n-2)+\ldots +s=s[(n-1)+(n-2)+\ldots+1] = sn(n-1)/2$, proving formally the necessary condition of 
	Theorem \ref{theo:OrdCond}.
\end{proof}	
	
\vspace{0.5cm}

We first report two lemmas providing important mathematical results that will be used for proving the theorem. 
Details can be obtained in \cite{Spivak65}.

\vspace{0.5cm}
 
\begin{lemma}
  \label{lemmaRWZ1}
	Let $M_1$ and $M_2$ be two differentiable manifolds of the same dimension and let $h$ be a continuously differentiable function 
	from $M_1$ to $M_2$. If $E$ is a set of measure zero in $M_1$, then $h(E)$ is of measure zero in $M_2$.
\end{lemma}

\begin{lemma}
  \label{lemmaRWZ2}
	Let $M_1$ be an i-dimensional differentiable manifold and $M_2$ a j-dimensional differentiable manifold and let $h$ be a continuously 
	differentiable function from $M_1$ to $M_2$ whose derivative is of rank j for all $x\in M_1$. If $\hat{M}_2$ is a differentiable 
	submanifold of $M_2$ and $E$ is a set of measure zero in $\hat{M}_2$, then $\hat{M}_1=h^{-1}\left(\hat{M}_2\right)$ is a differentiable 
	submanifold of $M_1$ and $h^{-1}(E)$ is of measure zero in $\hat{M}_1$.
\end{lemma}

%\begin{lemma}
%  \label{lemmaRWZ3}
%	For $1\leq j\leq n$, let $V_j$ be a linear subspace of $\mathbb{R}^m$ and let $V=V_1\times V_2\times\cdots\times V_n$. 
%	Define $S$ to be the set of all $\left(v_1,\cdots,v_n\right)\in V$, whose span is of dimension strictly less than $n$. 
%	Either $S=V$ or $S$ is a set of measure zero in $V$. 
%\end{lemma}

\vspace{0.5cm}

\begin{proof}[Proof of Theorem \ref{theo:IdentEverywhere}] 
  Consider the first shock ($j=1$). This is locally identified if and only if the system (\ref{eq:RadoSysComp}) admits the null 
	vector as the unique possible solution.
	Equivalently, we can say that it is identified if and only if the $f_{1}\times s(n-1)$ matrix $\textbf{V}_1(\theta)$ has full column rank. 
	However, the rank of $\textbf{V}_1(\theta)$ is equal to the rank of the square matrix $\big(\textbf{V}_1(\theta)^\prime
	\textbf{V}_1(\theta)\big)$, that will be full when the determinant is different from zero. 
	The determinant $|\textbf{V}_1(\theta)^\prime\textbf{V}_1(\theta)|$, however, is a finite order polynomial in 
	$\theta_1^\ast=\left(\theta_2^\prime,\,\ldots,\,\theta_n^\prime\right)^\prime$.
	Let
	\begin{equation}
	\label{eq:W1theta}
		W_1 = \big\{\theta_1\in \Re^{\tau_1} \hspace{0.3cm}\big| \hspace{0.3cm} \text{det}|\textbf{V}_1(\theta)^\prime\textbf{V}_1(\theta)|\neq 0
		\big. \big\}.\nonumber
	\end{equation}
	Since the Lebesgue measure of zero set of a polynomial function is zero (see, e.g. \citeauthor{Johansen95}, \citeyear{Johansen95}, or
	\citeauthor{CT05poly}, \citeyear{CT05poly}), then either $W_1$ is empty because the rank is not full, or its complement is empty 
	(as the determinant does not depend on $\theta_1$). Similarly we can proceed for all $j=\{1,\ldots,n\}$. We thus obtain the set
	\begin{equation}
	\label{eq:Wtheta}
		W = \big\{\theta\in \Re^{\tau} \hspace{0.3cm}\big| \hspace{0.3cm} 
		\text{det}|\textbf{V}_j(\theta)^\prime\textbf{V}_j(\theta)|\neq 0,
		\hspace{0.2cm} \forall \hspace{0.2cm} j=\{1,\ldots,n\}\big. \big\}.\nonumber
	\end{equation}
	that, either is empty, or its complement is of measure zero in $\Re^{\tau_1+\ldots+\tau_n}$.
	
	Now, from Eq. (\ref{eq:vecG1Gs}) we know that
	\[
		\left(\begin{array}{c} \ve G_1\\\vdots\\\ve G_s  \end{array}\right)=S^*\,\theta. 
	\]
	If we denote by $Z$ the set
	\begin{equation}
	\label{eq:Z}
		Z = \big\{\textbf{G}\Atot \in \Re^{sgn} \hspace{0.3cm}\big| \hspace{0.1cm} 
		\text{the sufficient condition holds for the shock $j=1$ in all regimes}\big. \big\}.\nonumber
	\end{equation}
	then, as the rank of $S^*$ is full by construction, according to Lemma \ref{lemmaRWZ2}, if $W$ is of measure zero in $\Re^\tau$,
	then $Z$ is of measure zero in $\Re^{sgn}$.
	
	Finally, let $K$ be defined as in Eq. (\ref{eq:k})
 	\[
		K = \left\{\Atot\in \Ar \hspace{0.3cm}\left| \hspace{0.3cm} \text{the sufficient rank condition holds}\right. \right\}.
	\]
	Hence, if $Z$ is of measure zero in $\Re^{sgn}$ then, according to Condition \ref{cond:Regular} and Lemma \ref{lemmaRWZ2}, 
	$K$ will be of measure zero as well.
\end{proof}	

\vspace{0.5cm}

\begin{proof}[Proof of Theorem \ref{theo:NecSuff}]
	The idea of the proof is to check whether in the neighborhood of an admissible parameter point there exists another point 
	observationally equivalent and still admissible. 	This can be done by checking about the existence of at least one skew 
	symmetric matrix $H_p$, $p=1,\ldots,s$, such that the orthogonal matrix $Q_p=(I_n+H_p)$ transforms the admissible parameters 
	$(A_{p0},A_{p+})$ into the still admissible $(Q_pA_{p0},Q_pA_{p+})$.
	
	The strategy, thus, is similar to the proof of Theorem \ref{theo:SuffCond}, but the identification check has to be pursued for all 
	the equations of the SVAR-WB jointly. In fact, if the number of restrictions imposed for the \textit{j}-th shock (in all the regimes) 
	$q_{\cdot j} < s(n-j)$, then such shock cannot be identified as the associated system of equations for evaluating the existence of any 
	$H_p$ has more unknowns (the \textit{j}-th column of $H_1,\ldots,H_s$) than equations (defined by the matrix $\tVtj$). 
	The identification issue, however, can be addressed by considering all the shocks jointly.
	
	Starting from the system (\ref{eq:RadoSysComp}) and extending it for all the shocks we can write
	\begin{equation}
	\label{eq:NecSuffSys}
		\left(\begin{array}{cccc}
		\tilde{\textbf{V}}_1(\theta) & & &\\
		& \tilde{\textbf{V}}_2(\theta) & &\\
		& & \ddots & \\
		& & & \tilde{\textbf{V}}_n(\theta)
		\end{array}\right)\:\:
		\ve\:\left(\begin{array}{c}
		H_1\\H_2\\ \vdots\\ H_s
		\end{array}\right)
		= \left(\begin{array}{c}
		0\\0\\ \vdots\\ 0
		\end{array}\right).
	\end{equation}
	However, the $n\times n$ matrices $H_1,\ldots,H_s$ are skew-symmetric, each of which characterized by $n(n-1)/1$ distinct elements 
	collected into the $h_1,\ldots,h_s$ vectors, respectively, 	representing the effective unknowns in the system (\ref{eq:NecSuffSys}). 
	
	Moving from $[\ve\left(H_1^\prime,\,\ldots\,, H_s^\prime\right)^\prime]$ to $\left(h_1^\prime\,,\ldots\,,h_s^\prime\right)^\prime$ 
	can be done through the matrix $\Ttt$ defined in Eq. (\ref{eq:Ttt}). In fact, if we define the matrix $\tilde{T}_{n,s}$ as in 
	Eq. (\ref{eq:Tt}), it is easy to see that
	\begin{eqnarray}
		\ve\:\left(\begin{array}{c}H_1\\\vdots\\ H_s	\end{array}\right) & = & (\tilde{T}_{n,s}\otimes I_n)\: 
		\left(\begin{array}{c}\ve\,(H_1)\\\vdots\\ \ve\,(H_s)\end{array}\right)\nonumber\\
		& = & (\tilde{T}_{n,s}\otimes I_n)\:\left(\begin{array}{c}\tilde{D}_n\,h_1\\\vdots\\ \tilde{D}_n\,h_s\end{array}\right)\nonumber\\
		& = & (\tilde{T}_{n,s}\otimes I_n)\:(I_s\otimes \tilde{D}_n)\left(\begin{array}{c}h_1\\\vdots\\h_s\end{array}\right)\nonumber\\
		& = & \Ttt\left(\begin{array}{c}h_1\\\vdots\\h_s\end{array}\right).\nonumber
	\end{eqnarray}
	The system (\ref{eq:NecSuffSys}) will admit the null vector as the unique solution, denoting local identification, 
	if and only if the matrix $\tVt \: \Ttt$ has full column rank, which proves the theorem.
\end{proof}	

\vspace{0.5cm}

\begin{proof}[Proof of Theorem \ref{theo:IdentEverywhereRank}]
	The proof essentially follows the one for proving Theorem \ref{theo:IdentEverywhere}. The only difference is represented by the 
	starting point, that in this case is the square matrix $\Big(\Vtt^\prime\Vtt\Big)$, that will be of full rank when the determinant 
	is different from zero. The remaining is completely equivalent. 
\end{proof}

%%%%%%%%%%%%%%%%%%%%%%%%%%%%%%%%%%% APPENDIX EXAMPLES %%%%%%%%%%%%%%%%%%%%%%%%%%%%%%%%%%%%%%%%%%
\section{Examples}
\label{app:examples}

In this section we provide some examples on how to impose the restrictions discussed in Section \ref{sec:ident_restr} and how to apply the 
theory of identification of SVAR-WBs presented in Section \ref{sec:Id}.

\subsection{Triangular SVAR-WB with no restrictions across regimes}
\label{ex:IdentChol}

	Consider a SVAR-WB model with one single break, i.e. simply two regimes. Moreover, suppose that the restrictions are imposed on the
	contemporaneous matrices such that $A_{10}$ and $A_{20}$ are both lower triangular, with no across-regime restrictions. In this	case, 
	the SVAR-WB can be investigated separately in each single regime. Hence, using the standard methodology proposed by RWZ, the SVAR-WB is
	globally identified as all single SVARs, in each of the regimes, are globally identified.

\subsection{Trivariate SVAR-WB with restrictions on the contemporaneous matrices}
\label{ex:IdentTrivariate}

	Consider a trivariate ($n=3$) SVAR-WB model with one single break. Moreover, as before, suppose that the restrictions are imposed on 
	the contemporaneous matrices $A_{10}$ and $A_{20}$ such that
	\begin{equation}
		\label{eq:ExTrivariate}
		A_{10}=\left(\begin{array}{ccc}
					 \encircled{a_{1,11}}&0&0\\a_{1,21}&\underset{}{a_{1,22}}&0\\a_{1,31}&a_{1,32}&a_{1,33}
					 \end{array}\right)\hspace{2cm}
		A_{20}=\left(\begin{array}{ccc}
					 \encircled{a_{1,11}}&a_{2,12}&0\\a_{2,21}&\underset{}{a_{2,22}}&0\\a_{2,31}&a_{2,32}&a_{2,33}
					 \end{array}\right),
	\end{equation}
	where it clearly emerges that a) $A_{10}$ is lower triangular, b) $a_{1,11}$ remains unchanged in the two regimes and c) $A_{20}$ is 
	not lower triangular and contains more parameters than in standard SVAR models. 
	This last point represents the gain that can be obtained when some structural parameters do not change across the regimes.
	
	The function allowing to impose the restrictions on the structural parameters of the \textit{p}-th regime is simply given by 
	$G\left(\Apt,\Apdt\right)=\Apt$, and the restrictions will take the form
	\begin{equation}
		\label{eq:ExRestrTrivariate}
		\Rj\left(\begin{array}{c}A_{10}^\prime\\ A_{20}^\prime \end{array}\right)e_j=0,\hspace{0.3cm}\text{for}\hspace{0.3cm}1\leq j\leq 3 
		\hspace{0.3cm}\text{and}\hspace{0.3cm}p=\{1,2\}. 
	\end{equation}
	It is important to note that the first column of $\left(\begin{array}{c}G_1\\G_2\end{array}\right)$, i.e. when $j=1$, is characterized by 
	$s(n-j)=2(3-1)=4$ restrictions, the second column ($j=2$) by 2 restrictions and, finally, the third column ($j=3$) by 0 restrictions. This
	is exactly what indicated in Corollary \ref{corol:OrderCond}, suggesting that we can use Theorem \ref{theo:SuffCond} for studying local
	identification in the SVAR-WB under investigation.
	
	More specifically, for the first equation (in the two regimes):
	\begin{eqnarray}
		\label{eq:ImRestrTrivariate1}
		j=1 \hspace{1cm}& \Rightarrow & \underbrace{\left(\begin{array}{c:c}R_{11,1} & R_{12,1}\\\hdashline 
		R_{21,1} & R_{22,1}\end{array}\right)}_{\textbf{R}_{1}}\left(\begin{array}{c}A_{10}^\prime\\ A_{20}^\prime \end{array}\right)e_1=0\nonumber\\
									  & \Rightarrow & \left(\begin{array}{ccc:ccc}0&1&0&0&0&0\\0&0&1&0&0&0\\1&0&0&-1&0&0\\\hdashline 0&0&0&0&0&1\end{array}\right)
										    	\left(\begin{array}{c}a_{1,11}\\a_{1,12}\\a_{1,13}\\a_{2,11}\\a_{2,12}\\a_{2,13}\end{array}\right)=
													\left(\begin{array}{c}0\\0\\0\\0\end{array}\right)\nonumber
	\end{eqnarray}
	where
	\[
		R_{11,1}=\left(\begin{array}{ccc}0&1&0\\0&0&1\\1&0&0\end{array}\right),\hspace{0.3cm}
		R_{12,1}=\left(\begin{array}{ccc}0&0&0\\0&0&0\\-1&0&0\end{array}\right),\hspace{0.3cm}
		R_{22,1}=\left(\begin{array}{ccc}0&0&1\end{array}\right). 
	\]
	The corresponding explicit form, instead, can be written as
	\begin{eqnarray}
		\label{eq:ExRestrTrivariate1}
		j=1 \hspace{1cm}& \Rightarrow & \left(\begin{array}{c}a_{1,11}\\a_{1,12}\\a_{1,13}\\a_{2,11}\\a_{2,12}\\a_{2,13}\end{array}\right)=
																		\underbrace{\left(\begin{array}{c:c}S_{11,1} & S_{12,1}\\\hdashline S_{21,1} & 
																		S_{22,1}\end{array}\right)}_{\textbf{S}_{1}}
																		\left(\begin{array}{c}\theta_{1,1}\\ \theta_{1,2} \end{array}\right)\nonumber\\
									  & \Rightarrow & \left(\begin{array}{c}a_{1,11}\\a_{1,12}\\a_{1,13}\\a_{2,11}\\a_{2,12}\\a_{2,13}\end{array}\right)=
																		\left(\begin{array}{cc}1&0\\0&0\\0&0\\\hdashline 1&0\\0&1\\0&0\end{array}\right)
																		\left(\begin{array}{c}\theta_{1,1}\\ \theta_{1,2} \end{array}\right)\nonumber
	\end{eqnarray}
	where $\theta_{1,1}$ and $\theta_{1,2}$ are the two elements of $\theta_1$, i.e. the free parameters associated to the first row of 
	$A_{1,0}$ and $A_{2,0}$, respectively, and where
	\[
		S_{12,1}=\left(\begin{array}{cc}1&0\\0&0\\0&0\end{array}\right),\hspace{0.3cm}
		S_{22,1}=\left(\begin{array}{cc}1&0\\0&1\\0&0\end{array}\right) 
	\]
	while $S_{11,1}$ is not defined, being $R_{11,1}$ a non-singular $3\times 3$ matrix.
	
	For the second equation, in both the regimes, 
	\begin{eqnarray}
		\label{eq:ImRestrTrivariate2}
		j=2 \hspace{1cm}& \Rightarrow & \underbrace{\left(\begin{array}{c:c}R_{11,2} & R_{12,2}\\\hdashline R_{21,2} & 
		R_{22,2}\end{array}\right)}_{\textbf{R}_{2}}\left(\begin{array}{c}A_{10}^\prime\\ A_{20}^\prime \end{array}\right)e_2=0\nonumber\\
									  & \Rightarrow & \left(\begin{array}{ccc:ccc}0&0&1&0&0&0\\\hdashline 0&0&0&0&0&1 \end{array}\right)
										    	\left(\begin{array}{c}a_{1,21}\\a_{1,22}\\a_{1,23}\\a_{2,21}\\a_{2,22}\\a_{2,23}\end{array}\right)=
													\left(\begin{array}{c}0\\0\end{array}\right)\nonumber
	\end{eqnarray}
	where
	\[
		R_{11,2}=\left(\begin{array}{ccc}0&0&1\end{array}\right),\hspace{0.3cm}
		R_{12,2}=\left(\begin{array}{ccc}0&0&0\end{array}\right),\hspace{0.3cm}
		R_{21,2}=\left(\begin{array}{ccc}0&0&0\end{array}\right),\hspace{0.3cm}
		R_{22,2}=\left(\begin{array}{ccc}0&0&1\end{array}\right). 
	\]
	The corresponding explicit form, as for the first equation, can be written as
	\begin{eqnarray}
		\label{eq:ExRestrTrivariate2}
		j=2 \hspace{1cm}& \Rightarrow & \left(\begin{array}{c}a_{1,21}\\a_{1,22}\\a_{1,23}\\a_{2,21}\\a_{2,22}\\a_{2,23}\end{array}\right)=
																		\underbrace{\left(\begin{array}{c:c}S_{11,2} & S_{12,2}\\\hdashline S_{21,2} & 
																		S_{22,2}\end{array}\right)}_{\textbf{S}_{2}}
																		\left(\begin{array}{c}\theta_{2,1}\\ \theta_{2,2} \\ \theta_{2,3} \\ \theta_{2,4}
																		\end{array}\right)\nonumber\\
									  & \Rightarrow & \left(\begin{array}{c}a_{1,21}\\a_{1,22}\\a_{1,23}\\a_{2,21}\\a_{2,22}\\a_{2,23}\end{array}\right)=
																		\left(\begin{array}{cc:cc}1&0&0&0\\0&1&0&0\\0&0&0&0\\\hdashline 0&0&1&0\\0&0&0&1\\0&0&0&0\\ 
																		\end{array}\right)
																		\left(\begin{array}{c}\theta_{2,1}\\ \theta_{2,2} \\ \theta_{2,3} \\ \theta_{2,4} \end{array}\right)
																		\nonumber
	\end{eqnarray}
	where $(\theta_{2,1},\,\theta_{2,2},\,\theta_{2,3},\,\theta_{2,4})^\prime$, is the vector of free parameters associated to the second 
	rows of $A_{1,0}$ and $A_{2,0}$, and where
	\[
  	S_{11,2}=\left(\begin{array}{cc}1&0\\0&1\\0&0\end{array}\right),\hspace{0.3cm}
		S_{12,2}=\left(\begin{array}{cc}0&0\\0&0\\0&0\end{array}\right),\hspace{0.3cm}
		S_{21,2}=\left(\begin{array}{cc}1&0\\0&0\\0&0\end{array}\right),\hspace{0.3cm}
		S_{22,2}=\left(\begin{array}{cc}1&0\\0&1\\0&0\end{array}\right). 
	\]

	Finally, considering the third equation ($j=3$), there are no restrictions associated, thus $\textbf{R}_{3}$ is not defined and  
	$\textbf{S}_{3}=\left(\begin{array}{c:c}I_3 & 0_{3\times 3}\\ \hdashline 0_{3\times 3}&I_3\end{array}\right)$, with $S_{11,3}=I_3$, 
	$S_{12,3}=0_{3\times 3}$, $S_{21,3}=0_{3\times 3}$ and $S_{22,3}=I_3$.
	
	According to Eq. (\ref{eq:Vjik}), it is possible to define the $V_{j,p,k}$. Specifically,
	\[
		\begin{array}{llllll}
		j=1 \hspace{0.3cm}& \Rightarrow & S_{1,1}^\ast=S_{12,1} & S_{2,1}^\ast=S_{22,1} 
																		&	R_{1,1}^\ast=\left(\begin{array}{c}R_{11,1}\\\hdashline 0_{3\times 1}\end{array}\right)
																		&	R_{2,1}^\ast=\left(\begin{array}{c}R_{12,1}\\\hdashline R_{22,1}\end{array}\right)	\\
		j=2 \hspace{0.3cm}& \Rightarrow & S_{1,2}^\ast=\left(\begin{array}{c:c}S_{11,2}&0_{3\times 2}\end{array}\right)
																		& S_{2,2}^\ast=\left(\begin{array}{c:c}0_{3\times 2}\end{array}\right)
																		&	R_{1,2}^\ast=\left(\begin{array}{c}R_{11,2}\\\hdashline 0_{3\times 1}\end{array}\right)
																		&	R_{2,2}^\ast=\left(\begin{array}{c}0_{3\times 1}\\\hdashline R_{22,2}\end{array}\right)	\\
		j=3 \hspace{0.3cm}& \Rightarrow & S_{1,3}^\ast=\left(\begin{array}{c:c}I_3&0_{3\times 3}\end{array}\right)& 
																		S_{2,3}^\ast=\left(\begin{array}{c:c}0_{3\times 3}&I_{3}\end{array}\right)& &
		\end{array},
	\]	
	which leads to
	\begin{eqnarray}
		V_{1,1,2} & = & \left(\begin{array}{c}R_{1,1}^\ast \cdot S_{1,2}^\ast\end{array}\right)
								  = \left(\begin{array}{c}\left(\begin{array}{ccc} 0&1&0\\0&0&1\\1&0&0\\0&0&0\end{array}\right)
									\cdot\left(\begin{array}{cccc} 1&0&0&0\\0&1&0&0\\0&0&0&0\end{array}\right)\\\end{array}\right)
									= \left(\begin{array}{cccc} 0&1&0&0\\0&0&0&0\\1&0&0&0\\0&0&0&0\end{array}\right)\nonumber\\
		V_{1,2,2} & = & \left(\begin{array}{c}R_{2,1}^\ast \cdot S_{2,2}^\ast\end{array}\right)
									= \left(\begin{array}{c}\left(\begin{array}{ccc}0&0&0\\0&0&0\\-1&0&0\\0&0&1\end{array}\right)
									\cdot\left(\begin{array}{cccc} 0&0&1&0\\0&0&0&1\\0&0&0&0\end{array}\right)\end{array}\right)
									= \left(\begin{array}{cccc} 0&0&0&0\\0&0&0&0\\0&0&-1&0\\0&0&0&0\end{array}\right)\nonumber\\
		V_{1,1,3} & = & \left(\begin{array}{c}R_{1,1}^\ast \cdot S_{1,3}^\ast\end{array}\right)
									= \left(\begin{array}{c}\left(\begin{array}{ccc}0&1&0\\0&0&1\\1&0&0\\0&0&0\end{array}\right)
									\cdot\left(\begin{array}{c:c}I_3 & 0_{3\times 3}\end{array}\right)\end{array}\right)
									= \left(\begin{array}{cccccc} 0&1&0&0&0&0\\0&0&1&0&0&0\\1&0&0&0&0&0\\0&0&0&0&0&0\end{array}\right)\nonumber\\
		V_{1,2,3} & = & \left(\begin{array}{c}R_{2,1}^\ast \cdot S_{2,3}^\ast\end{array}\right)
									= \left(\begin{array}{c}\left(\begin{array}{ccc}0&0&0\\0&0&0\\-1&0&0\\0&0&1\end{array}\right)
									\cdot\left(\begin{array}{c:c} 0_{3\times 3} & I_3\end{array}\right)\end{array}\right)
									= \left(\begin{array}{cccccc} 0&0&0&0&0&0\\0&0&0&0&0&0\\0&0&0&-1&0&0\\0&0&0&0&0&1\end{array}\right).\nonumber
	\end{eqnarray}
	We can use Theorem \ref{theo:SuffCond} to check whether the first structural shock is identified. According to the definition of
	$V_{1,1,2}$, $V_{1,2,2}$, $V_{1,1,3}$, $V_{1,2,3}$, the matrix $V_1(\theta)$ becomes:
	\[
		V_1(\theta) = \Big(V_{1,1,2}\theta_2,\: V_{1,2,2}\theta_2,\: V_{1,1,3}\theta_3,\: V_{1,2,3}\theta_3\Big) = 
		\left(\begin{array}{cccc}
		\theta_{2,2}&\theta_{3,2}&0&0\\
		0&\theta_{3,3}&0&0\\
		\theta_{2,1}&\theta_{3,1}&-\theta_{2,3}&-\theta_{3,4}\\
		0&0&0&\theta_{3,6}
		\end{array}\right),
	\]
	that has full column rank for almost all values of $\theta$.
		
	For the identification of the second equation in the two regimes we have to start from $V_{2,1,3}$ and $V_{2,2,3}$, i.e.
	\begin{eqnarray}
		V_{2,1,3} & = & \left(\begin{array}{c}R_{1,2}^\ast \cdot S_{1,3}^\ast\end{array}\right)
									= \left(\begin{array}{c}\left(\begin{array}{ccc}0&0&1\\0&0&0\end{array}\right)
									\cdot\left(\begin{array}{cc} I_3 & 0_{3\times 3}\end{array}\right)\end{array}\right)
									= \left(\begin{array}{cccccc} 0&0&1&0&0&0\\0&0&0&0&0&0\end{array}\right)\nonumber\\
		V_{2,2,3} & = & \left(\begin{array}{c}R_{2,2}^\ast \cdot S_{2,3}^\ast\end{array}\right)
									= \left(\begin{array}{c}\left(\begin{array}{ccc}0&0&0\\0&0&1\end{array}\right)
									\cdot\left(\begin{array}{cc} 0_{3\times 3} & I_3\end{array}\right)\end{array}\right)
									= \left(\begin{array}{cccccc} 0&0&0&0&0&0\\0&0&0&0&0&1\end{array}\right).\nonumber
	\end{eqnarray}
	As before, we can use Theorem \ref{theo:SuffCond} to check whether the second structural shock is identified. According to the definition of
	$V_{2,1,3}$ and $V_{2,2,3}$, the matrix $V_2(\theta)$ becomes:
	\[
		V_2(\theta) = \Big(V_{2,1,3}\theta_3,\: V_{2,2,3}\theta_3\Big) = 
		\left(\begin{array}{cc}
		\theta_{3,3}&0\\
		0&\theta_{3,6}\\
		\end{array}\right)
	\]
	that has full column rank for almost all values of $\theta$.

	Finally, it is worth noting that the pattern of restrictions satisfies Definition \ref{def:RecJustIdent}, and thus the SVAR-WB is 
	recursively just identified.

\subsection{Bivariate SVAR-WB with short- and long-run restrictions}
\label{ex:IdentBivariate}

	Consider the bivariate SVAR-WB characterized by two restrictions across the regimes: the former imposes that the response on impact of the 
	second variable to the first shock is equal in the two regimes; the latter imposes that the long-run effect of the second variable to 
	the first shock to be the same in the two regimes.
	In other words, we have the following two restrictions:
	\[
		[IR_{1,0}]_{(2,1)}=[IR_{2,0}]_{(2,1)} \text{ and } 	[IR_{1,\infty}]_{(2,1)}=[IR_{2,\infty}]_{(2,1)}.
	\]
	The condition of Corollary \ref{corol:OrderCond} is met, suggesting to use Theorem \ref{theo:SuffCond} for studying local identification.	
	More precisely, let 
	\[
		\begin{array}{c}
		G_1\equiv G(A_{10}^\prime,A_{1+}^\prime)\equiv \left[\begin{array}{c}IR_{1,0}\\IR_{1,\infty}\end{array}\right]\equiv 
		\left[\begin{array}{c}A_{1,0}^{-1}\\(A_{10}-\sum_{i=1}^{p}A_{1i})^{-1}\end{array}\right]\\
		\\
		G_2\equiv G(A_{20}^\prime,A_{2+}^\prime)\equiv \left[\begin{array}{c}IR_{2,0}\\IR_{2,\infty}\end{array}\right]\equiv 
		\left[\begin{array}{c}A_{2,0}^{-1}\\(A_{20}-\sum_{i=1}^{p}A_{2i})^{-1}\end{array}\right].
		\end{array}
		\]
		The restrictions take the form:
		\[
		\begin{array}{ccc}
			\left(\begin{array}{cc}R_{11,1} & R_{12,1}\\ R_{21,1} & R_{22,1}\end{array}\right)\left(\begin{array}{c}G_1\\G_2\end{array}\right)e_1=0 
			& \Rightarrow &
			\left(\begin{array}{cccc:cccc}0 & 1 & 0 & 0 & 0 & -1 & 0 & 0\\ 0 & 0 & 0 & 1 & 0 & 0 & 0 & -1\end{array}\right)
			\left(\begin{array}{cc}\times & \times\\ \times & \times\\ \times & \times\\ \times & \times\\ \times & \times\\ \times & \times\\ \times 
			& \times\\ \times & \times \end{array}\right)
			\left(\begin{array}{c}1\\0\end{array}\right)=\left(\begin{array}{c}0\\0\end{array}\right)
		\end{array}
	\]
	where 
	\[
		R_{11,1}=\left(\begin{array}{cccc}0&1&0&0\\0&0&0&1\end{array}\right)=R_{1,1}^\ast \hspace{0.5cm}\text{ and }\hspace{0.5cm}
		R_{12,1}=\left(\begin{array}{cccc}0&-1&0&0\\0&0&0&-1\end{array}\right)=R_{2,1}^\ast
	\]
	while $R_{21,1}$ and $R_{22,1}$ are not defined.
	
	The explicit form of the restrictions can be written as
	\[
		\begin{array}{l}
			\left(\begin{array}{c}G_1\\G_2\end{array}\right)\left(\begin{array}{c}1\\0\end{array}\right)=\left(\begin{array}{cc}S_{11,1}&S_{12,1}\\
			S_{21,1}&S_{22,1}\end{array}\right)\theta_{1}\\
			\hspace{2cm}\Longleftrightarrow 
			\left(\begin{array}{cc}\times & \times\\ \times & \times\\ \times & \times\\ \times & \times\\ \times & \times\\ \times & \times\\ \times 
			& \times\\ \times & \times \end{array}\right)
			\left(\begin{array}{c}1\\0\end{array}\right)=
			\left(\begin{array}{cc:cccc}1&0&0&0&0&0\\0&0&0&1&0&0\\0&1&0&0&0&0\\0&0&0&0&0&1\\\hdashline 
			0&0&1&0&0&0\\0&0&0&1&0&0\\0&0&0&0&1&0\\0&0&0&0&0&1\end{array}\right)
			\left(\begin{array}{c}\theta_{1,1}\\\theta_{1,2}\\\theta_{1,3}\\\theta_{1,4}\\\theta_{1,5}\\\theta_{1,6}\end{array}\right)
		\end{array}
	\]
	where 
	\[
		S_{11,1}=\left(\begin{array}{cc}1&0\\0&0\\0&1\\0&0\end{array}\right), \hspace{0.5cm}
		S_{12,1}=\left(\begin{array}{cccc}0&0&0&0\\0&1&0&0\\0&0&0&0\\0&0&0&1\end{array}\right), \hspace{0.5cm}
		S_{21,1}=0_{4\times 2}, \hspace{0.5cm}
		S_{22,1}= I_4.
	\]
	Moreover, as the restrictions involve only the first column of $G_1$ and $G_2$, then $\textbf{R}_2$ is not defined, while 
	$\textbf{S}_2=I_8$, and, more specifically, $S_{11,2}=I_4$, $S_{12,2}=S_{21,2}=0_{4 \times 4}$ and $S_{22,2}=I_4$. As a consequence,
	$S_{1,2}^\ast=(I_4\:0_{4\times 4})$, and $S_{2,2}^\ast=(0_{4\times 4}\:I_4)$.
	
	Based on these matrices characterizing the restrictions, we can obtain
	\begin{eqnarray}
		V_{1,1,2}   & = & (R_{1,1}^\ast \cdot S_{1,2}^\ast) = \left(\begin{array}{cccc}0&1&0&0\\0&0&0&1\end{array}\right) 
										\left(\begin{array}{ccc}I_4&&0_{4\times 4}\end{array}\right)\nonumber\\
								& = & \left(\begin{array}{cccccccc}0&1&0&0&0&0&0&0\\0&0&0&1&0&0&0&0\end{array}\right)\nonumber\\
		V_{1,2,2}	  & = & (R_{2,1}^\ast \cdot S_{2,2}^\ast) = \left(\begin{array}{cccc}0&-1&0&0\\0&0&0&-1\end{array}\right) 
										\left(\begin{array}{ccc}0_{4\times 4}&&I_4\end{array}\right)\nonumber\\
						  & = & \left(\begin{array}{cccccccc}0&0&0&0&0&-1&0&0\\0&0&0&0&0&0&0&-1\end{array}\right),\nonumber
	\end{eqnarray}
	from which we derive the $V_1(\theta)$ matrix as follows:
	\[
		V_1(\theta)=\left(V_{1,1,2}\theta_2,\:V_{1,2,2}\theta_2\right)=
		\left(\begin{array}{cc}
		\theta_{2,2}&-\theta_{2,6}\\
		\theta_{2,4}&-\theta_{2,8}
		\end{array}\right)
	\]
	that has full column rank for almost all values of $\theta$.

\subsection{Non-identified Trivariate SVAR-WB}
\label{ex:NoIdentTrivariate}

	Consider the trivariate SVAR-WB characterized by the following responses on impact in the two regimes:
	\[
	A_{10}^{-1}\equiv IR_{1,0}\equiv \left(\begin{array}{ccc}c_{1,11}& 0 & 0\\0 & c_{1,22}& \encircled{c_{1,23}}\\c_{1,31}& c_{1,32}& 
	c_{1,33}\end{array}\right)
	\hspace{0.5cm}\text{ and }\hspace{0.5cm}
	A_{20}^{-1}\equiv IR_{2,0}\equiv \left(\begin{array}{ccc}c_{2,11}& 0 & 0\\c_{2,21} & c_{2,22}& \encircled{c_{2,23}}\\c_{2,31}& c_{2,32}& 
	c_{2,33}\end{array}\right)
	\]
	with the further across-regime restriction $c_{1,23}=c_{2,23}$. The SVAR-WB is characterized by a total of six restrictions, 
	satisfying the necessary condition of Theorem \ref{theo:OrdCond}.  
	
	The functions of the parameters subject to restrictions can be defined as $G_1 = A_{10}^{-1}$ and $G_2 = A_{20}^{-1}$. The condition of 
	Corollary \ref{corol:OrderCond} is not met. The SVAR-WB is not recursive and identification has to be investigated through 
	Theorem \ref{theo:NecSuff}. %or through Theorem \ref{theo:NecSuffStruct}. 
	We will see that this latter proves to be more suited for the SVAR-WB at hand. When considered jointly, the third columns of 
	$IR_{1,0}$ and $IR_{2,0}$ are the most restricted ones, thus we start our analysis from these ones. 
	Let $j=1$ be referring to the third column of $IR_{1,0}$ and $IR_{2,0}$. When $j=1$ the matrices of restrictions, 
	in implicit form, are
	\[
	\textbf{R}_1 = \left(\begin{array}{cc}R_{11,1}&R_{12,1}\\R_{21,1} & R_{22,1}\end{array}\right)  
							 = \left(\begin{array}{ccc:ccc}1&0&0&0&0&0\\0&1&0&0&-1&0\\\hdashline 0&0&0&1&0&0\end{array}\right)
	\]
	from which 
	\[
	R_{1,1}^\ast=\left(\begin{array}{ccc}1&0&0\\0&1&0\\0&0&0\end{array}\right), \hspace{0.5cm}
	R_{2,1}^\ast=\left(\begin{array}{ccc}0&0&0\\0&-1&0\\1&0&0\end{array}\right)
	\]
	and, in explicit form  
	\[
	\textbf{S}_1 = \left(\begin{array}{cc}S_{11,1} & S_{12,1}\\S_{21,1} & S_{22,1}\end{array}\right)  
							 = \left(\begin{array}{c:cc}0&0&0\\0&1&0\\1&0&0\\\hdashline 0&0&0\\0&1&0\\0&0&1\end{array}\right)
	\]
	from which
	\[
	S_{1,1}^\ast=\left(\begin{array}{ccc}0&0&0\\0&1&0\\1&0&0\end{array}\right), \hspace{0.5cm}
	S_{2,1}^\ast=\left(\begin{array}{ccc}0&0&0\\0&1&0\\0&0&1\end{array}\right).
	\]
	We then consider the second columns of $IR_{1,0}$ and $IR_{2,0}$, as the most restricted ones after the third columns. The restrictions are:
	\[
	\textbf{R}_2 = \left(\begin{array}{cc}R_{11,2} & R_{12,2}\\R_{21,2} & R_{22,2}\end{array}\right)  
							 = \left(\begin{array}{ccc:ccc}1&0&0&0&0&0\\\hdashline 0&0&0&1&0&0\end{array}\right)
	\]
	from which	
	\[
	R_{1,2}^\ast=\left(\begin{array}{ccc}1&0&0\\0&0&0\end{array}\right), \hspace{0.5cm}
	R_{2,2}^\ast=\left(\begin{array}{ccc}0&0&0\\1&0&0\end{array}\right)
	\]
	and
	\[
	\textbf{S}_2 = \left(\begin{array}{cc}S_{11,2} & S_{12,2}\\S_{21,2} & S_{22,2}\end{array}\right)  
							 = \left(\begin{array}{cc:cc}0&0&0&0\\1&0&0&0\\0&1&0&0\\\hdashline 0&0&0&0\\0&0&1&0\\0&0&0&1\end{array}\right)
	\]
	from which
	\[
	S_{1,2}^\ast=\left(\begin{array}{cccc}0&0&0&0\\1&0&0&0\\0&1&0&0\end{array}\right), \hspace{0.5cm}
	S_{2,2}^\ast=\left(\begin{array}{cccc}0&0&0&0\\0&0&1&0\\0&0&0&1\end{array}\right).
	\]
	Finally we consider the first column of $G_1$ and $G_2$. The restrictions are
	\[
	\textbf{R}_3 = \left(\begin{array}{cc}R_{11,3} & R_{12,3}\\R_{21,3} & R_{22,3}\end{array}\right)  
							 = \left(\begin{array}{ccc:ccc}0&1&0&0&0&0\end{array}\right)
	\]
	from which
	\[
	R_{1,3}^\ast=\left(\begin{array}{ccc}0&1&0\end{array}\right), \hspace{0.5cm}
	R_{2,3}^\ast=\left(\begin{array}{ccc}0&0&0\end{array}\right)
	\]
	and
	\[
	\textbf{S}_3 = \left(\begin{array}{cc}S_{11,3} & S_{12,3}\\S_{21,3} & S_{22,3}\end{array}\right)  
							 = \left(\begin{array}{cc:ccc}1&0&0&0&0\\0&0&0&0&0\\0&1&0&0&0\\\hdashline 0&0&1&0&0\\0&0&0&1&0\\0&0&0&0&1\end{array}\right)
	\]
	from which
	\[
	S_{1,3}=\left(\begin{array}{ccccc}1&0&0&0&0\\0&0&0&0&0\\0&1&0&0&0\end{array}\right), \hspace{0.5cm}
	S_{1,3}=\left(\begin{array}{ccccc}0&0&1&0&0\\0&0&0&1&0\\0&0&0&0&1\end{array}\right).
	\]
	Remembering that $V_{j,p,k}=\big(R_{p,j}^\ast\cdot S_{p,k}^\ast\big)$, we can calculate $\tilde{V}_1(\theta)$, $\tilde{V}_2(\theta)$ and
	$\tilde{V}_3(\theta)$ as follows:
	\[
	\tilde{V}_1(\theta)=\big[\begin{array}{ccc:ccc}V_{1,1,1}\theta_1&V_{1,1,2}\theta_2&V_{1,1,3}\theta_3&
																								 V_{1,2,1}\theta_1&V_{1,2,2}\theta_2&V_{1,2,3}\theta_3\end{array}\big]
	\]
	where
	\[
		\begin{array}{ll}
			V_{1,1,1}=\big(R_{1,1}^\ast S_{1,1}^\ast\big)=\left(\begin{array}{ccc}0&0&0\\0&1&0\\0&0&0\end{array}\right)&
			V_{1,1,2}=\big(R_{1,1}^\ast S_{1,2}^\ast\big)=\left(\begin{array}{cccc}0&0&0&0\\1&0&0&0\\0&0&0&0\end{array}\right)\\
			V_{1,1,3}=\big(R_{1,1}^\ast S_{1,3}^\ast\big)=\left(\begin{array}{ccccc}1&0&0&0&0\\0&0&0&0&0\\0&0&0&0&0\end{array}\right)&
			V_{1,2,1}=\big(R_{2,1}^\ast S_{2,1}^\ast\big)=\left(\begin{array}{ccc}0&0&0\\0&-1&0\\0&0&0\end{array}\right)\\
			V_{1,2,2}=\big(R_{2,1}^\ast S_{2,2}^\ast\big)=\left(\begin{array}{cccc}0&0&0&0\\0&0&-1&0\\0&0&0&0\end{array}\right)&
			V_{1,2,3}=\big(R_{2,1}^\ast S_{2,3}^\ast\big)=\left(\begin{array}{ccccc}0&0&0&0&0\\0&0&0&-1&0\\0&0&1&0&0\end{array}\right);
		\end{array}
	\]
	\[
	\tilde{V}_2(\theta)=\big[\begin{array}{ccc:ccc}V_{2,1,1}\theta_1&V_{2,1,2}\theta_2&V_{2,1,3}\theta_3&
																								 V_{2,2,1}\theta_1&V_{2,2,2}\theta_2&V_{2,2,3}\theta_3\end{array}\big]
	\]
	where
	\[
		\begin{array}{ll}
			V_{2,1,1}=\big(R_{1,2}^\ast S_{1,1}^\ast\big)=\left(\begin{array}{ccc}0&0&0\\0&0&0\end{array}\right)&
			V_{2,1,2}=\big(R_{1,2}^\ast S_{1,2}^\ast\big)=\left(\begin{array}{cccc}0&0&0&0\\0&0&0&0\end{array}\right)\\
			V_{2,1,3}=\big(R_{1,2}^\ast S_{1,3}^\ast\big)=\left(\begin{array}{ccccc}1&0&0&0&0\\0&0&0&0&0\end{array}\right)&
			V_{2,2,1}=\big(R_{2,2}^\ast S_{2,1}^\ast\big)=\left(\begin{array}{ccc}0&0&0\\0&0&0\end{array}\right)\\
			V_{2,2,2}=\big(R_{2,2}^\ast S_{2,2}^\ast\big)=\left(\begin{array}{cccc}0&0&0&0\\0&0&0&0\end{array}\right)&
			V_{2,2,3}=\big(R_{2,2}^\ast S_{2,3}^\ast\big)=\left(\begin{array}{ccccc}0&0&0&0&0\\0&0&1&0&0\end{array}\right);
		\end{array}
	\]
	\[
	\tilde{V}_3(\theta)=\big[\begin{array}{ccc:ccc}V_{3,1,1}\theta_1&V_{3,1,2}\theta_2&V_{3,1,3}\theta_3&
																								 V_{3,2,1}\theta_1&V_{3,2,2}\theta_2&V_{3,2,3}\theta_3\end{array}\big]
	\]
	where
	\[
		\begin{array}{ll}
			V_{3,1,1}=\big(R_{1,3}^\ast S_{1,1}^\ast\big)=\left(\begin{array}{ccc}0&1&0\end{array}\right)&
			V_{3,1,2}=\big(R_{1,3}^\ast S_{1,2}^\ast\big)=\left(\begin{array}{cccc}1&0&0&0\end{array}\right)\\
			V_{3,1,3}=\big(R_{1,3}^\ast S_{1,3}^\ast\big)=\left(\begin{array}{ccccc}0&0&0&0&0\end{array}\right)&
			V_{3,2,1}=\big(R_{2,3}^\ast S_{2,1}^\ast\big)=\left(\begin{array}{ccc}0&0&0\end{array}\right)\\
			V_{3,2,2}=\big(R_{2,3}^\ast S_{2,2}^\ast\big)=\left(\begin{array}{cccc}0&0&0&0\end{array}\right)&
			V_{3,2,3}=\big(R_{2,3}^\ast S_{2,3}^\ast\big)=\left(\begin{array}{ccccc}0&0&0&0&0\end{array}\right).
	\end{array}
	\]
	If we define the vector of free parameters 
	\[
	\theta=\big(\theta_1^\prime,\,\theta_2^\prime,\,\theta_2^\prime\big)^\prime=
	\big(\theta_{1,1},\,\theta_{1,2},\,\theta_{1,3},\,\theta_{2,1},\,\theta_{2,2},\,\theta_{2,3},\,\theta_{2,4},\,\theta_{3,1},\,\theta_{3,2},\,
	\theta_{3,3},\,\theta_{3,4},\,\theta_{3,5}\big)^\prime,
	\]
	we can analytically calculate $\tilde{V}(\theta)$ as in Eq. (\ref{eq:tVtheta})
	\[
		\scalebox{0.9}{
		$
		\tilde{V}(\theta)=\left(\begin{array}{c:c:c}
		\left(\begin{array}{cccccc}0&0&\theta_{3,1}&0&0&0\\
					\theta_{1,2}&\theta_{2,1}&0&\theta_{1,2}&-\theta_{2,3}&-\theta_{3,4}\\0&0&0&0&0&\theta_{3,3}\end{array}\right)&&\\
		&\left(\begin{array}{cccccc}0&0&\theta_{3,1}&0&0&0\\0&0&0&0&0&\theta_{3,3}\end{array}\right)&\\
		&&\left(\begin{array}{cccccc}\theta_{1,2}&\theta_{2,1}&0&0&0&0\end{array}\right)
		\end{array}\right).
		$}
	\]
	Moreover, given $n=3$ and $s=2$, the known matrix $\Ttt$ will take the form
	\[
		\Ttt = (\tilde{T}_{3,2}\otimes I_3)(I_2\otimes \tilde{D}_3)=\left(\begin{array}{cc}
		\tilde{D}_3^{(1)}&0\\0&\tilde{D}_3^{(1)}\\\tilde{D}_3^{(2)}&0\\0&\tilde{D}_3^{(2)}\\\tilde{D}_3^{(3)}&0\\0&\tilde{D}_3^{(3)}
		\end{array}\right),\hspace{0.2cm}\text{where}\hspace{0.2cm}\tilde{D}_3=
		\left(\begin{array}{ccc}0&0&0\\1&0&0\\0&1&0\\\hdashline -1&0&0\\0&0&0\\0&0&1\\\hdashline 0&-1&0\\0&0&-1\\0&0&0\end{array}\right)=
		\left(\begin{array}{c}\tilde{D}_3^{(1)}\\ \hdashline \tilde{D}_3^{(2)}\\\hdashline \tilde{D}_3^{(3)}\end{array}\right).
	\]
	We are now ready to calculate the matrix $\Vtt=\tVt \Ttt$, whose column rank, according to Theorem \ref{theo:NecSuff}, allows
	to check for the local identification of the SVAR-WB:
	\[
		\Vtt=\tVt\:\Ttt=\left(\begin{array}{ccc:ccc}
				0&\theta_{3,1}&0&0&0&0\\
				\theta_{2,1}&0&0&-\theta_{2,3}&-\theta_{3,4}&0\\
				0&0&0&0&\theta_{3,3}&0\\
				\hdashline
				0&0&\theta_{3,1}&0&0&0\\
				0&0&0&0&0&\theta_{3,3}\\
				\hdashline
				0&-\theta_{1,2}&-\theta_{2,1}&0&0&0
				\end{array}\right)
	\]
	From the analytical definition of $\Vtt$ it clearly emerges that the first and fourth columns are linearly dependent irrespectively of 
	the values of $\theta_{2,1}$ and $\theta_{2,3}$. The trivariate SVAR-WB, thus, does not meet the rank condition in Theorem \ref{theo:NecSuff}
	and is not locally identified. 
	
	Obtaining the analytical formulation of $\Vtt$, however, can result difficult when the number of variables or regimes is 	larger. 
	In this respect, we can use Algorithm \ref{algo:IdGeneral} and calculate numerically $\Vtt$, once generating random entries for
	$\theta$. If the rank is full for at least one $\theta$, according to Theorem \ref{theo:IdentEverywhere}, the model is locally identified
	almost everywhere in the parametric space. On the contrary, if the rank is not full, the algorithm suggests to repeat the procedure for 
	a large number of draws for $\theta$ (say, $N=10,000$). Clearly, the decision about the identification cannot be definitive but, if for 
	none of the draws the condition is met, it is reasonable to suppose that effectively the SVAR-WB be not locally identified.

%%%%%%%%%%%%%%%%%%%%%%%%%%%  APPENDIX recursively just-identified SVAR-WB %%%%%%%%%%%%%%%%%%%%%%%%%%%%%%%%%%%%%%%%%%%
\section{Two algorithms for the estimation of recursively just-identified SVAR-WBs}
\label{sec:AlgorithmSeqRec}

%\subsection{An algorithm for the estimation of recursively just-identified SVARs}
%\label{sec:AlgorithmRec}

The previous Algorithm \ref{algo:EstimGen} is very general, but can be enormously improved, in terms of efficiency and control, 
when the restrictions follow particular patterns. The next algorithm, in fact, is devoted to just-identified SVAR-WBs supporting the 
restrictions postulated by Definition \ref{def:RecJustIdent}. Essentially, given the recursive pattern of the restrictions,
the algorithm operates recursively and, for each $j\in\{1,\ldots,n\}$, calculates the \textit{j}-th column of the orthogonal matrices 
$Q_1,\ldots,Q_s$. As we allow for restrictions across the regimes, for each $j$ the calculation has to be performed jointly, i.e. we
search for admissible $q_{j,1},\ldots,q_{j,s}$ that satisfy the restrictions, are orthogonal to the previous columns of $Q_1,\ldots,Q_s$,
and are of unit length. In this respect, it is fundamental to see that Eq. (\ref{eq:GQ}) allows to write the restrictions on the 
\textit{j}-th column of all the admissible orthogonal matrices as 
\begin{equation}
\label{eq:TrasfZeroRestr}
	R_{1,j}^*G_1q_{1,j}+\cdots+R_{s,j}^*G_sq_{s,j}=0\;\Longrightarrow\;
	\Big(R_{1,j}^*G_1\;\cdots\;R_{s,j}^*G_s\Big)
	\left(\begin{array}{c}q_{1,j}\\\vdots\\q_{s,j}\end{array}\right)=0\nonumber
\end{equation}	
implying that $\Big(q_{1,j}^\prime,\,\ldots\,,q_{s,j}^\prime\Big)^\prime$ has to be orthogonal to 
$\Big(R_{1,j}^*G_1\;\cdots\;R_{s,j}^*G_s\Big)$.

\begin{algo}
\label{algo:EstimRec}
Consider a recursively just identified SVAR-WB with sign restrictions as in Eq. (\ref{eq:SignRestr_jp}) and equality restrictions as in 
Definition \ref{def:RecJustIdent}. Let $\phi=\Btot=\big[B_{1+}\,,\,\Sigma_1\,,\,\cdots\,,\,B_{s+}\,,\,\Sigma_{s}\big]$
be the reduced-form parameter point.
\begin{enumerate}
	\item Set $j=1$ and the total number of solutions of the estimation problem $m=1$. Moreover, set the counter of the number of 
		solutions $\nu=1$. Finally, set the ordering suggested by Definition \ref{def:RecJustIdent}. 
  \item 
		Based on the restrictions as in Eq. (\ref{eq:ImpForm}), define the quantity
		\begin{equation}
    \label{eq:GammajApp}
				\Gamma_j=\left(\begin{array}{ccc}
				R_{1,j}^{*}G_1&\ldots&R_{s,j}^{*}G_s\\&\tilde{Q}_1^\prime&\\&\vdots&\\&\tilde{Q}_{j-1}^\prime&\end{array}\right).
		\end{equation}
		If $j=1$, then $\Gamma_j=\big(\begin{array}{ccc}R_{1,j}^{*}G_1&\ldots&R_{s,j}^{*}G_s\end{array}\big)$.
	\item Let $\Gamma_{j\perp}$ be the space of all $(ns\times 1)$ vectors orthogonal to $\Gamma_j$, with $\mathrm{dim}\,(\Gamma_{j\perp})=s$.
		Moreover, let
		\begin{equation}
		\label{eq:GammaBase}
				\left(\begin{array}{ccc}b_{j,1}&\ldots&b_{j,s}\end{array}\right)=
				\left[\begin{array}{ccc}
					\left(\begin{array}{c}b_{j,1}^{(1)}\\\hdashline b_{j,1}^{(2)}\\\hdashline\vdots\\\hdashline b_{j,1}^{(s)}\end{array}\right)&
					\cdots&
					\left(\begin{array}{c}b_{j,s}^{(1)}\\\hdashline b_{j,s}^{(2)}\\\hdashline\vdots\\\hdashline b_{j,s}^{(s)}\end{array}\right)
					\end{array}\right]\nonumber
		\end{equation}
		be a base for $\Gamma_{j\perp}$, where each sub-vector $b_{j,k}^{(p)}$ is of dimension $(n\times 1)$.
		Then, for $p\in\{1,\ldots,s\}$, form the $(n\times s)$ matrix $b_j^{(p)}=\big(b_{j,1}^{(p)},\,\ldots,\,b_{j,s}^{(p)}\big)$ and,
		subsequently, the $(s\times s)$ matrix $\tilde{B}_j^{(p)}=b_j^{(p)\prime}\,b_j^{(p)}$.
		Solve the system of equations
		\begin{equation}
		\label{eq:sysAlgoRec}
				\left\{\begin{array}{rcl}\tilde{\lambda}_j^\prime\tilde{B}_j^{(1)}\tilde{\lambda}_j&=&1\\
				\vdots&&\\\tilde{\lambda}_j^\prime\tilde{B}_j^{(s)}\tilde{\lambda}_j&=&1
				\end{array}\right.
		\end{equation}
		for the $(s\times 1)$ vector of unknowns $\tilde{\lambda}_j=\left(\lambda_1\,\ldots,\,\lambda_s\right)^\prime$ and collect all the 
		admissible solutions $\left(\tilde{\lambda}_j^{(1)}\,\ldots,\,\tilde{\lambda}_j^{(k_{j,\nu})}\right)$, where $k_{j,\nu}$ is the number
		of real solutions. 
		For $i\in\{1,\ldots,k_{j,\nu}\}$, calculate the $\big(q_{j,1}^{(i)\prime},\,\ldots\,,q_{j,s}^{(i)\prime}\big)^\prime$ vector as
		\[
			\left(\begin{array}{c}q_{j,1}^{(i)}\\\vdots\\q_{j,s}^{(i)}\end{array}\right)=
			\left(\begin{array}{ccc}b_{j,1}&\ldots&b_{j,s}\end{array}\right)
			\tilde{\lambda}_j^{(i)},
		\]
		and retain only those satisfying the sign restrictions; update the number of admissible solutions $k_{j,\nu}$.
	\item Given $\nu$, if $k_{j,\nu}\neq 0$, for $i\in\{1,\ldots,k_{j,\nu}\}$, update the admissible partial orthogonal matrices
		\[
			Q_{j,p}^{(\tau)}=\big[Q_{j-1,p}^{(\nu)}\,|\,q_{j,p}^{(i)}\big]
		\]
		with $p\in\{1\,\ldots,s\}$, and where $\tau$ is a progressive number going from $(k_{j,1}+\cdots+k_{j,\nu-1}+1)$ to 
		$(k_{j,1}+\cdots+k_{j,\nu})$, denoting the updated number of admissible solutions for the first $j$ columns of the orthogonal matrices 
		$Q_1,\ldots,Q_s$, for the set of solutions associated to $\nu$. If, instead, $k_{j,\nu}=0$, then delete all the solutions 
		$Q_{j-1,1}^{(\nu)},\ldots,Q_{j-1,s}^{(\nu)}$, i.e. all the solutions associated to $\nu$.
	\item Update the number of admissible solutions for all $\nu\in\{1,\ldots,m\}$, i.e. $m=k_{j,1}+\cdots+k_{j,m}$. 
		If $m=0$, then no admissible solutions and STOP. Otherwise, if $m\neq 0$ and $j=n$, then STOP. If, instead, $m\neq 0$ but $j<n$, update
		$j=j+1$ and repeat Steps 2-4 for $\nu\in\{1\ldots,m\}$, where, in turn, based on the obtained $Q_{j,p}^{(1)},\ldots,Q_{j,p}^{(m)}$,
		define the matrices
		\begin{equation}
		\label{eq:QtAlgo}
			\tilde{Q}_1=\left(\begin{array}{cccc}
			Q_{1,1}^{(\nu)}&0&\cdots&0\\
			0&Q_{2,1}^{(\nu)}&\cdots&0\\
			\vdots&\vdots&\ddots&\vdots\\
			0&0&\cdots&Q_{s,1}^{(\nu)}
			\end{array}\right)
			\hspace{0.3cm}\cdots\hspace{0.3cm}
			\tilde{Q}_{j-1}=\left(\begin{array}{cccc}
			Q_{1,j-1}^{(\nu)}&0&\cdots&0\\
			0&Q_{2,j-1}^{(\nu)}&\cdots&0\\
			\vdots&\vdots&\ddots&\vdots\\
			0&0&\cdots&Q_{s,j-1}^{(\nu)}
			\end{array}\right).
		\end{equation}
\end{enumerate}
\end{algo}

\vspace{0.3cm}
In the previous algorithm, the general matrix $Q_{p,j}^{(\tau)}$, for each regime $p$, contains the already defined 
$[q_{p,1}^{(i)},\ldots,q_{p,j-1}^{(i)}]$ orthogonal vectors; the counter $1\leq j\leq n$ indicates the column of $Q_p$; 
the index $m$ counts the number of admissible $Q_p$ matrices; finally, for each column $j$, the counter $1\leq i\leq m$ updates 
the $Q_{p,j}$ matrices by adding the new admissible column(s). The matrices $\tilde{Q}_1,\ldots,\tilde{Q}_{j-1}$ allow to find
a new vector containing $(q_{j,1}^\prime,\ldots,q_{j,s}^\prime)^\prime$ that is orthogonal to the restrictions and to the 
already obtained $j-1$ columns in the $Q_1,\ldots,Q_s$ matrices. Given the reduced form parameters $\phi=\Btot$, Algorithm 
\ref{algo:EstimRec} produces the set of all admissible $Q\in \Qr$. From this set of admissible matrices, depending on the 
function $\eta(\phi,Q)$ of the structural parameters we are interested in, it is thus possible to easily obtain the identified 
set $\IS$.

The most delicate point of the algorithm is finding all the solutions of the system of quadratic equations 
(\ref{eq:sysAlgoRec}) in (Step 3), made of $s$ quadratic equations in $s$ unknowns. However, given that in empirical applications 
the number of regimes is relatively small, solving the system represents a simple problem and can be easy implemented in modern 
softwares.\footnote{As before, the Matlab commands \texttt{vpasolve} and \texttt{solve} represent simple and powerful solutions.} 

The way the algorithm works is substantially based on Algorithm 2 in \cite{BK19}, which provides all the admissible orthogonal 
matrices for locally- but not globally-identified SVARs. Our Algorithm \ref{algo:EstimRec} is an extension in that it accounts 
for the presence of possibly related $s$ regimes, where the relationships among the regimes is intended as possible restrictions 
involving the parameters (or functions of them) of different regimes. 

In the end, the number of possible solutions, the locally-identified parameters, can be rather large. In this vein, it can be very 
useful to introduce plausible and largely recognized sign restrictions, that can significantly reduce the number of observationally 
equivalent admissible structural parameters.

Finally, if the condition for local identification is met, as well as the pattern in Definition \ref{def:RecJustIdent}, but there are 
no restrictions across regimes, then the SVAR-WB will be globally identified too, and Algorithm \ref{algo:EstimRec} will become a
generalization of Algorithm 1 in RWZ. All solutions of Eq. (\ref{eq:sysAlgoRec}) will be real and the two unit vectors
$q_{p,j,1}$ and $q_{p,j,2}$ will lie on the same straight line, but one opposite to the other. The sign normalization will discard
one of the two and, for each regime, only one admissible solution will be retained.

\vspace{0.5cm}
In this second part of the appendix we report an algorithm to estimate a SVAR-WB with very specific restrictions. In particular other 
than being recursively just identified according to Definition \ref{def:RecJustIdent}, it also presents, for each $j\in\{1,\ldots,n\}$, 
a specific ordering of the regimes such that, for the \textit{p}-th regime, the structural parameters for the previous ones have been 
already estimated, and thus, considered as known. We introduce, thus, the following formal definition:

\begin{defin}[Sequential and recursive just identification]
  \label{def:SeqRecJustIdent}
	A recursive just identified SVAR-WB is also said sequentially and recursively just identified if the number of restrictions is 
	$f_{p,j}=n-j$, with $j\in\{1,\ldots,n\}$ and $p\in\{1,\ldots,s\}$, and, moreover, 
	for any $j\in\{1,\ldots, n\}$, it is possible to find an ordering of the regimes $k_{1,j},\ldots,k_{s,j}$ such that:
	\begin{itemize}
		\item[a.] the \textit{j}-th equation in the $k_{1,j}$-th regime is identified (locally or globally);
		\item[b.] the \textit{j}-th equation in the $k_{p,j}$-th regime, $k_{p,j}\in\{k_{2,j},\ldots,k_{s,j}\}$, is either not subject to 
		restrictions across regimes or depends on the parameters of the \textit{j}-th equation in regimes $k_{1,j},\ldots,k_{p-1,j}$.
	\end{itemize}
\end{defin}
\vspace{0.5cm}

The following matrices will play a relevant role in the algorithm:
\begin{equation}
\label{eq:FiiTilde}
	\tilde{F}_{pp,j}\equiv\left(\begin{array}{c}R_{pp,j}G_p\\q_{p,1}^\prime\\\vdots\\q_{p,j-1}^\prime\end{array}\right)
	\hspace{0.5cm}\text{ and }\hspace{0.5cm}
	\tilde{c}_{p,j}\equiv\left(\begin{array}{c}-R_{p1,j}G_1q_{1,j}-\cdots-R_{p(p-1),j}G_{p-1}q_{p-1,j}\\0\\\vdots\\0\end{array}\right)
\end{equation}
where $p=\{1,\ldots,s\}$ is the index of the regimes, while $j=\{1,\ldots,n\}$ indicates the column of the function $G(\cdot)$. 
For simplifying the notation, as before $G_p$ denotes the quantity $G_p=G_p(\Cholpit,B_p^\prime)$. 
Moreover, when $j=1$ we define $\tilde{F}_{pp,1}\equiv R_{pp,1}G_1$, and when $p=1$, we define $\tilde{c}_{p,j}\equiv(0)$.
Furthermore, for $p=\{1,\ldots,s\}$ and $j=\{1,\ldots,n\}$, let 
\begin{eqnarray}
	\tilde{A}_{p,j} & = & \tilde{F}_{pp,j}^\prime\bigg(\tilde{F}_{pp,j}\tilde{F}_{pp,j}^\prime\bigg)^{-1}\tilde{c}_{p,j}\label{eq:Ai}\\
	\tilde{B}_{p,j} & = & \bigg(I_n-\tilde{F}_{pp,j}^\prime\bigg(\tilde{F}_{pp,j}\tilde{F}_{pp,j}^\prime\bigg)^{-1}\tilde{F}_{pp,j}\bigg),
	\label{eq:Bi}
\end{eqnarray}
that are functions of the restrictions and the reduced-form parameters. Finally, let $\tilde{\alpha}_{p,j}$ be any bases for 
$\text{span}(\tilde{B}_{p,j})$ that, as proved in \cite{BK19}, given the particular recursive pattern of the restrictions, has 
dimension equal to 1.

\begin{algo}
\label{algo:EstimSeqRec}
Consider a just identified SVAR-WB with sign restrictions as in Eq. (\ref{eq:SignRestr_jp}) and equality restrictions as in Definition 
\ref{def:SeqRecJustIdent}. Let $\phi=\Btot=\big[B_{1+}\,,\,\Sigma_1\,,\,\cdots\,,\,B_{s+}\,,\,\Sigma_{s}\big]$
be the reduced-form parameter point.
\begin{enumerate}
	\item Set the counter for the columns of the orthogonal matrices $j=1$, the counter for the number of solutions of the estimation 
	problem $i=1$, and the total number of solutions of the estimation problem $m=1$, and check for the ordering suggested by 
	Definition \ref{def:RecJustIdent}. Let the correct ordering be $k_{1,j}\leq k_{2,j}\leq \ldots\leq k_{s,j}$ and fix the quantity 
	$\tilde{Q}_{p,j}^{(i)}=\{\:\:\}$, $\forall \;p\,\in\,\{k_{1,j},k_{2,j},\ldots,k_{s,j}\}$;
	\item Sequentially, for each $p\,\in\,\{k_{1,j},k_{2,j},\ldots,k_{s,j}\}$, form the matrices $\tilde{F}_{pp,j}^{(i)}$ and 
	$\tilde{c}_{p,j}^{(i)}$ as in Eq. (\ref{eq:FiiTilde}), calculate $A_{p,j}^{(i)}$, $B_{p,j}^{(i)}$ as in Eq.s (\ref{eq:Ai})-(\ref{eq:Bi}) 
	and define $\alpha_{p,j}^{(j)}$ any bases for the $\text{span}(\tilde{B}_{p,j})$. Then solve the quadratic equation in the scalar 
	unknown $z$:
  \begin{equation}
     \label{eq:QuadEqNHAlgo}
     \tilde{A}_{p,j}^{(i)\prime}\tilde{A}_{p,j}^{(i)}+2\tilde{A}_{p,j}^{(i)\prime}\tilde{\alpha}_{p,j}^{(i)}\:z+
		 \tilde{\alpha}_{p,j}^{(i)\prime}\tilde{\alpha}_{p,j}^{(i)}\:z^2-1\,=\,0.
  \end{equation}
	If both solutions $z_1$ and $z_2$ are not real, then remove $\tilde{Q}_{p,j}^{(i)}$ and
				\begin{enumerate}
                \item[a.] if $i<m$, then $i=i+1$ and go to Step 2;
								\item[b.] if $i=m$, then go to Step 3.
				\end{enumerate}
	If the solutions $z_1$ and $z_2$ are real, then
        \begin{enumerate}
                \item[a.] calculate the two vectors $q_{p,j,1}=\tilde{A}_{p,j}^{(i)}+\tilde{\alpha}_{p,j}^{(i)}z_1$ and 
													$q_{p,j,2}=\tilde{A}_{p,j}^{(i)}+\tilde{\alpha}_{p,j}^{(i)}z_2$;
                \item[b.] define $k$ as the number of vectors satisfying the normalization and sign restrictions;
                \item[c.] if $k=0$, then remove $\tilde{Q}_{p,j}^{(i)}$; otherwise, for $h=1,\ldots,k$, 
													update the $\tilde{Q}_{p,j}^{(i)}$ matrix as
									\begin{equation}
										\tilde{Q}_{p,j}^{(i)}=[Q_{p,j-1}^{(i)}\,|\,q_{p,j,h}];
										\label{eq:QijMat}\nonumber
									\end{equation}
								\item[d.] if $i<m$, then $i=i+1$ and go to Step 2; otherwise, if $i=m$, then go to Step 3.									
        \end{enumerate}
\item Update $m$ by counting the number of $\tilde{Q}_i$ matrices; if $m=0$, then relax some sign restriction or change the 
			equality restrictions and go to Step 1; otherwise, if $j<n$, set $j=j+1$, $i=1$ and go to Step 2; else, if $j=n$
			set $M$ as the number of admissible solutions and STOP.
\end{enumerate}
\end{algo}

In the previous algorithm, the general matrix $\tilde{Q}_{p,j}^{(i)}$, for each regime $p$, contains the already defined 
$[q_{p,1}^{(i)},\ldots,q_{p,j-1}^{(i)}]$ orthogonal vectors; the counter $1\leq j\leq n$ indicates the column of $Q_p$; 
the index $m$ counts the number of admissible $Q_p$ matrices; finally, for each column $j$, the counter $1\leq i\leq m$ updates 
the $\tilde{Q}_{p,j}$ matrices by adding the new admissible column(s).
Given the reduced form parameters $\phi=\Btot$, Algorithm \ref{algo:EstimSeqRec} produces the set of all admissible $Q\in \Qr$. 
From this set of admissible matrices, depending on the function $\eta(\phi,Q)$ of the structural parameters we are interested in, 
it is thus possible to easily obtain the identified set $\IS$.

The way the algorithm works is substantially based on Algorithm 2 in \cite{BK19}, which provides all the admissible orthogonal 
matrices for locally- but not globally-identified SVARs. Our Algorithm \ref{algo:EstimSeqRec} is an extension in that it accounts 
for the presence of possibly related $s$ regimes, where the relationships among the regimes is intended as possible restrictions 
involving the parameters (or functions of them) of different regimes. 

Given the particular pattern of the restrictions it is possible to attack the problem recursively, equation-by-equation first 
(or more precisely, column-by-column of $Q_p$), and regime-by-regime for each specific equation (or column). The estimation problem, 
thus, reduces to solving a simple scalar quadratic equation, that, depending on the reduced-form parameters, can have no real solutions, 
two identical solutions, and two distinct ones. If the solutions are not real, they are simply discarded. If they are real, 
then we check whether they satisfy the normalization and sign restrictions. Those passing this further check will be added to all 
the existing $\tilde{Q}_{p,j}^{(i)}$ matrices. If there are two admissible real solutions, the total number $m$ of admissible 
(partial) matrices will be doubled, one for the first solution, and one for the second.

\end{document}